\documentclass[11pt]{amsart}

\usepackage{amsfonts,amssymb,graphicx,natbib}
\usepackage[colorlinks=true, allcolors=blue]{hyperref}
\usepackage[a4paper,top=2.5cm,bottom=2.5cm,left=2.5cm,right=2.5cm,marginparwidth=1.75cm]{geometry}

\theoremstyle{plain}

\newtheorem{theorem}{Theorem}[section]
\newtheorem{lemma}[theorem]{Lemma}
\newtheorem{proposition}[theorem]{Proposition}
\newtheorem{corollary}[theorem]{Corollary}
\theoremstyle{definition}
\newtheorem{definition}[theorem]{Definition}
\newtheorem*{example}{Example}

\newtheorem*{remark}{Remark}
\newtheorem{assumption}{Assumption}
\newtheorem{fakealgoenv}{Algorithm}
\usepackage{tikz}
\usetikzlibrary{positioning, shapes.geometric, arrows.meta, calc, decorations.pathreplacing}
\usepackage{mathtools}
\usepackage[plain,noend]{algorithm2e}
\usepackage{placeins}
\usepackage{dsfont}

\newcommand{\dd}{\mathrm{d}}  %

\DeclarePairedDelimiter\abs{\lvert}{\rvert}

\DeclarePairedDelimiter\paren{(}{)}
\DeclarePairedDelimiter\brac{[}{]}
\DeclarePairedDelimiter\curly{\{}{\}}
\DeclarePairedDelimiter\norm{{\|}}{{\|}}
\DeclarePairedDelimiter{\ceil}{\lceil}{\rceil}

\newcommand{\Exp}{\mathbb E}
\newcommand{\Prob}{\mathbb P}

\DeclareMathOperator{\Var}{Var}

\DeclareMathOperator{\Corr}{Corr}

\newcommand{\bbR}{\mathbb{R}}  %
\newcommand{\setC}{\mathcal C}

\newcommand{\ess}{\textsc{ess}}
\newcommand{\kl}{\operatorname{KL}}
\newcommand{\tv}{\textsc{TV}}

\newcommand{\PG}{\operatorname{PG}}
\newcommand{\proj}{\textsc{proj}}
\newcommand{\Uniform}{\mathcal{U}}

\newcounter{lastappendix} %
\begin{document}

\title[General coupling-based $f$-divergence diagnostics]{General coupling-based \(f\)-divergence diagnostics for Markov chain Monte Carlo}
\author{Adrien Corenflos}
\address{Department of Statistics, University of Warwick, Coventry, United Kingdom}
\email{adrien.corenflos@warwick.ac.uk}
\author{Hai-Dang Dau}
\address{College of Computing and Data Science, Nanyang Technological University, Singapore}
\email{haidang.dau@ntu.edu.sg}
\begin{abstract}
A long-standing gap exists between the theoretical analysis of Markov chain Monte Carlo convergence, which is often based on statistical divergences, and the diagnostics used in practice.
We introduce the first general convergence diagnostics for Markov chain Monte Carlo based on a broad class of $f$-divergence, allowing users to directly monitor, among others, the Kullback--Leibler and the chi-squared divergences as well as the total variation distance.
Our approach rests on a coupling-based `weight harmonization' scheme that produces direct, computable, and consistent importance weights for interacting Markov chains with respect to their target distribution.
Beyond their use as convergence diagnostics, these weights give unbiased estimates of the Radon--Nikodym derivative between the un-normalized target and the distribution of the current iterate, a richer object than the convergence bounds alone, with natural applications to importance-weighted inference.
We show how such weightings can provide asymptotic upper bounds to any $f$-divergence, prove that these bounds converge to zero as the chains approach stationarity, and demonstrate that, while more conservative than existing coupling-based total variation estimators, our method remains a practical and broadly applicable diagnostic tool.
\end{abstract}
\maketitle

\section{Introduction}\label{sec:introduction}
\subsection{Markov chain Monte Carlo and convergence diagnostics}\label{subsec:mcmc}
Computing the expectation $\pi(\varphi) = \Exp_{\pi}(\varphi)$ of a test function $\varphi$ with respect to a distribution $\pi$ of interest is routinely achieved by means of Markov chain Monte Carlo~\citep[MCMC, see, e.g.,][]{brooks2011handbook}.
In most applications, only an un-normalized measure $\gamma$ admits a tractable density, whereas the normalizing constant $Z$ such that $\pi = \gamma/Z$ is unknown.
The class of MCMC methods generates samples $X_{t+1} \sim K(X_t, \cdot)$ under a Markov kernel $K$ designed to keep $\pi$ invariant: $(\pi K)(\dd y) = \int K(x, \dd y) \pi(\dd x) = \pi(\dd y)$.
Expectations under $\pi$ are then often computed in one of two compatible ways: either by averaging over the iterations generated by the Markov chain, $\pi(\varphi) \approx \sum_{t = B}^T \varphi(X_t) / (T - B + 1)$; or by combining $N > 1 $ such independent estimates $\hat{\varphi}^n = \sum_{t = B}^T \varphi(X^n_t) / (T - B + 1)$ together as $\pi(\varphi) \approx \sum_{n=1}^N \hat{\varphi}^n / N$.
Here $B > 0$ is a burn-in period, used to discard the initial samples of the Markov chain, which are often not representative of the target distribution $\pi$.
Under technical conditions on the initial distribution of $X_0 \sim \mu_0$, the kernel, and the target, the marginal distribution $\mu_t$ of $X_t$ converges to $\pi$ as $t \to \infty$.
Such results are typically obtained in terms of the rate of decay of the total variation distance $\norm{\mu_t - \pi}_{\tv}$ between the two distributions, or of the Wasserstein distance~\citep[see, e.g.,][]{meyn2012markov,douc2018mc}.
However, despite their elegant theoretical foundations, these results are hardly amenable to practical use, as they often make use of non-accessible properties of the target distribution $\pi$.

Diagnostic tools for MCMC have therefore historically drifted away from the theoretical metrics highlighted above, focusing instead on more accessible and interpretable quantities that can provide guidance for practitioners.
Central among these is the Gelman--Rubin diagnostic~\citep{Gelman1992diagnostics}, which, from $N$ independent simulations $X_{0:T}^{1:N}$, tests whether all chains agree about their estimate $\sum_{t=0}^T \varphi(X^n_t) / (T+1)$ of $\pi(\varphi)$. 
This method, refined over the years~\citep{vats2021,vehtari2021rank} has been very successful, in part due to its integration into software packages, its ease of use in practice, and the fact that it provides an estimator of the \emph{effective sample size} of a test statistic: 
\begin{equation}\label{eq:ess}
  \ess = \frac{T+1}{1 + 2 \sum_{t=1}^T \Corr\curly*{\varphi(X_0), \varphi(X_{t})}}.
\end{equation}
This effective sample size is interpreted as the number of independent samples from $\pi$ that would yield the same variance as exhibited by the Markov chain $X_{0:T}$ for $\varphi$.
Nonetheless, it is not without limitations: it directly targets a specific aspect of the convergence of the Markov chain, namely the agreement in terms of a specific functional, rather than a more fundamental property of the chain itself.
While this may serve as a surrogate (for example when $\varphi$ is a collection of moments), it does not provide a complete picture of the convergence behaviour of the chain.

Complementary to Gelman--Rubin-style diagnostics, few but important monitoring tools have been proposed that directly assess the convergence of the Markov chain to stationarity in terms of a given metric, always either the total variation or the Wasserstein distance as far as we are aware, or in terms of Stein's discrepancies~\citep{gorham2015measuring,gorham2017measuring,srinivasan2025the}.
When considering the total variation or Wasserstein metrics, the idea of diagnosing convergence using coupled Markov chains has been recognized early~\citep{Johnson1996convergence} and has played an important role.
A notable example is \citet{biswas2019estimating} who obtain generally applicable and computable (via Monte Carlo) upper bounds to both total variation and Wasserstein distances.
Their method is based on a debiasing technique using coupled Markov chains, originating from \citet{glynn2014exact,jacob2020unbiasedmcmc}.
For instance, for the total variation, $\norm{\mu_t - \pi}_{\tv}$ can be upper bounded by $\Exp\brac*{\ceil{({\tau} - L - t) / L}_+}$, where $\tau = \min\left(t > L \mid X_t = Y_{t-L}\right)$ and $(X_t)_{t \geq 0}$ and $(Y_t)_{t \geq 0}$ are marginally distributed according to the same Markov chain.
Implementing such a method in practice often amounts to designing a coupling between the two lagged chains, $\bar{K}(x_t, y_{t-L}, \dd x_{t+1}, \dd y_{t+1-L})$ under which $\tau$ is almost surely finite.
Many Markov kernel couplings have since been proposed in the literature, covering a range of popular kernels~ \citep{heng2019couplings,jacob2020unbiasedsmoothing,lee2020coupled,wang2021coupling,biswas2022coupling,ceriani2024linearcostunbiasedposteriorestimates,Papp2024couplings,corenflos2025pdmp}, which we shall review in Supplement~\ref{app:couplings}.
While it can simply be taken equal to $1$, the lag $L > 0$ is a hyper-parameter controlling the tightness of the bound $\Exp\brac*{\ceil{({\tau} - L - t) / L}_+}$, and ideally chosen such that $X_L$ is approximately at stationarity.
Calibrating it must therefore be implemented as an adaptive offline procedure, first using some pre-fixed choice for $L$ (e.g., $L=1$), and updating $L$ as to eventually make the bounds non-vacuous~\citep[Section 2.3]{biswas2019estimating}.
Another practical choice~\citep[Section 3]{atchade_unbiased} is to choose $L$ adaptively as a statistic of the meeting time $\tau$, for instance the 95-th percentile of meeting times thus observed for a previous choice of $L$. 
Both choices give implementable and practical ways to upper bound the total variation and Wasserstein distances of the Markov chain to the target.

Still, controlling these two distances does not provide a mechanism to correct for the discrepancy between $\mu_t$ and $\pi$, and thus on their own cannot be exploited further than for diagnostics.
More specifically, while generating an infinitely large number of draws from $\mu_t$ allows one to understand that distribution perfectly, the information that $\norm{\mu_t - \pi}_{\tv} < \varepsilon$ alone does not allow us to translate perfect knowledge of $\mu_t$ to perfect knowledge of $\pi$.
In other terms, while the method of \citet{biswas2019estimating} is useful as an alert tool providing warnings when the Markov chain does not converge, or as a comparison tool where different kernels can be pitched together, it does not fully confer to the user the information on $\pi$ when the chain does converge.

On the other hand, if we knew the Radon--Nikodym derivative $\dd \pi/\dd \mu_t$ (which in most cases of interest is equivalent to the ratio of their densities), we could obtain consistent estimates of $\pi$ from a $\mu_t$-distributed sample using the celebrated \textit{importance sampling} procedure~\citep[dating at least from][]{Kahn1951SplittingParticleTransmission}.
In addition, we could estimate any statistical divergences of $\pi$ with respect to $\mu_t$ provided that they are defined solely in terms of the Radon--Nikodym derivative. 
This class of divergences, called $f$-divergences~\citep{renyi1961measures}, includes popular objectives such as the Kullback--Leibler divergence ($\kl$), the total variation distance ($\tv$), and the Rényi divergences (see Section~\ref{subsec:divergences} for more details). Unfortunately the exact ratio $\dd \pi/\dd \mu_t$ is intractable in all but the most trivial cases.

\subsection{Contributions}\label{subsec:contrib}
In this article, we propose a novel way to run parallel MCMC chains which gives a consistently weighted approximation of $\pi$ at each time step. 
These can be interpreted as noisy approximations of the un-normalized Radon--Nikodym derivative $\dd \gamma/\dd \mu_t$.
We leverage this construction to introduce upper bounds to the $f$-divergences of the target with respect to the current sample distribution.
Contrary to existing diagnostics, our method is fully online and valid from step one, requiring no lag or warm-up.
Our construction is based on Markov kernel coupling techniques arising from the MCMC~\citep{glynn2014exact,jacob2020unbiasedmcmc} literature and a weighting structure arising from sequential Monte Carlo literature~\citep[see e.g.][]{chopin2020book}, in particular \citet[Section 4.2]{dau2023complexity} using two interacting copies of $N$ independent chains.

The article is organized as follows.
Section~\ref{sec:importance-weighting} explains how any weighted approximation of the target can be used as a diagnostic and highlights the challenge of finding a good one in the context of MCMC. 
Section~\ref{sec:weight-harmonization} details our method, termed weight harmonization, and explains how it can be used both to diagnose the convergence of MCMC and to produce a consistent estimate of expectations under the target for any finite $t$.
Section~\ref{sec:theory} provides consistency and unbiasedness results for the resulting importance weighted MCMC chains.
Additionally, under strong mixing assumptions and a coupling assumption akin to uniform ergodicity, we prove that the system of weights converges exponentially fast to uniform weighting of all trajectories, recovering a computational version of exponential ergodicity theorems and validating the use of our method as a diagnostic.
Section~\ref{sec:exps} illustrates the empirical behaviour of our method and shows that it returns practical albeit conservative estimators of the convergence of the system.\footnote{The code used to generate the results can be found at \url{https://github.com/AdrienCorenflos/ImportanceMCMC}.}
Section~\ref{sec:discussion} concludes with a discussion of different advantages and drawbacks of our method. 
In particular, we identify two likely points of improvement: increasing the interaction between particles by Rao--Blackwellization, and offline correction of the estimators, in a fashion similar to particle smoothing operations~\citep[see, e.g.][Chapter 12]{chopin2020book}.
\subsection{Notations and assumptions}\label{subsec:notations}
We work on Polish spaces.
For two measurable spaces $(\mathcal X, \mathcal B_X)$ and $(\mathcal Y, \mathcal B_Y)$; a kernel $V: \mathcal X \times \mathcal B_Y \to [0,1]$; and a measurable function $f: \mathcal Y \to \bbR$, we denote by $Vf$ the function $Vf: \mathcal X \to \bbR$ defined by $(Vf)(x) = \int_{\mathcal Y} f(y) V(x, \dd y)$.
For any measure $\mu$ on $\mathcal X$, we write $\mu V$ for the measure $(\mu V)(B) = \int_{\mathcal X} V(x, B) \mu(\dd x), \forall B \in \mathcal B_Y$.
We also write $\mu \times V$ for the joint measure $\mu(\dd x) V(x, \dd y)$.

With respect to a dominating measure on $\mathcal X$, let $\pi(x) = \gamma(x) / Z$ be a target density known up to its normalizing constant $Z = \int \gamma(x) \dd x$. 
We consider $\pi$-invariant (and therefore $\gamma$-invariant) Markov kernels $K(x, \cdot)$: $\pi K = \pi$. 
Couplings over $K$ are defined as joint constructions $X', Y' \sim \bar{K}(x, y, \cdot, \cdot)$, which marginally verify $(X' \mid x) \sim K(x, \cdot)$ and $(Y' \mid y) \sim K(y, \cdot)$.
We write $K^t(x, \cdot) = \int \int \cdots \int K(x_{t}, \cdot) K(x_{t-1}, \dd x_t) \cdots K(x, \dd x_1)$ for $K$ applied $t$ times, with $K^1 = K$ and $K^0 = \delta_x(\dd y)$ being the identity kernel.

We use subscript for time index and superscript for sample index: $X^n_t$ is the particle $n$ evolved until time $t$. 
Collections over both indices are written as $X_{0:t}^{1:N}$, the set of all $N$ `trajectories' $(X^n_{0:t})_{n=1}^N$ over the time interval $(0:t) = 0, \ldots, t$.

In this article, Markov chains are started using samples from a tractable distribution $\mu_0$, and we write $\mu_t = \mu_0 K^t$ for the marginal distribution of the chain at time $t$.
Gaussian distributions are denoted as $\mathcal{N}(\mu, \Sigma)$, where $\mu$ and $\Sigma$ are the mean and covariance of the distribution, respectively. 
We do not distinguish between the multivariate and the univariate cases.
The vector of $n$ ones is denoted $1_n$, the vector of $n$ zeroes $0_n$, and we write $I_n$ for the identity matrix of size $n$.

\section{Importance weighting and Markov chain}\label{sec:importance-weighting}
\subsection{\texorpdfstring{$f$}{f}-divergences}\label{subsec:divergences}
We first recall the definition of $f$-divergences and apply it to discrete representations.
\begin{definition}
	\label{def:f_divergence}
	Let $\mu$ and $\pi$ be two mutually absolutely continuous probability distributions on $\mathcal X$. Let $f:]0, \infty[ \to ]-\infty, \infty[$ be a convex function such that $f(1) = 0$. The $f$-divergence of $\pi$ with respect to $\mu$ is defined by
\begin{equation}
D_f(\pi\|\mu) = \int \mu(\dd x) f\paren*{\frac{\dd \pi}{\dd \mu}(x)}.
\end{equation}
\end{definition}
\begin{example}\label{ex:divergences}
    Typical examples of $f$-divergences include the alpha divergences for $f_\alpha(t)=\brac{t^\alpha-\alpha t+\alpha-1}/\brac{\alpha(\alpha-1)}$, $\forall \alpha \in \mathbb R \setminus \{0, 1\}$, including the chi-squared divergence corresponding to $\alpha = 2$ up to a constant, the total variation distance for $f(t) = \abs{t-1} / 2$, the Kullback--Leibler divergence for $f(t) = t\log t$, the reverse Kullback--Leibler for $f(t) = -\log t$, and the squared Hellinger distance for $f(t) = (\sqrt{t} - 1)^2 / 2$.
\end{example}

The following immediate lemma applies Definition~\ref{def:f_divergence} to the case of discrete measures on $\{1, \ldots, M \}$.
\begin{lemma}\label{lem:empirical-div}
	Let $W^1, \ldots, W^M$ be strictly positive such that $\sum_{m=1}^M W^m = 1$. 
    Then
	\begin{equation} 
    \label{eq:f-div-emp-meas}
	D_f\paren*{\sum_{m=1}^M W^m \delta_{m} \|
		\sum_{m=1}^M \frac 1M \delta_{m}}
	= \frac 1M \sum_{m=1}^M f(MW^m).
	\end{equation}
\end{lemma}
This lemma is particularly appealing in the scenario where $X^1, \ldots, X^M$ are independent draws from $\mu$ and the weights $W^m$ are proportional to $\dd \pi/\dd\mu (X^m)$.
Then we expect~\eqref{eq:f-div-emp-meas} to converge to $D_f(\pi\|\mu)$ in probability under appropriate conditions.

We shall now investigate the more general situation in which $X^1, \ldots, X^M$ are not necessarily independent and the weights $W^m$ are not necessarily proportional to the Radon-Nikodym derivative $\dd\pi/\dd \mu$ evaluated at $X^m$.

\subsection{Weighted approximations as diagnostics}
\label{subsect:weighted_as_diag}
A central object in our methodology is the class of weighted approximations of the target distribution.
Let $X^1, \ldots, X^M$ be $M$ (not necessarily independent) identically distributed draws from a distribution $\mu$. 
We are most interested in the case where $\mu = \mu_t = \mu_0 K^t$ is the distribution of a Markov chain state after $t$ iterations of a $\pi$-invariant MCMC kernel $K$ when started at $\mu_0$. 
Suppose further that we have $M$ strictly positive random weights $W^1, \ldots, W^M$ such that $\sum_{m=1}^M W^m = 1$ and
\begin{equation} 
\label{eq:weighted_representation}
\sum_{m=1}^M W^m \varphi(X^m) \rightarrow \int \varphi(x) \pi(\dd x)
\end{equation}
converges in probability\footnote{Fully rigorous notations read $\sum_{m=1}^M W_M^m \varphi(X_M^m) \to \int \varphi(x) \pi(\dd x)$.} as $M \to \infty$, for all functions $\varphi$ in a reasonably large function class $\Phi$. 
We say that $\sum_{m=1}^M W^m \delta_{X^m}$ is a \textit{weighted representation} of $\pi$. 
We now state a theorem showing how this consistent approximation can be turned into a convergence diagnostic.
We write $\psi = \dd \pi/\dd \mu$ for the Radon--Nikodym derivative of $\pi$ with respect to $\mu$, and $f'(x)$ for the \textit{one-sided} limit $\lim_{h \to 0^+} h^{-1}[f(x+h) - f(x)]$.
By Lemma~\ref{lem:convexity} in Supplement~\ref{app:proof-weighted-diag}, this limit always exists and is finite for finite-valued convex functions defined on an open interval.

\begin{assumption}\label{ass:unweighted-LLN}
    The two functions $x \mapsto f(\psi(x))$ and $x \mapsto \psi(x) f'(\psi(x))$ are finitely integrable under $\mu$, and
    the unweighted empirical measure satisfies the weak law of large numbers for them,
    \begin{align}
    \frac 1M \sum_{m=1}^M f(\psi(X^m)) &\to \int f(\psi(x)) \mu(\dd x), \\
    \frac 1M \sum_{m=1}^M \psi(X^m) f'(\psi(X^m)) &\to \int \psi(x) f'(\psi(x)) \mu(\dd x),
    \end{align}
    where both convergences are in probability as $M \to \infty$.
\end{assumption}

\begin{assumption}\label{ass:weighted-convergence}
    The integral $\int f'(\psi(x)) \pi(\dd x)$ is finite, and 
    the weights $W^m$ are such that
    \[
    \sum_{m=1}^M W^m f'(\psi(X^m)) \to \int f'(\psi(x)) \pi(\dd x)
    \]
    in probability as $M \to \infty$.
\end{assumption}

\begin{theorem}\label{thm:weighted_as_diag}
    Let $\mu$ and $\pi$ be two mutually absolutely continuous distributions, and
    let $X^1, \ldots, X^M$ be M (not necessarily independent) identically distributed draws from $\mu$. 
    Suppose that Assumptions~\ref{ass:unweighted-LLN} and~\ref{ass:weighted-convergence} are satisfied for some strictly positive random weights $W^1, \ldots, W^M$ summing to $1$.
    Then,
    for any fixed $\varepsilon > 0$ we have, as $M \to \infty$,
    \[
    \Prob\paren*{\frac 1M \sum_{m=1}^M f(MW^m) \leq D_f(\pi\|\mu) - \varepsilon} \rightarrow 0.
    \]
\end{theorem}
In plain English, Theorem~\ref{thm:weighted_as_diag} states that the estimator~\eqref{eq:f-div-emp-meas} is an asymptotic upper bound for the $f$-divergence of interest.

\subsection{Naive approximation in MCMC and its suboptimality}
\label{subsec:naive}
Consider $M$ independent parallel MCMC chains $(X_{0:T}^m)$ for $m=1, \ldots, M$ where $X_0 \sim \mu_0$ and $X_{t+1}^m \sim K(x_t^m, \cdot)$.
We make the following assumption throughout the paper.
\begin{assumption}
    \label{asp:absolute-continuity}
    The starting distribution $\mu_0$ and the target distribution $\pi$ are mutually absolutely continuous. 
\end{assumption}
This allows us to define
\[
W_0^m = \frac{(\dd \pi/\dd \mu_0)(X_0^m)}{\sum_{i=1}^M (\dd \pi/\dd \mu_0)(X_0^i)},
\]
which is tractable because the quantity does not depend on the normalizing constant of $\pi$.
Let $\varphi$ be a test function such that $\int |\varphi| \dd \pi < \infty$.
Then, by applying the law of large numbers to the function $(x_0, x_t) \mapsto (\dd \pi/\dd \mu_0)(x_0) \varphi(x_t)$ and the measure $\mu_0(x_0) K^t(x_0, \dd x_t)$, we have the convergence in probability
\[
\sum_{m=1}^M W_0^m \varphi(X_t^m) \to \int \varphi(x) \pi(\dd x)
\]
as $M \to \infty$, for any fixed $t$.
Therefore, Theorem~\ref{thm:weighted_as_diag} says that the $f$-divergence of $\pi$ with respect to $\mu_t$, for \textit{any} $t$, can be asymptotically upper-bounded by $\sum_{m=1}^M f(MW_0^m) / M$.
This bound is suboptimal: it does not vary in $t$ and ignores the mixing of the Markov chain.
We now present a coupling scheme where the weights are `harmonized' as the chain progresses, reflecting its mixing.

\section{Weight-harmonization via couplings}\label{sec:weight-harmonization}
\subsection{Couplings of Markov kernels}
In Section~\ref{subsec:naive}, we considered MCMC chains running independently in parallel, an approach that precludes information sharing. 
The key component of our method relies instead on pairwise interactions between two sets of $N$ chains. 
When the states of two Markov chains carrying different weights coincide (we say that they couple or meet), the chains become indistinguishable and consequently average their accumulated weights.

However, for target distributions with Lebesgue density in $\bbR^d$ simulated via popular MCMC algorithms, such as most Metropolis--Hastings-corrected dynamics~\citep{metropolis1953equation,hastings1970mcmc}, the states of two distinct chains will never naturally coincide.
Formally, for continuous state-spaces the \textit{independent coupling} $\bar{K}_{\textrm{ind}}(x, y, \dd x', \dd y') = K(x, \dd x') K(y, \dd y')$ has a zero meeting probability: $\Prob(X'=Y'\mid X,Y) = 0$ for any $X \neq Y$.

To resolve this, we instead consider couplings $\bar{K}(x,y, \dd x', \dd y')$ designed such that the meeting probability can be positive, while each chain taken in isolation behaves as an ordinary Markov chain: the marginal $\bar{K}(x, y, \dd x') = K(x, \dd x')$ does not depend on $y$ and $\bar{K}(x, y, \dd y') = K(y, \dd y')$ does not depend on $x$. 
In Supplement~\ref{app:couplings}, we detail couplings for usual MCMC algorithms.

\subsection{Equalizing weights of coupled particles}\label{subsec:construction}
Because our scheme relies on the notion of `couples', it is more convenient to work with a system of $M$ draws where $M$ is even.
In the rest of this paper we assume that this condition holds, and let $N$ be such that $M=2N$.

Consider a $\pi$-invariant kernel $K(x, \dd x')$ and a coupling $\bar{K}$ of $K$ with itself.
At initialization ($t=0$), we simulate $X_0^1, \ldots, X_0^{2N}$ i.i.d. from an initial distribution $\mu_0$.
Recalling that $\gamma$ is the un-normalized density of $\pi$ that we have access to, we define the un-normalized weights $w_0^m = \gamma(X_0^m)/\mu_0(X_0^m)$ for $m \in \curly{1, \ldots, 2N}$ as well as their normalized counterparts $W_0^m = w_0^m/\sum_{i=1}^M w_0^i$.

Now, for $t\geq0$, let $\sum_{m=1}^{2N} W_{t}^m \delta_{X_t^m}$ be a weighted representation of $\pi$, in the sense of~\eqref{eq:weighted_representation}, for a population $X^{1:2N}_t$ distributed identically according to $\mu_t = \mu_0 K^t$.
We obtain $(X_{t+1}^{n}, X_{t+1}^{n + N})$ by applying the kernel $\bar{K}$ to the pairs $(X_t^n, X_t^{n + N})$, independently for $n = 1, \ldots, N$.
This gives a new particle representation of $\pi$ as $\sum_{m=1}^{2N} W_{t}^m \delta_{X_{t+1}^m}$.

Assume that, under the coupled simulation, we obtain $X_{t+1}^j = X_{t+1}^{j + N}$ for some $j \in \{1, \ldots, N\}$.
Then 
\begin{equation}
  W^j_{t} \delta_{X_{t+1}^j} + W^{j + N}_{t} \delta_{X_{t+1}^{j+N}} = 
  \curly*{\alpha W^j_{t} + (1 - \alpha)W^{j + N}_{t}} \delta_{X_{t+1}^j}+
  \curly*{(1-\alpha)W^j_{t} + \alpha W^{j + N}_{t}} \delta_{X_{t+1}^{j+N}}
\end{equation}
for any $\alpha \in [0, 1]$. 
Taking $\alpha = 1/2$ in what follows,
we can therefore set, for $n \in \{1, \ldots, N \}$,
\begin{equation}
\label{eq:weight-update}
\begin{cases}
    W_{t+1}^n = W_{t+1}^{n+N} = (W_t^n + W_t^{n+N})/2 &\text{ if } X_{t+1}^n = X_{t+1}^{n+N},\\
    W_{t+1}^n = W_t^n; W_{t+1}^{n+N} = W_{t}^{n+N} &\text{ otherwise.}
\end{cases}
\end{equation}
This combines the weights of the two particles into a unique shared weight, reducing the variability of the weights and hopefully decreasing the $f$-divergence upper bound as per Theorem~\ref{thm:weighted_as_diag}.
Clearly, more indices $n$ that couple lead to a greater reduction in the variance of the weights.

In Supplement~\ref{app:thm1bis}, we derive an upper bound on the maximum improvement of the $f$-divergence bound we can achieve after one step of weight harmonization and justify the choice $\alpha=1/2$.
\begin{remark}\label{rem:unnormalized_weights}
    Under the weight update rule~\eqref{eq:weight-update}, we have $\sum_{m=1}^{2N} W_{t+1}^{m} = \sum_{m=1}^{2N} W_t^{m} = 1$.
    While we have written the procedure over normalized weights, the same operation can be performed over un-normalized weights:
\begin{equation}
\begin{cases}
    w_{t+1}^n = w_{t+1}^{n+N} = (w_t^n + w_t^{n+N})/2 &\text{ if } X_{t+1}^n = X_{t+1}^{n+N},\\
    w_{t+1}^n = w_t^n; w_{t+1}^{n+N} = w_{t}^{n+N} &\text{ otherwise.}
\end{cases}
\end{equation}
  Consequently, the sum of the weights is unchanged by the harmonization operation: $\sum_{m=1}^{2N} w^m_{t+1} = \sum_{m=1}^{2N} w_{t}^m$.
  In particular, computing
  $\tilde{W}^m_{t+1} = w^m_{t+1} / \sum_{i=1}^{2N} w^i_{t+1}$ post-normalization yields the same result as directly applying the update on the normalized weights.
\end{remark}

\subsection{Exchanging pairs of particles}\label{subsec:weight-harmonization}
Coupling the same particles over and over again will result in $N$ pairs of equal particles, but will not modify the weights of the particles across these pairs.
In our notations, repeatedly coupling $X_t^n$ and $X_t^{n+N}$ for all $t$ might eventually result in $W_t^n = W_t^{n+N}$ for $t$ large enough, but not $W_t^n = W_t^m$ for $1 \leq m < n \leq N$ since the $m$-th and $n$-th chains would have no chance to communicate.
We introduce exchange of information by randomizing the pairings of particles.
\begin{figure}[t!]
    \centering
    \resizebox{\linewidth}{!}{
    \begin{tikzpicture}[
    node distance=0.35cm and 1.5cm,
    cell/.style={
        rectangle, 
        draw, 
        thick, 
        minimum width=2.8cm, 
        align=center, 
        font=\normalsize,
        text height=2.8ex, %
        text depth=1.5ex   %
    },
    kernel/.style={
        circle, 
        draw, 
        thick, 
        align=center, 
        inner sep=5pt
    },
    bracket/.style={
        decorate, 
        thick,
        decoration={brace, 
                    amplitude=5pt,
                    mirror}
    }
]

\node[align=center] (title1) {Time $t$};

\node[cell, fill=gray!20, below=of title1] (c1_t) {$X_t^1, W_t^1$};
\node[cell, fill=gray!20, below=of c1_t] (c2_t) {$X_t^2, W_t^2$};
\node[cell, below=of c2_t] (c3_t) {$X_t^3, W_t^3$};
\node[cell, below=of c3_t] (c4_t) {$X_t^4, W_t^4$};

\draw[bracket] ([xshift=-5pt]c1_t.north west) -- ([xshift=-5pt]c2_t.south west)
    node[midway, xshift=-10pt, rotate=90] {$1:N$};

\draw[bracket] ([xshift=-5pt]c3_t.north west) -- ([xshift=-5pt]c4_t.south west)
    node[midway, xshift=-10pt, rotate=90] {$N+1:2N$};

\node[kernel, right=of c2_t, anchor=center] (k13) {$\bar{K}$};
\node[kernel, right=of c3_t, anchor=center] (k24) {$\bar{K}$};
\node[align=center, right=of title1, anchor=west] (title2) {Coupling};

\draw[-{Stealth[length=2mm]}] (c1_t.east) -- (k13.north west);
\draw[-{Stealth[length=2mm]}] (c3_t.east) -- (k13.south west);
\draw[-{Stealth[length=2mm]}] (c2_t.east) -- (k24.north west);
\draw[-{Stealth[length=2mm]}] (c4_t.east) -- (k24.south west);

\node[cell, fill=gray!20, right=of c1_t, xshift=2cm] (c1_a) {$X_{t+1}^1, \frac{W_t^1+W_t^3}{2}$};
\node[cell, fill=gray!20, right=of c2_t, xshift=2cm] (c2_a) {$X_{t+1}^2, \frac{W_t^2+W_t^4}{2}$};
\node[cell, right=of c3_t, xshift=2cm] (c3_a) {$X_{t+1}^3, \frac{W_t^1+W_t^3}{2}$};
\node[cell, right=of c4_t, xshift=2cm] (c4_a) {$X_{t+1}^4, \frac{W_t^2+W_t^4}{2}$};

\node[align=center, above=of c1_a] (title3) {Weight Averaging};

\coordinate (split1) at ($(k13.east)!0.5!(c1_a.west |- k13.east)$);
\draw (k13.east) -- (split1);
\draw[-{Stealth[length=2mm]}] (split1) -- (c1_a.west);
\draw[-{Stealth[length=2mm]}] (split1) -- (c3_a.west);

\coordinate (split2) at ($(k24.east)!0.5!(c2_a.west |- k24.east)$);
\draw (k24.east) -- (split2);
\draw[-{Stealth[length=2mm]}] (split2) -- (c2_a.west);
\draw[-{Stealth[length=2mm]}] (split2) -- (c4_a.west);

\node[cell, fill=gray!20, right=of c1_a] (c2_s) {$X_{t+1}^2, \frac{W_t^2+W_t^4}{2}$};
\node[cell, fill=gray!20, right=of c2_a] (c1_s) {$X_{t+1}^1, \frac{W_t^1+W_t^3}{2}$};
\node[cell, right=of c3_a] (c3_s) {$X_{t+1}^3, \frac{W_t^1+W_t^3}{2}$};
\node[cell, right=of c4_a] (c4_s) {$X_{t+1}^4, \frac{W_t^2+W_t^4}{2}$};

\node[align=center, above=of c2_s] (title4) {Time $t+1$};

\draw[-{Stealth[length=2mm]}, thick, dashed] (c1_a.east) to[bend left] node[pos=0.38, above=0.4cm, font=\small]{swap} (c1_s.west);
\draw[-{Stealth[length=2mm]}, thick, dashed] (c2_a.east) to[bend right] (c2_s.west);
\draw[-{Stealth[length=2mm]}] (c3_a.east) -- (c3_s.west);
\draw[-{Stealth[length=2mm]}] (c4_a.east) -- (c4_s.west);

\node[kernel, right=of c1_s, anchor=center] (ks13) {$\bar{K}$};
\node[kernel, right=of c3_s, anchor=center] (ks24) {$\bar{K}$};
\node[align=center, right=of title4, anchor=west] (title5) {Coupling};

\draw[-{Stealth[length=2mm]}] (c1_s.east) -- (ks24.north west);
\draw[-{Stealth[length=2mm]}] (c3_s.east) -- (ks13.south west);
\draw[-{Stealth[length=2mm]}] (c2_s.east) -- (ks13.north west);
\draw[-{Stealth[length=2mm]}] (c4_s.east) -- (ks24.south west);

\coordinate (end_x) at ([xshift=1.5cm]ks13.east);

\coordinate (split3) at ($(ks13.east)!0.5!(end_x |- ks13.east)$);
\draw[-{Stealth[length=2mm]}, dashed] (ks13.east) -- (split3);

\coordinate (split4) at ($(ks24.east)!0.5!(end_x |- ks24.east)$);
\draw[-{Stealth[length=2mm]}, dashed] (ks24.east) -- (split4);

\end{tikzpicture}
    }
    \caption{Step of Algorithm~\ref{alg:weight-harmonization} for $2N = 4$ particles and successful couplings: $X_{t+1}^1 = X_{t+1}^3$ and $X_{t+1}^2 = X_{t+1}^4$.}
    \label{fig:illustration}
\end{figure}
More precisely, given the particles $X_{t+1}^{1:N}$, we do not couple the one at index $n$ with the one at index $n+N$, but the one with index $n$ with the one at index ${A_{t+1}^n+N}$ for a time-changing permutation $A_s^{1:N}$ of $\{1, \ldots, N\}$, $s \geq 0$. 
We determine $A_{t+1}^{1:N}$ from the previous pairing $A_t^{1:N}$ with two objectives in mind: pairs that have not yet coupled should be left alone, and pairs that have coupled should be exchanged whenever possible. 
This amounts to setting $A_{t+1}^{1:N}$ to a modification of the array $A_t^{1:N}$ where the subset of coupled indices is permuted and the rest of the indices are left intact.
The algorithmic description of the method is given in Algorithm~\ref{alg:weight-harmonization} while a more intuitive visual illustration is given in Figure~\ref{fig:illustration}.

\begin{remark}\label{rem:unnormalized_weights_algorithm}
    We have written Algorithm~\ref{alg:weight-harmonization} in terms of normalized weights $W^{1:2N}$, but as per Remark~\ref{rem:unnormalized_weights}, this can be written directly in terms of the un-normalized weights with no impact on the downstream representation.
\end{remark}
The initialization of the method is given by simply taking independent samples from an initial distribution $\mu_0$ and weighting them accordingly as $X^{n}_{0} \sim \mu_0(\dd x)$, $w_0^n = \gamma(X^n_0) / \mu_0(X^n_0)$, and $W_0^n = w_0^n / \sum_{m=1}^{2N} w_0^m$ for $n=1, \ldots, 2N$.

\begin{fakealgoenv}\label{alg:weight-harmonization}
Weight-harmonization of MCMC simulations
\begin{tabbing}
\qquad \enspace \textbf{Input:} \= Particles $X_t^{1:2N}$, weights $W_t^{1:2N}$, pairings $A_t^{1:N}$, and coupled kernel $\bar{K}$.\\
\qquad \enspace \textbf{Output:} Updated particles $X_{t+1}^{1:2N}$, weights $W_{t+1}^{1:2N}$, and pairings $A_{t+1}^{1:N}$. \\
\\
\qquad \enspace \textit{1. Couple particles and update weights} \\
\qquad \qquad Set $C_{t+1} \leftarrow \emptyset$. \hspace{4.5em} \textit{\% Collect indices of coupled pairs.} \\
\qquad \qquad \textbf{For} $n=1$ \textbf{to} $N$ \textbf{do} \\
\qquad \qquad \qquad Let $m = A_t^n$. \\
\qquad \qquad \qquad Sample $(X_{t+1}^n, X_{t+1}^{m+N}) \sim \bar{K}(X_t^n, X_t^{m+N}, \cdot, \cdot)$. \\
\qquad \qquad \qquad \textbf{If} $X_{t+1}^n = X_{t+1}^{m+N}$ \textbf{then} \\
\qquad \qquad \qquad \qquad Let $w_{*} = (W_t^n + W_t^{m+N}) / 2$. \\
\qquad \qquad \qquad \qquad Set $W_{t+1}^n = w_{*}$ and $W_{t+1}^{m+N} = w_{*}$. \\
\qquad \qquad \qquad \qquad Set $C_{t+1} = C_{t+1} \cup \{n\}$. \\
\qquad \qquad \qquad \textbf{Else} \\
\qquad \qquad \qquad \qquad Set $W_{t+1}^n = W_t^n$ and $W_{t+1}^{m+N} = W_t^{m+N}$. \\
\qquad \qquad \qquad \textbf{End If} \\
\qquad \qquad \textbf{End For} \\
\\
\qquad \enspace \textit{2. Reshuffle pairings for the next step} \\
\qquad \qquad Set $A_{t+1} = A_t$. \\
\qquad \qquad \textbf{If} $\abs{C_{t+1}} > 1$ \textbf{then} \\
\qquad \qquad \qquad Sample uniformly a permutation $\sigma$ of $\{ 1, \ldots, N \}$ given $\sigma(n) = n, \forall n \notin C_{t+1}$, \\
\qquad \qquad \qquad and optionally, $\sigma(m) \neq m, \forall m \in C_{t+1}$.\\
\qquad \qquad \qquad \textbf{For each} $n =1, \ldots, N$ \textbf{do} \\
\qquad \qquad \qquad \qquad Set $A_{t+1}^n = A_t^{\sigma(n)}$. \\
\qquad \qquad \qquad \textbf{End For} \\
\qquad \qquad \textbf{End If} \\
\end{tabbing}
\end{fakealgoenv}

\subsection{Interpretation and usage}
At each step, the algorithm produces $2N$ draws $X_t^{1:2N}$ and weights $W_{t}^{1:2N}$.
These quantities can be used in at least two ways.
First, $\sum_{m=1}^{2N} W_t^m \varphi(X_t^m)$ is an approximation of $\int \varphi(x) \pi(\dd x)$.
Second, $\frac{1}{2N} \sum_{m=1}^{2N} f(2N W_t^m)$ is an asymptotic upper bound of $D_f(\pi\|\mu_t)$.
Moreover this upper bound converges to $0$ almost surely as $t \to \infty$.
This justifies the use of Algorithm~\ref{alg:weight-harmonization} in convergence diagnostics.

In Section~\ref{sec:theory} these statements will be rigorously formulated and proved under appropriate regularity conditions; see Theorems~\ref{thm:consistency} and~\ref{thm:as-convergence} and Corollaries~\ref{cor:diag-smooth} and~\ref{cor:diag_kl}.
Finally, in Supplement~\ref{app:weight-harmonization-fixed}, we detail an end-to-end implementation of Algorithm~\ref{alg:weight-harmonization} for a fixed horizon $T$.

\section{Theoretical results}
\label{sec:theory}
\subsection{Consistency for a large number of chains}
The following proposition shows that Algorithm~\ref{alg:weight-harmonization} preserves expectations.
\begin{proposition}[Invariance of expectations under Algorithm~\ref{alg:weight-harmonization}]\label{prop:expectation-invariance}
	Consider the un-normalized estimator $\hat{I}_{t,N}(\varphi) = \sum_{m=1}^{2N} w_t^m \varphi(X_t^m)$ at step $t$ of Algorithm~\ref{alg:weight-harmonization}, where the $w_t^n$ are the un-normalized weights as per Remark~\ref{rem:unnormalized_weights_algorithm}.
	If the Markov kernel $K$ admits $\pi$ as an invariant distribution and Assumption~\ref{asp:absolute-continuity} holds, then for any bounded function $\varphi$ and for all $t \ge 0$, the expectation
	\begin{equation}
	\begin{split}
	\Exp\brac{\hat{I}_{t,N}(\varphi)}
	&= \Exp\brac{\hat{I}_{0,N}(\varphi)} = 2N \int \varphi(x) \gamma(\dd x)
	\end{split}
	\end{equation}
	is constant over time.
\end{proposition}
Under a coupler assumption that is usually satisfied in practice, we can interpret the weights as noisy estimates of un-normalized Radon-Nikodym derivatives.
\begin{assumption}
    \label{asp:kernel_symmetric}
    The coupling $\bar K(x, y, \dd x', \dd y')$ is symmetric, in the sense that for real-valued bounded functions $f$,
    \[
    \int f(x',y') \bar K(x, y, \dd x', \dd y') = \int f(x', y') \bar K(y, x, \dd y', \dd x').
    \]
\end{assumption}
\begin{corollary}[Radon-Nikodym interpretation]
\label{coroll:unbiased}
    Under the same setting as Proposition~\ref{prop:expectation-invariance}, if Assumption~\ref{asp:kernel_symmetric} holds, the un-normalized weight $w_t^m$ is an unbiased estimate of the Radon-Nikodym derivative between the un-normalized target and the law of the current iterate,
    \[
    \Exp\brac{w_t^m | X_t^m} = \frac{\dd \gamma}{\dd \mu_t}(X_t^m).
    \]
\end{corollary}
In addition to having invariant expectations, the estimates in Proposition~\ref{prop:expectation-invariance} are consistent at any iteration of the algorithm.
\begin{theorem}[Consistency]\label{thm:consistency}
	For any initial distribution such that $\Exp_\pi[w_0^2]< \infty$ and function $\varphi$ such that $\Exp_\pi\brac{\varphi(X)^4} < \infty$, we have
	\begin{align}
	\frac{1}{2N}\sum_{m=1}^{2N} w_t^m \varphi(X_t^m) {\rightarrow} \int \varphi(x) \gamma(\dd x), \quad
    \sum_{m=1}^{2N} W_t^m \varphi(X_t^m) {\rightarrow} \int \varphi(x) \pi(\dd x),
	\end{align}
	as $N \to\infty$, where the convergences are in probability.
\end{theorem}
The exact rate at which the convergence happens for different times $t$ is complex and may depend on the properties of the Markov chain, the coupling, and the target distribution. 
However it cannot worsen as $t$ increases: the variance of the weights cannot increase from iteration to iteration, as we shall prove.

\subsection{From consistency to diagnostics}\label{subsec:consistency-diagnostics}

Combining Theorem~\ref{thm:weighted_as_diag} and Theorem~\ref{thm:consistency}, we see that $\sum_{m=1}^{2N} f(2NW^m) /(2N)$ can be used as an asymptotic upper bound for $D_f(\pi \| \mu_t)$, where $\mu_t$ is the distribution of the MCMC chain at time $t$.
While it is not obvious to express Assumptions~\ref{ass:unweighted-LLN} and~\ref{ass:weighted-convergence} directly in terms of the initial distribution $\mu_0$ and the target distribution $\pi$, we point out some important cases in which these conditions are easily verified.
\begin{corollary}
\label{cor:diag-smooth}
    Suppose that the weights $w_0^n$ are bounded, i.e. there exists $L < \infty$ such that the Radon--Nikodym derivative $\dd \pi/\dd \mu_0$ satisfies $0 < \dd \pi/\dd \mu_0 \leq L$ almost surely with respect to $\mu_0$.
    Consider an $f$-divergence such that $f$ can be extended to be continuously differentiable at $0$. Then
    \[
    \Prob\paren*{\frac{1}{2N} \sum_{m=1}^{2N} f(2NW_t^m) \leq D_f(\pi\|\mu_t) - \varepsilon} \rightarrow 0
    \]
    as $N \to \infty$.
\end{corollary}

This corollary requires the continuous differentiability of $f$ at time $0$ and therefore is not applicable to the Kullback--Leibler divergence.
The following result addresses this shortcoming.
\begin{corollary}
\label{cor:diag_kl}
    Let the space be $\bbR^d$ and suppose that the initial distribution $\mu_0$ and the target distribution $\pi$ have densities with respect to the Lebesgue measure. 
    Assume that there exist finite $L_1$ and $L_2$ such that $0 < \pi(x) \leq L_1 \mu_0(x)$ and $0 < \mu_0(x) \leq L_2$. Suppose further that $\Exp_{\pi}\curly{(1+ \abs{\log \pi(X)})^4} < \infty$ and that the kernel $K$ is regular in the sense of Supplement~\ref{app:regular-kernel}. Then
    \[
    \Prob\paren*{\frac {1}{2N}\sum_{m=1}^{2N} f(2NW_t^m) \leq \kl(\pi\|\mu_t) - \varepsilon} \rightarrow 0
    \]
    as $N \to \infty$ for $f(t) = t\log t$ with the convention that $f(0) = 0$.
\end{corollary}
In Supplement~\ref{app:corollary-kl-gauss} we explicitly verify the regularity conditions for a random walk Metropolis--Hastings kernel on a Gaussian target distribution.
\begin{remark}
Corollaries~\ref{cor:diag-smooth} and~\ref{cor:diag_kl} do not cover the \emph{reversed} $\kl$ divergence, given by $f(t) = -\log t$. The derivative magnitude $\abs{f'(t)} = 1/t$ blows up more rapidly at time $0$ than what can be compensated by usual moment conditions.
\end{remark}

\subsection{Nonasymptotic behaviour of expected bound}
We have constructed an asymptotic upper bound for $D_f(\pi\|\mu_t)$ as $M\to \infty$.
This part investigates its properties for fixed $M$.
We start with a negative result: it is not always true that $\Exp\brac{M^{-1}\sum f(MW_t^m)}$ is an upper bound of $D_f(\pi\|\mu_t)$.
\begin{example}
Take $t=0$ and $M=1$.\footnote{We allow $M$ odd here for illustrative purposes. At $T=0$, weight harmonization has not started yet.} Then the normalized weight $W_0^1 = w_0^1/w_0^1 = 1$ leads to
\[\Exp\brac{M^{-1}\sum f(MW_t^m)} = 0,\]
whereas $D_f(\pi\|\mu_0) > 0$ in general.
\end{example}
This example seems degenerate but is deeply linked to the fact that the estimated inverse of the normalizing constant is biased for finite $M$, with $M=1$ permitting an easy calculation.
Indeed, had we known $Z$ exactly, we would have a non-asymptotic upper bound in expectation (note the use of the un-normalized weights $w_t^m$).
\begin{proposition}
\label{prop:exp_upper_bound}
Under Assumption~\ref{asp:absolute-continuity} and Assumption~\ref{asp:kernel_symmetric},
we have
\[
\Exp\brac*{M^{-1} \sum_m f(w_t^m/Z)} \geq D_f(\pi\|\mu_t).
\]
\end{proposition}
This raises the natural question of what exactly $\Exp\brac{M^{-1}\sum f(MW_t^m)}$ upper-bounds.
\begin{proposition}
\label{prop:exp_upper_bound_imprecise}
Suppose that Assumption~\ref{asp:absolute-continuity} holds. Let $\tilde \pi_0^M$ be the probability measure defined by
\[
\int_{\mathcal X} \varphi(x) \tilde \pi_0^M(\dd x) = \Exp\brac*{\sum W_0^m \varphi(X_0^m)}.
\]
Informally, we say that $\tilde \pi_0^M = \Exp\brac{\hat \pi_0^M}$ with
\[
\hat \pi_0^M = \sum W_0^m \delta(X_0^m).
\]
Then
\[
\Exp\brac*{M^{-1}\sum f(MW_t^m)} \geq D_f(\tilde \pi_0^M K^t \| \mu_t).
\]
\end{proposition}
This proposition explains the asymptotic nature of the bound: as $M \to \infty$, we expect that $\tilde \pi_0^M \approx \pi$, and
$D_f(\tilde \pi_0^M K^t \| \mu_t) \approx D_f(\pi K^t \| \mu_t) = D_f(\pi\|\mu_t)$.
Under certain hypotheses, it is possible to quantify this bias.
We use a quantitative version of ergodicity which, importantly, does not require the kernel to be uniformly ergodic.
\begin{assumption}[$V$-geometric ergodicity]
\label{asp:geometric_ergodicity}
There exist a measurable function $V: \mathcal X \to [1, \infty)$ with $\pi(V) < \infty$, a constant $C < \infty$, and $\rho \in [0, 1)$ such that, for all $x \in \mathcal X$ and $t \geq 0$,
\[
\sup_{|\varphi| \leq V} \abs*{K^t \varphi(x) - \pi(\varphi)} \leq C \rho^t V(x).
\]
\end{assumption}
Assumption~\ref{asp:geometric_ergodicity} is $V$-uniform geometric ergodicity in the sense of \citet[Definition~15.2.1]{douc2018mc}, called $V$-uniform ergodicity by \citet[Theorem~16.0.1]{meyn2012markov}; it is typically established through drift and minorization conditions~\citep[Theorem~15.2.4]{douc2018mc}.
\begin{remark}
    Uniform ergodicity corresponds to the special case $V \equiv 1$ of Assumption~\ref{asp:geometric_ergodicity}, but is not required.
    For instance, the Gaussian autoregressive kernel of Section~\ref{subsec:ornstein-uhlenbeck} satisfies the assumption with $V(x) = 1 + |x|$~\citep[Example~15.2.5]{douc2018mc}, whereas it is not uniformly ergodic: for any fixed $t$, $K^t(x, \cdot)$ is a Gaussian distribution whose mean is a nonzero multiple of $x$, so that $\sup_x \operatorname{TV}(K^t(x, \cdot) \| \pi) = 1$.
\end{remark}
\begin{theorem}
\label{thm:exp_upper_bound_precise}
Suppose that $\gamma(x) = Z\pi(x)$ is such that $\gamma(x)/\mu_0(x) \in [L_1, L_2]$ for two strictly positive and finite $L_1 \leq L_2$, write $\kappa = L_2 / L_1$, and let Assumption~\ref{asp:geometric_ergodicity} hold.
Then, for any function $f$ continuously differentiable on  $]0, \infty[$ and all $t \geq 0$,
\[
D_f(\pi\|\mu_t) \leq \Exp\brac*{M^{-1}\sum f(MW_t^m)} + C \pi(V) \frac{\rho^t}{M}\, \kappa^2 (\kappa - 1) \sup_{[\kappa^{-1}, \kappa]} |f'| .
\]
For the total variation distance
\[\operatorname{TV}(\pi\|\mu_t) \leq \Exp\brac*{M^{-1}\sum \frac{|MW_t^m-1|}{2}} + C \pi(V) \frac{\rho^t}{2M}\, \kappa^2 (\kappa - 1).
\]
\end{theorem}
This theorem shows that the downward bias of the upper bound is of order $\rho^t / M$: it evaporates as $M \to \infty$, but also as the chain mixes, at the geometric rate of the kernel.
It furthermore vanishes identically when $\mu_0 = \pi$ ($\kappa = 1$), in which case all weights are equal.
The proof, given in Supplement~\ref{app:proof-exp-upper-bound}, rests on a pointwise control of $\dd \tilde \pi_0^M / \dd \pi$, which is of independent interest.

\subsection{Time convergence of harmonization}\label{subsec:geom-harm}
Section~\ref{subsec:consistency-diagnostics} proved that, at any iteration of Algorithm~\ref{alg:weight-harmonization}, our weighted system of particles can be used to form a valid asymptotic upper bound to the $f$-divergence $D_f(\pi \| \mu_t)$.
The following proposition ensures that this bound is furthermore non-increasing.
\begin{proposition}[Non-increasing $f$-divergence bounds]
	\label{prop:decreasing-bound}
	Let $W_t^n$ be the normalized weights at step $t$ of Algorithm~\ref{alg:weight-harmonization}.
	For any convex function $f$, the asymptotic upper bound verifies
	\begin{equation}\label{eq:weight-decrease}
	\frac{1}{2N}\sum_{m=1}^{2N} f(2N W_{t+1}^m) \leq \frac{1}{2N}
    \sum_{m=1}^{2N} f(2N W_{t}^m)
	\end{equation}
    almost surely.
\end{proposition}
\begin{remark}
    This property is an empirical desirable counterpart to the data-processing inequality applied to Markov chains: for any $f$-divergence $D_f$, $D_f(\pi \| \mu_t)$ must be non-increasing with $t$~\citep[see Theorem 16.1.10 in][for the special case of the $\kl$ divergence, the $f$-divergence case being an immediate generalization]{elementsinformationtheory}.
\end{remark}

In order to make the diagnostic useful however, we need to prove that~\eqref{eq:weight-decrease} decreases to $0$ when the Markov chain converges to its stationary distribution.
We make the following assumptions.
\begin{assumption}
	\label{asp:ergodicity}
	The Markov kernel $K$ is ergodic, i.e.\ $\lim_{n\to\infty} \norm{\delta_x K^n -  \pi}_{\tv} = 0$ for all $x$.
\end{assumption}
\begin{assumption}
\label{asp:compact_set}
    The coupled kernel $\bar{K}$ has a positive probability of coupling, i.e., for all $x, y \in \mathcal X^2$, we have
	\[ \Prob(X' = Y' \mid x, y) > 0 \]
	where the probability is taken over $\bar{K}(x, y, \dd x', \dd y')$.
    Additionally the function $(x, y) \mapsto \Prob(X' = Y' \mid x, y)$ is jointly continuous in $x$ and $y$.
\end{assumption}
\begin{assumption}
	\label{asp:finite_coupling_time}
	Under $\bar{K}$, the coupling time $\tau(x, y) = \inf\{t \mid X_t=Y_t, X_0=x, Y_0=y\}$ is almost surely finite for all $(x, y) \in \mathcal X^2$.
\end{assumption}
Assumptions~\ref{asp:compact_set} and~\ref{asp:finite_coupling_time} are relatively weak and expected to hold in most cases of interest (they for example hold for random-walk kernels, with the coupling of Section~\ref{subsec:reflection-maximal}).

\begin{theorem}\label{thm:as-convergence}
    Under Assumptions~\ref{asp:ergodicity},~\ref{asp:compact_set}, and~\ref{asp:finite_coupling_time}, the system of weights $W^{1:2N}_t$ converges to $\{1/(2N), \ldots 1/(2N)\}^{\top} \in \mathbb{R}^{2N}$ almost surely and in expectation.
    The diagnostic quantity
    \[
    (2N)^{-1} \sum_{m=1}^{2N} f(2NW_t^m)
    \]
also converges to $0$ almost surely, and when $f$ is continuous at $0$, in expectation.
\end{theorem}

While Theorem~\ref{thm:as-convergence} ensures convergence of the weights when the chain mixes well, it does not assess how fast this is happening.
To get the rate of convergence, we introduce the following modification of Assumption~\ref{asp:compact_set}.
\begin{assumption}[Uniform Coupling]\label{ass:coupling}
	The coupled kernel $\bar{K}$ has a \emph{uniform} positive probability of coupling, i.e., there exists a constant $p_c > 0$ such that for any states $x, y$,
	\begin{equation*}
        \Prob(X' = Y' \mid x, y) \ge p_c,
	\end{equation*}
	where the probability is taken over $\bar{K}(x, y, \dd x', \dd y')$.
\end{assumption}
\begin{remark}
    This assumption informally states that the coupling time is at most a geometric distribution with parameter $p_c$ no matter where we start from.

    We note that the P\'olya--Gamma sampler~\citep{Polson01122013} used in Section~\ref{subsec:polya-gamma} verifies this~\citep[Appendix D]{jacob2020unbiasedmcmc}, as well as pseudo-marginal MCMC samplers~\citep[see, e.g., ][Proposition 10]{middleton2019unbiased} and independent Metropolis--Hastings-like samplers such as conditional sequential Monte Carlo methods~(see e.g., \citealt{jacob2020unbiasedsmoothing}, Lemma 1.1 in the supplementary material, or \citealt{lee2020coupled}, Theorems 8, 9).
    
	Still, this is admittedly an unrealistic assumption in many practical samplers, including simple autoregressive processes similar to Section~\ref{subsec:ornstein-uhlenbeck}, or the example considered in Section~\ref{subsec:stochvol}. However, it allows us to compute an informative rate of convergence that we expect to hold more generally.
\end{remark}
Under this assumption, the following theorem ascertains that the variance of the weights converges to $0$ exponentially fast as the number of MCMC iterations increases.
\begin{theorem}\label{thm:ergodic-weight-convergence}
	Let $W_t = (W_t^1, \dots, W_t^{2N})$ be the vector of weights at iteration $t$ of the weight-harmonization algorithm. 
    Under Assumption~\ref{ass:coupling}, as $t \to \infty$, the weight vector $W_t$ converges exponentially fast in mean square to the uniform weight vector $\bar{W} = \{1/(2N), \dots, 1/(2N)\}$, i.e.,
	\begin{equation*}
	\Exp\paren*{\norm{W_t - \bar{W}}_2^2} = \mathcal O(\rho^{t/2})
	\end{equation*}
	with $\rho = 1 - p_c^3 / 2$.
\end{theorem}

\section{Numerical illustrations}\label{sec:exps}
\subsection{A fully tractable system}\label{subsec:ornstein-uhlenbeck}
We first turn to a case where all marginal distributions are tractable, and therefore the properties of the proposed method can be assessed without resorting to approximations.
When $\mu_0 \sim \mathcal{N}(\mu, \Sigma)$ and $\pi \sim \mathcal{N}(0, I)$ are Gaussian and $K(x, \dd y) = \mathcal{N}\curly{y; \rho x, (1 - \rho^2) I} \dd y$, for $\rho \in (0, 1)$, we have $\mu_t \sim \mathcal{N}\curly{\rho^t \mu, \rho^{2t} \Sigma + (1 - \rho^{2t}) I}$. 
As a consequence, it is possible to compute $\int f(\dd \pi / \dd \mu_t) \dd \mu_t$ for any $t \geq 0$ for several choices of $f$-divergences.

A natural choice of coupling $\bar{K}(\dd y', \dd x' \mid y, x)$ for $K$ is the \emph{reflection maximal coupling}~\citep[see, e.g.][]{bou2020coupling}, which we describe in Supplement~\ref{subsec:reflection-maximal}.
In order to analyse the relative efficiency of our method compared to the real weights, in Figure~\ref{fig:ess-gaussian}, we report the effective sample size (ESS) profiles using our method versus the theoretical expected sample size as described in for different values of $\rho$ and $N$, essentially measuring the gap in Theorem~\ref{thm:weighted_as_diag}.
The effective sample size for importance sampling is defined as
\begin{equation}\label{eq:ess-is}
\ess(W^{1:M}) = \frac{1}{\sum_{m=1}^M (W^m)^2},
\end{equation}
which measures the relative variance of a self-normalized weighted estimator $\sum_{m=1}^{2N} W_t^m \varphi(X_t^m)$ to the intractable variance of a empirical average $\sum_{m=1}^{2N} \varphi(Y^m) / (2N)$ where $Y^m\sim \pi$ are i.i.d.
It approximates an asymptotic loss of efficiency~\citep{kong1994sequential} computed via the chi-squared divergence as $\ess^*(M) = \frac{M}{\chi^2\paren*{\pi\| \mu} + 1}$, here corresponding to our theoretical ESS.
In order to make them comparable, all samplers were rescaled into the same ``physical time'': for different $\rho$, $t \gets t (\log \rho) / (\log \rho_{\max})$.

\begin{figure}
    \centering
    \includegraphics[width=0.9\linewidth]{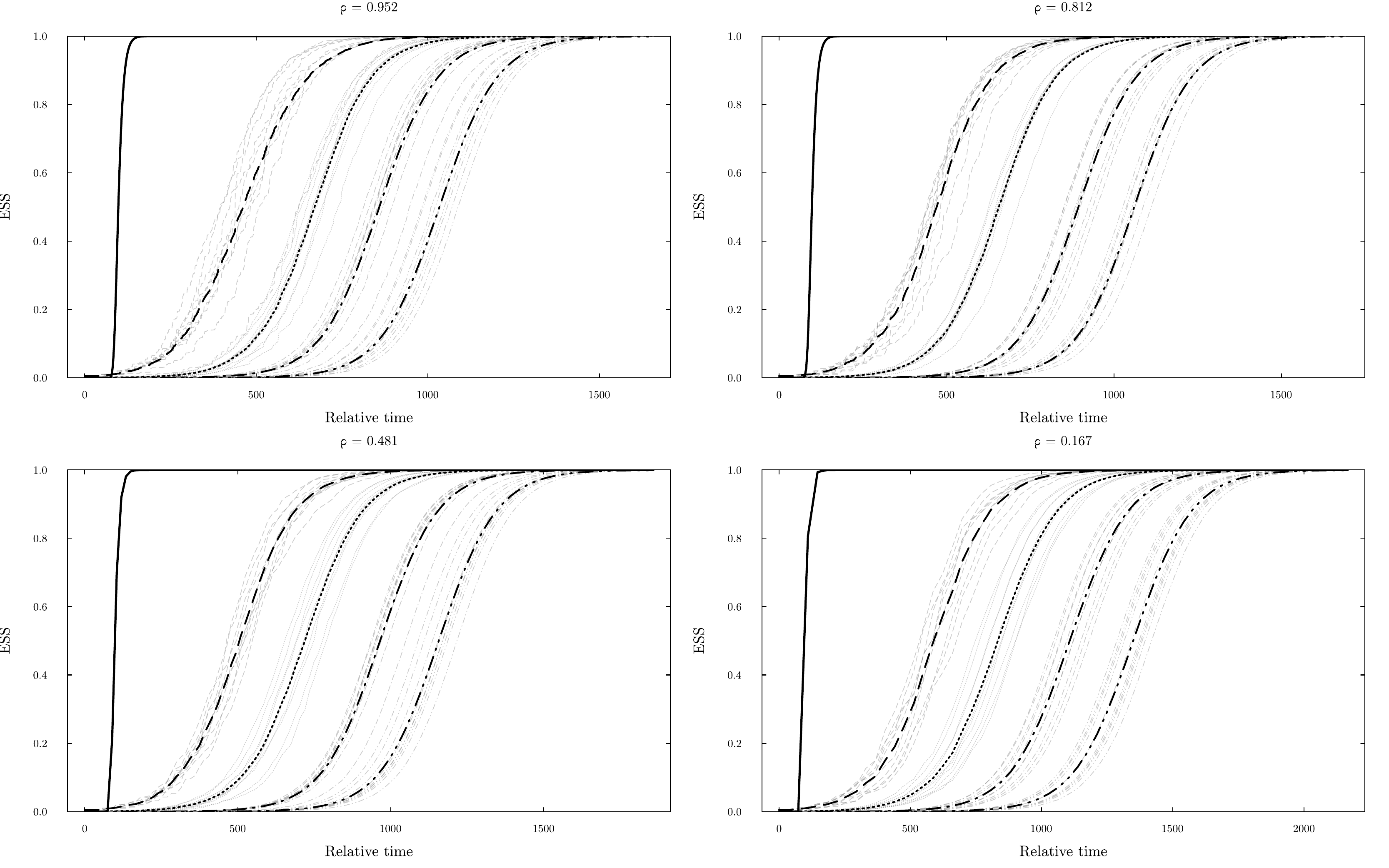}
    \caption[]{Relative effective sample size $\mathrm{rESS} = 1/(1+\chi^2(\pi\|\mu_t))$ for the Ornstein--Uhlenbeck Gaussian example across four mixing rates $\rho$ (panels), against rescaled time $t\log\rho/\log\rho_{\max}$. {\protect\tikz[baseline=-0.5ex]\draw[black, line width=2.5pt] (0,0)--(0.5cm,0);}~closed-form theoretical value; empirical harmonization: {\protect\tikz[baseline=-0.5ex]\draw[gray, line width=1.0pt, dashed] (0,0)--(0.5cm,0);} $N=100$; {\protect\tikz[baseline=-0.5ex]\draw[gray, line width=1.0pt, dotted] (0,0)--(0.5cm,0);} $N=1{,}000$; {\protect\tikz[baseline=-0.5ex]\draw[gray, line width=1.0pt, dashdotted] (0,0)--(0.5cm,0);} $N=10{,}000$; {\protect\tikz[baseline=-0.5ex]\draw[gray, line width=1.0pt, dash pattern=on 5pt off 2pt on 1pt off 2pt on 1pt off 2pt] (0,0)--(0.5cm,0);} $N=100{,}000$. Thin transparent lines show individual seed realizations; thick lines show the mean.}
    \label{fig:ess-gaussian}
\end{figure}
All experiments target the 100-dimensional standard Gaussian, and start from $\mu_0 \sim \mathcal{N}(10\cdot 1_{100}, 5\cdot I_{100})$, largely outside of stationarity.
Confidence intervals are shown as two standard deviations for ten independent realizations of Algorithm~\ref{alg:weight-harmonization}.

A first thing to note: there is no difference in the efficiency that can be attributed to the mixing speed of the kernels: all choices of $\rho$ essentially give the same (after re-scaling) results.
Furthermore, small sample-size effects notwithstanding (we can never give a relative effective sample size smaller than $1 / (2N)$), our method is \emph{systematically} conservative in estimating the $\chi^2$-based effective sample size: this is both a feature, as per Theorem~\ref{thm:weighted_as_diag}, and a drawback, as the underestimation does not seem to improve with the mixing of the kernel, and additionally worsens with the number of particles used.
Indeed, in the figure, all empirical lines are ordered from left to right in their respective number of particles.

In practice, we conjecture that the empirical effective sample size degradation converges to a non-degenerate ``worst case'' scenario which still provides useful bounds on the convergence of the system.
We come back to this point in Section~\ref{subsec:worst-case} and offer possible avenue for future work easing this problem.
Other $f$-divergences profiles are reported in Supplement~\ref{app:other-empirical-res}.

\subsection{P\'olya--Gamma Gibbs sampler for a logistic regression}\label{subsec:polya-gamma}
We now turn to the same real-data example as used by \citet{biswas2019estimating} for evaluating their total-variation-based convergence diagnostic.
The target distribution is defined as a Bayesian logistic regression on the German credit dataset~\citep{statlog_(german_credit_data)_144}.
It consists of n=$1{,}000$ data entries comprising, after encoding and the addition of an intercept, d=$49$ covariates. The covariate vector $x_i \in \mathbb R^d$ is used to predict the creditworthiness of client $i$, encoded as an outcome variable $y_i \in \{-1, 1\}$.
The model is formally described as 
$p(y_{1:n}, \beta \mid x_{1:n}) = \mathcal{N}(\beta; 0_{49}, 10 \cdot I_{49}) \prod_{i=1}^n (1 + e^{-y_i x_i^{\top} \beta})^{-1}$.
A powerful MCMC sampler for $p(\beta \mid y_{1:n}, x_{1:n})$ is the P\'olya--Gamma Gibbs sampler~\citep{Polson01122013}, which performs a Gibbs~\citep{geman1984stochastic} routine over an augmented state space.
We describe the sampler in Supplement~\ref{app:pgg}.
In order to produce comparable results as \citet{biswas2019estimating}, we use the same coupling strategy, again described in Supplement~\ref{app:pgg}, and choose the same initial distribution $\mathcal{N}(0_{49}, 10 \cdot I_{49})$ for all the particles $\beta_0^{1:2N}$ (note the slight change of notations compared to Section~\ref{sec:weight-harmonization} where we used $X$ for the state).
To match their $N=100$ independent estimators of the total variation between two chains, we use $2N = 200$ particles in total.
\begin{figure}[t]
    \centering
    \includegraphics[width=0.75\linewidth]{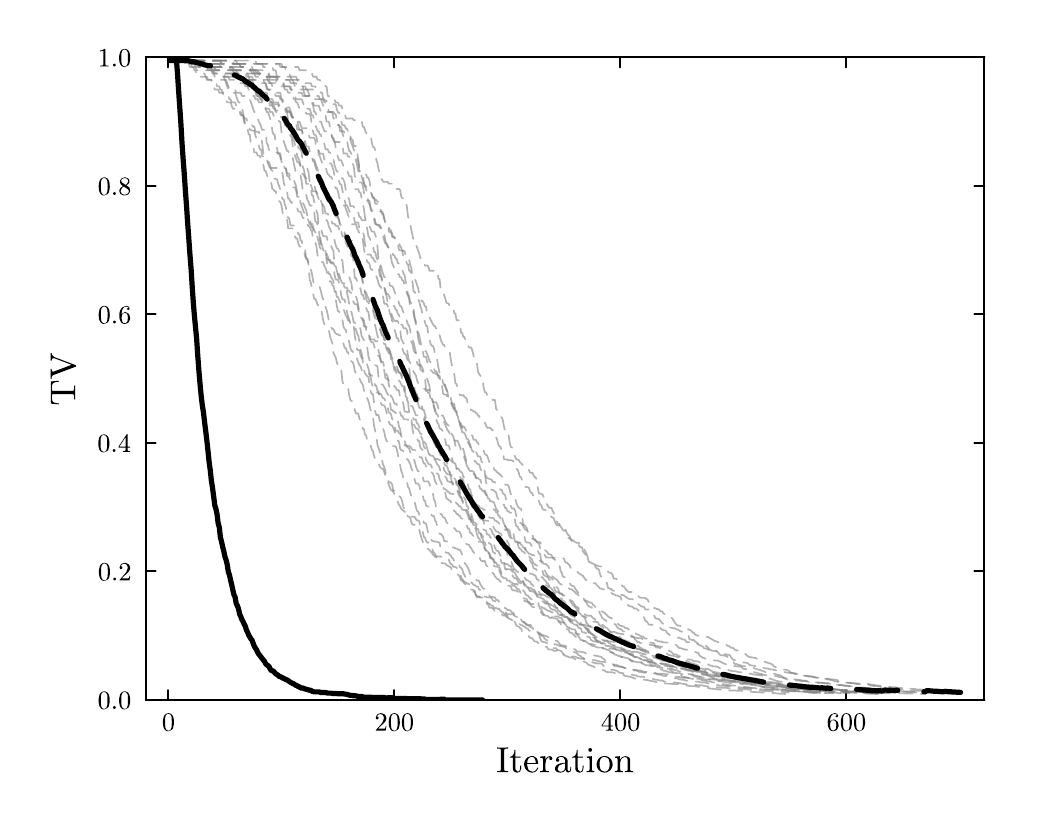}
    \caption[]{ 
    Total variation asymptotic upper bounds for the P\'olya--Gamma Gibbs sampler applied to the German credit dataset logistic regression ($N=100$ particle pairs).{\protect\tikz[baseline=-0.5ex]\draw[black, line width=2.5pt] (0,0)--(0.5cm,0);} \citet{biswas2019estimating} L-lag bound; {\protect\tikz[baseline=-0.5ex]\draw[black, line width=2.5pt, dashed] (0,0)--(0.5cm,0);}~harmonization mean; {\protect\tikz[baseline=-0.5ex]\draw[gray, line width=0.8pt, dashed] (0,0)--(0.5cm,0);}~individual harmonization realizations (20 seeds).}
    \label{fig:ess-pgg}
\end{figure}
In Figure~\ref{fig:ess-pgg}, we report 20 independent realizations of our method and the resulting mean for the approximate total variation bound.%
We also reproduce the total variation bound obtained by implementing \citet{biswas2019estimating} for ease of comparison. 

Our method proves more conservative than \citet{biswas2019estimating} for total variation estimation, though it does not require a pre-run warmup of 350 iterations~\citep[in this specific case for][Section 3.2]{biswas2019estimating}.
Several points help place this finding in context: (i) while \citet{biswas2019estimating} can only ever improve by increasing the number of estimators (because it is unbiased), we have highlighted in the previous section that our bounds increased with $N$, (ii) improvements over \citet{biswas2019estimating} exist: for instance \citet{craiu2022double} implements control variates to reduce the need for the warm-up, but importantly not changing the eventual bound, (iii) \citet{biswas2019estimating} is embarrassingly parallel because only two chains ever communicate, while our method does require communication between $2N$ chains, potentially slowing it down.\footnote{This is however dependent on the choice of hardware: parallelization on CPU clusters would naturally benefit \citet{biswas2019estimating}, but our method would be more easily parallelized on GPUs owing to its deterministic run time.}
Crucially, the two methods differ in a fundamental way: our approach produces importance weights that are unbiased estimates of $\dd \gamma / \dd \mu_t$, and as such carry additional information beyond the total variation bound, with natural applications to importance-weighted inference that \citet{biswas2019estimating} does not offer.

Other $f$-divergences profiles are reported in Figure~\ref{fig:all-stats-credit} in Supplement~\ref{app:other-empirical-res}.

\subsection{MALA for a stochastic volatility model}\label{subsec:stochvol}
We now turn to a stochastic volatility model~\citep{kim1998stochvol} for which we aim to measure the convergence to stationarity of the Metropolis-adjusted Langevin algorithm~\citep{besag1994representations}.
The model is defined as the distribution 
\begin{align}
    \pi(x_{0:L}) 
        \propto \brac*{\mathcal{N}\{x_0; 0, \sigma^2 / (1 - \phi^2)\} \prod_{l=1}^L \mathcal{N}\{x_l; \phi x_{l-1}, \sigma^2\}}\, \times \brac*{\prod_{l=0}^L \mathcal{N}\{y_l; 0, \beta^2 \exp(x_l)\}}
\end{align}
over $\mathbb{R}^{L+1}$ for $L = 2{,}499$.
In order to proceed with the simulation, we fix the hyper-parameters $(\beta, \phi, \sigma) = (0.65, 0.98, 0.15)$, generate a dataset $y_{0:K}$ as per the model, and use the proposal distribution 
\begin{equation}\label{eq:proposal}
    q(x'_{0:L} \mid x_{0:L}) = \mathcal{N}\left(x'_{0:L}; x_{0:L} + \frac{\tau}{2} A \nabla \log \pi(x_{0:L}), \tau A\right)
\end{equation}
where $A$ was computed as the covariance corresponding to the Laplace approximation $\mathcal{N}(\mu, A)$ of $\pi$, and $\tau = 2.89 D^{-1/3}$, which is the optimal value recommended by \citet{roberts1998optimal} and corresponds to a $54\%$ acceptance rate in practice, close to the theoretical optimum.
We then initialize the particles $X_0^{1:2N}$ from independent draws of $\mathcal{N}(\mu, A)$, and compute the corresponding weights $w_0^n \propto \pi(X^n_0) / \mathcal{N}(x^n_0; \mu, A)$. 
The harmonization procedure is then carried using the method in Section~\ref{sec:weight-harmonization} for $N=100$ pairs of chains, and with a maximal reflection coupling of~\eqref{eq:proposal}, as described in Supplement~\ref{subsec:reflection-maximal}.

For comparison purposes, we also computed the total variation upper bound of \citet[see also Section~\ref{sec:introduction}]{biswas2019estimating} using the same coupling strategy, albeit with a larger number $N=250$ pairs of chains, owing to some observed variability in the results for lower numbers. 
The method was used with a warm-up lag of $500$ steps, which, based on our harmonization study was largely sufficient to ensure the chains had reached stationarity.
We report both of these, together with additional statistics computed by our method, in Figure~\ref{fig:stoch-vol}.
\begin{figure}
    \centering
    \includegraphics[width=0.75\linewidth]{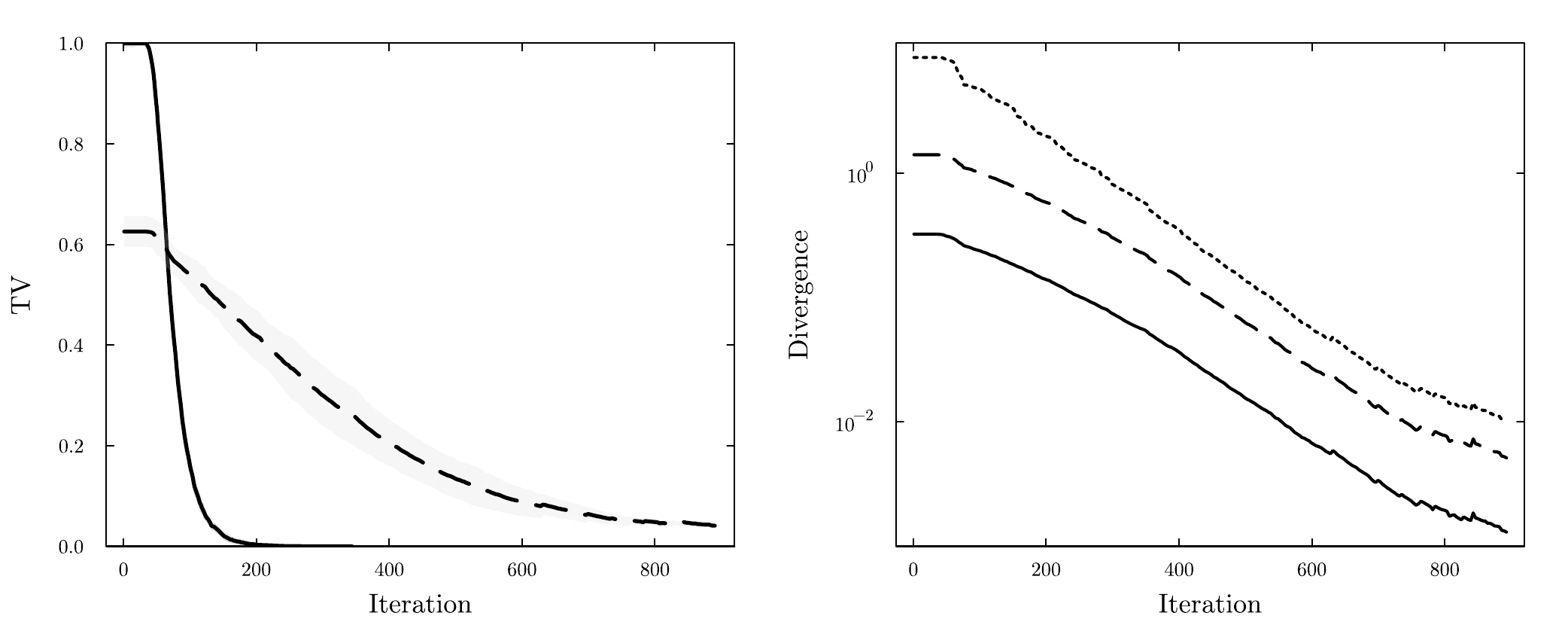}
    \caption[]{Total variation and other $f$-divergences asymptotic upper bounds for the stochastic volatility model. Left: total variation upper bounds for the stochastic volatility model ($D=2{,}500$). {\protect\tikz[baseline=-0.5ex]\draw[black, line width=2.5pt] (0,0)--(0.5cm,0);} \citet{biswas2019estimating} L-lag estimate ($N=250$ chain pairs); {\protect\tikz[baseline=-0.5ex]\draw[black, line width=2.5pt, dashed] (0,0)--(0.5cm,0);}~weight harmonization mean ($N=100$ particle pairs). Shaded bands show $\pm 1$ standard deviation over 10 seeds. Right: harmonization $f$-divergence upper bounds. {\protect\tikz[baseline=-0.5ex]\draw[black, line width=2.0pt] (0,0)--(0.5cm,0);}~$\mathrm{Hellinger}^2$; {\protect\tikz[baseline=-0.5ex]\draw[black, line width=2.0pt, dashed] (0,0)--(0.5cm,0);}~$\mathrm{KL}$; {\protect\tikz[baseline=-0.5ex]\draw[black, line width=2.0pt, dotted] (0,0)--(0.5cm,0);}~$\chi^2$.}
    \label{fig:stoch-vol}
\end{figure}
As in Section~\ref{subsec:polya-gamma}, we exhibit a more conservative upper bound than~\cite{biswas2019estimating}, with the gap being particularly pronounced in the ``warm'' phase of the algorithm, when the total variation of \citet{biswas2019estimating} has already started decreasing.
This behaviour seems independent of the number of dimensions $L$ and we were able to replicate it for lower values too.
We believe this to be a similar occurrence as in Section~\ref{subsec:ornstein-uhlenbeck}, where we reach a regime where the harmonization is penalized by its low level of interaction.
However, a key point to note is that our procedure, in contrast with \citet{biswas2019estimating}, does not require a warm-up lag to be used. 
In this specific instance, we can tell \emph{after the fact} that a reasonable lag should have been around or above $200$ iterations for this, lower than the $500$ steps we used, but it could not have easily been known ahead of time.

\section{Discussion}\label{sec:discussion}
\subsection{Summary}\label{subsec:summary}
We have presented a novel coupling-based approach to estimating the Radon--Nikodym weights $(\dd \gamma / \dd \mu_t)(X_t)$ of the un-normalized target with respect to the current Markov chain iterates.
The procedure, easily implementable as soon as a Markov kernel coupling is available, can be summarized as alternating Markov chain couplings and shufflings to exchange information.

Several properties of the construction have been investigated: first, we (i) have shown consistency of expectations $\pi(\varphi) \approx \sum_{n=1}^{2N} W^n_t \varphi(X^n_t)$ at any point in time for increasing number of particles. 
This is a desirable property that ensures that we can accept the outcome of our algorithm as a trustworthy representation of the target distribution.
Second, we (ii) proved that our Radon--Nikodym estimator recovered the expected convergence behaviour: as the number of iterations increases, our weights homogenize towards a common average value and their variance decreases (under strong assumptions) exponentially fast.
We believe this speed still holds under more realistic assumptions, akin to those in \citet{biswas2019estimating,jacob2020unbiasedmcmc}.
The proof would be largely more convoluted and is left to future work.
Finally, based on a novel bound of $f$-divergences, we (iii) introduced diagnostic criteria for any \emph{$f$-divergence} of the system of Markov chains generated with respect to its target distribution.
This covered in particular the total variation distance.
To the best of our knowledge, our estimator provides the first computable asymptotic upper bound to general $f$-divergences for Markov chains.

Despite its favourable theoretical properties, our method consistently produces more conservative total variation upper bounds than the state of the art~\citep{biswas2019estimating} across all benchmarks considered (Sections, \ref{subsec:ornstein-uhlenbeck}, ~\ref{subsec:polya-gamma} and~\ref{subsec:stochvol}), including when closed-form solutions are available for reference (Section~\ref{subsec:ornstein-uhlenbeck}).
Its key advantage over purely coupling-based alternatives lies in the importance weights themselves: being unbiased estimates of $\dd \gamma / \dd \mu_t$, they can be leveraged to correct Monte Carlo expectations under $\mu_t$ toward the target $\pi$, an application that goes beyond convergence diagnosis and that we identify as an important direction for future work.
While we believe the method to still be useful \emph{as is}, in the following sections we describe possible future avenues for improvement of the diagnostic bounds themselves.

\subsection{Perfect sampling example and Rao--Blackwellization}\label{subsec:worst-case}
Consider the case when the kernel $K(x, \cdot) = \pi(\cdot)$ produces perfect samples from the target no matter $x$, in which case, as soon as $t \geq 1$, $\dd \pi / \dd \mu_t \equiv 1$ is constant almost everywhere, and $D_f(\pi\|\mu_t) = 0$ for every $f$-divergence.
However, it can be proved that, under appropriate conditions  (Theorem~\ref{thm:weight-harmonization} and Corollary~\ref{cor:sweep-improvement} in Supplement~\ref{app:thm1bis}), the improvement brought by Algorithm~\ref{alg:weight-harmonization} on the diagnostic quantity $\mathcal{D}_f(W_t^{1:2N}) = (2N)^{-1}\sum_{m=1}^{2N} f(2N W_t^m)$ is limited by $\mathcal{D}_f(W_{t+1}^{1:2N}) \geq \mathcal{D}_f(W_{t}^{1:2N}) - \bar\lambda_f$, for $\bar\lambda_f = \Delta_f(2N \min W_0^{1:2N}, 2N \max W_0^{1:2N})$ the Jensen gap $\Delta_f(u, v) = \curly*{f(u) + f(v)}/2 - f\curly*{(u+v)/2}$ of $f$ across the range of the initial weights.
For large $N$, because $2N W_0^m$ approaches $(\dd \pi / \dd \mu_0)(X_0^m)$, this can be approximated as 
\begin{equation}
    \mathcal{D}_f(W_{t+1}^{1:2N}) \gtrsim \mathcal{D}_f(W_{t}^{1:2N}) - \Delta_f(r_*, r^*),
\end{equation}
for $r_*$ and $r^*$ the essential infimum and supremum of $\dd \pi / \dd \mu_0$.
This gap, as $N \to \infty$ can be seen as an asymptotic regime of our method: because the range of the initial weights can only widen as more particles are used, the achievable improvement has to increase as $N \to \infty$ until reaching $\Delta_f(r_*, r^*)$, a fixed amount per iteration that no number of particles can improve upon.
We conjecture that this structure holds in general for the expected improvement and general kernels and couplings.

Nonetheless, this sub-efficiency is entirely due to the fact we use pairwise couplings based on the realization of a random permutation.
In theory, other pairings could have been chosen, and in this instance, they would all have been successful: when integrating over all possible permutations of $N+1:2N$, for any $m=1, \ldots, N$, we have
\begin{equation}
\begin{split}
    \Exp\paren*{W^{m}_1 \mid W^{1:2N}_0, X_0^{1:2N}}
        &= \frac{1}{2N}\sum_{n=1}^N (W^m_0 + W^{n+N}_0) = \frac{W^m_0 + \frac{1}{N} \sum_{n=1}^N W^{n+N}_0}{2},
\end{split}
\end{equation}
and similarly for $m=N+1, \ldots, 2N$.
If we were able to update $W_1^m$ directly to the right hand size, the weights would harmonize much faster. 

This simple analysis offers a first route of improvement for the method: (partial) Rao--Blackwellization over the set of chosen permutations.
While this is easy for the simple case highlighted above, doing so for a more general problem is likely harder. 
Indeed, one would need to keep track of possible coupling trajectories throughout the simulation of the Markov chain, a task that will computationally increase with the number of chains and time steps.
A version of this idea has since been implemented by \citet{phan2026MMC}, who replace the pairwise couplings of Algorithm~\ref{alg:weight-harmonization} by multi-marginal ones: the chains are partitioned into groups of size $m$, coupled jointly, and the weights of all the chains that coalesce within a given group are averaged at once, our construction being recovered for $m=2$.
This sidesteps the need to keep track of coupling trajectories, and instead averages each weight against several partners within a single iteration.
On the tractable Gaussian example of Section~\ref{subsec:ornstein-uhlenbeck}, they report squared Hellinger upper bounds markedly closer to the ground truth than the ones obtained by pairwise harmonization, which they attribute to the order-invariance of their Poisson matching construction, and which confirms that the sub-efficiency discussed above is an artefact of the pairwise pairings rather than of weight harmonization itself.

\subsection{Forward-coupling, backward-correcting}\label{subsec:forward-backward}
We have presented an online, single-pass algorithm, where all trajectories are simulated, and past realizations are discarded in computing weights.
In some sense, and ignoring the coupling construction, this can be understood as approximating weights functions at time $t+1$:
\begin{equation}
    w_{t+1}(x_{t+1}) = \int w_t(x_t) \mu_{t+1}(\dd x_t \mid x_{t+1})
\end{equation}
for $\mu_{t+1}(x_t \mid x_{t+1}) \propto \mu_t(x_t) K(x_{t}, x_{t+1})$, using samples from the prior distribution $X^{1:N}_t \sim \mu_t$.
This is similar to approximating smoothing distributions using filtering genealogies in particle filtering~\citep{dau2023complexity}, an approach known to quickly exhibit degeneracy.

This hints at a second route of improvement for the method: if one is willing to perform offline corrections, backward simulation similar to those used in sequential Monte Carlo methods~\citep{chopin2020book} could be used to improve upon the computation of the weights \emph{after an initial pass of our algorithm has been achieved}.

\subsection{Control variates and variance reduction}\label{subsec:control-variates}
In \citet{craiu2022double}, the authors introduce control variates for the upper bound of \citet{biswas2019estimating}, thereby improving the required ``L-lag'' warm-up of the method. 
Such a style of improvement is a likely fruitful research direction for harmonization too: for instance, related but tractable target distributions may be used in a fashion similar to \citet{goodman2009coupling}, albeit at the cost of implementing a 4-way coupling.
The resulting methodology, analysis, and implementation are however not directly obvious and we therefore leave this for future work too.

\appendix
\section{Fixed-horizon end-to-end implementation of Algorithm~\ref{alg:weight-harmonization}}
\label{app:weight-harmonization-fixed}

Algorithm~\ref{alg:weight-harmonization-fixed} describes an end-to-end implementation of Algorithm~\ref{alg:weight-harmonization} to produce an upper bound of the $f$-divergence of the target distribution $\pi$ with respect to the law $\mu_T$ of the MCMC iterates at time $T$.

\FloatBarrier
\begin{fakealgoenv}\label{alg:weight-harmonization-fixed}
End-to-end implementation of weight harmonization
\begin{tabbing}
\qquad \enspace \textbf{Input:} \= Initial distribution $\mu_0$, un-normalized target density $\gamma$, coupled kernel $\bar{K}$, and\\
\qquad \enspace final horizon $T$.\\
\qquad \enspace \textbf{Output:} For any $t$, particles $X_{t}^{1:2N}$ and weights $W_{t}^{1:2N}$ s.t. \\
\qquad \enspace $\sum_{n=1}^{2N} W_t^n \varphi(X_t^n)$ approximates $\int \varphi(x) \pi(\dd x)$\\
\qquad \enspace 
and $(2N)^{-1}\sum_{m=1}^{2N} f(2NW_t^{m})$ asymptotically upper bounds $D_f(\pi\|\mu_t)$. \\

\\
\qquad \enspace \textit{1. Initialization} \\
\qquad \qquad Generate $2N$ i.i.d. random variables $X_0^n \sim \mu_0$ for $n=1, \ldots, 2N$. \\
\qquad \qquad Set $w_0^n \gets \gamma(X_0^n)/\mu_0(X_0^n)$ for $n=1, \ldots, 2N$. \\
\qquad \qquad Normalize $W_0^n = w_0^n/\sum_{i=1}^{2N} w_0^i$ for $n=1, \ldots, 2N$.\\
\qquad \qquad Set $A_0^n = n$ for $n=1, \dots, N$.\\

\\
\qquad \enspace \textit{2. Loop} \\
\qquad \qquad \textbf{For} $t=0$ \textbf{to} $T-1$ \textbf{do} \\
\qquad \qquad \qquad Get $X_{t+1}^{1:2N}$, $W_{t+1}^{1:2N}$, and $A_{t+1}^{1:N}$ from $X_{t}^{1:2N}$, $W_{t}^{1:2N}$, and $A_{t}^{1:N}$ via Algorithm~\ref{alg:weight-harmonization}.  \\
\qquad \qquad \qquad Compute the asymptotic upper bound $\sum_{m=1}^{2N} f(2N W^m_{t+1}) / (2N)$\\
\qquad \qquad \textbf{End For} \\
\end{tabbing}
\end{fakealgoenv}
\setcounter{lastappendix}{\value{section}}

\section*{Acknowledgments}
We thank Pierre E. Jacob for his detailed and useful comments on the article and Yvann Le Fay for saving a factor 2 from being forgotten.

During the preparation of this work the authors used \emph{Google's Gemini} and \emph{Anthropic's Claude}  to assist with coding, formatting figures, and to obtain feedback on the paper. 
This was done using the AI-based coding tools \emph{GitHub Copilot} and \emph{command line interfaces}. 
Additionally, the same tools were used to provide feedback on the manuscript.
After using these services and tools, the authors reviewed and edited the content as necessary and take full responsibility for the content of the publication.

\section*{Funding}
AC acknowledges the financial support provided by UKRI for grant EP/Y014650/1, as part of the ERC Synergy project OCEAN.
HD acknowledges support from the Singaporean Ministry of Education Tier 2 grant MOE-T2EP20123-0010.

\section*{Supplementary material}
The Supplementary Material includes proofs of all results listed in the paper, background on coupling methods for MCMC kernels, as well as additional figures generated for Section~\ref{sec:exps}.
Specifically for our main results, Theorem~\ref{thm:weighted_as_diag} is proven in Supplement~\ref{app:proof-weighted-diag}, Theorem~\ref{thm:consistency} is proven in Supplement~\ref{app:proof-consistency}, and Theorem~\ref{thm:ergodic-weight-convergence} is proven in Supplement~\ref{app:proof-ergodic}.
\bibliographystyle{plainnat} %
\bibliography{references}       %

@Article{andrieu2009pseudomarginal,
    author	 = {Christophe Andrieu and Gareth O. Roberts},
    title		 = {The pseudo-marginal approach for efficient {Monte Carlo}
    computations},
    volume	 = {37},
    journal	 = {The Annals of Statistics},
    number	 = {2},
    publisher	 = {Institute of Mathematical Statistics},
    pages		 = {697 -- 725},
    year		 = {2009}
}

@Article{besag1994representations,
    title		 = {Contribution to the discussion on `{R}epresentations of
		 knowledge in complex systems' by {G}renander, {U} and
		  {M}iller, {M}.\ {I}.\@},
    author	 = {Besag, J. E.},
    journal	 = {Journal of the Royal Statistical Society: Series B
		 (Methodological)},
    volume	 = {56},
    number	 = {4},
    pages		 = {549--581},
    year		 = {1994},
    publisher	 = {Wiley Online Library}
}

@Article{bierkens2019zig,
    title		 = {{The Zig-Zag process and super-efficient sampling for
		 Bayesian analysis of big data}},
    author	 = {Bierkens, Joris and Fearnhead, Paul and Roberts, Gareth},
    journal	 = {The Annals of Statistics},
    volume	 = {47},
    number	 = {3},
    pages		 = {1288--1320},
    year		 = {2019}
}

@InProceedings{bierkens2020boomerang,
    title		 = {{The Boomerang Sampler}},
    author	 = {Bierkens, Joris and Grazzi, Sebastiano and Kamatani, Kengo
		 and Roberts, Gareth},
    booktitle	 = {Proceedings of the 37th International Conference on
		 Machine Learning},
    pages		 = {908--918},
    year		 = {2020},
    editor	 = {III, Hal Daumé and Singh, Aarti},
    volume	 = {119},
    series	 = {Proceedings of Machine Learning Research},
    month		 = {13--18 Jul},
    publisher	 = {PMLR}
}

@Article{biswas2019estimating,
    title		 = {Estimating convergence of {Markov chains with L-lag
		 couplings}},
    author	 = {Biswas, Niloy and Jacob, Pierre E. and Vanetti, Paul},
    journal	 = {Advances in Neural Information Processing Systems},
    volume	 = {32},
    year		 = {2019}
}

@Article{biswas2022coupling,
    title		 = {Coupling-based convergence assessment of some {G}ibbs
		 samplers for high-dimensional {B}ayesian regression with
		 shrinkage priors},
    author	 = {Biswas, Niloy and Bhattacharya, Anirban and Jacob, Pierre
		 E. and Johndrow, James E.},
    journal	 = {Journal of the Royal Statistical Society: Series B
		 (Statistical Methodology)},
    volume	 = {2},
    year		 = {2022}
}

@Article{bou2020coupling,
    title		 = {Coupling and convergence for {H}amiltonian {M}onte
		  {C}arlo},
    author	 = {Bou-Rabee, Nawaf and Eberle, Andreas and Zimmer, Raphael},
    journal	 = {The Annals of Applied Probability},
    volume	 = {30},
    number	 = {3},
    pages		 = {1209--1250},
    year		 = {2020},
    publisher	 = {Institute of Mathematical Statistics}
}

@Article{bouchard2018bps,
    author	 = {Alexandre Bouchard-C\^ot\'e and Sebastian J. Vollmer and
		 Arnaud Doucet},
    title		 = {The Bouncy Particle Sampler: A Nonreversible
		 Rejection-Free {Markov} Chain {Monte Carlo} Method},
    journal	 = {Journal of the American Statistical Association},
    volume	 = {113},
    number	 = {522},
    pages		 = {855-867},
    year		 = {2018},
    publisher	 = {Taylor \& Francis}
}

@book{brooks2011handbook,
    title = {Handbook of {M}arkov chain {Monte Carlo}},
    author = {Brooks, Steve and Gelman, Andrew and Jones, Galin and Meng, Xiao-Li},
    year = {2011},
    publisher = {CRC press}
}

@Book{chopin2020book,
    author	 = {Chopin, Nicolas and Papaspiliopoulos, Omiros},
    title		 = {An Introduction to Sequential {M}onte {C}arlo},
    year		 = {2020},
    publisher	 = {Springer}
}

@article{cuturi2019compOT,
url = {http://dx.doi.org/10.1561/2200000073},
year = {2019},
volume = {11},
journal = {Foundations and Trends in Machine Learning},
title = {{Computational Optimal Transport: With Applications to Data Science}},
doi = {10.1561/2200000073},
issn = {1935-8237},
number = {5-6},
pages = {355-607},
author = {Gabriel Peyré and Marco Cuturi}
}

@Article{dau2023complexity,
    author	 = {Dau, Hai-Dang and Chopin, Nicolas},
    title = {{On backward smoothing algorithms}},
    volume = {51},
    journal = {The Annals of Statistics},
    number = {5},
    publisher = {Institute of Mathematical Statistics},
    pages = {2145 -- 2169},
    year = {2023},
    doi = {10.1214/23-AOS2324},
    URL = {https://doi.org/10.1214/23-AOS2324}
}

@Book{douc2018mc,
    author	 = "Douc, Randal and Moulines, Eric and Priouret, Pierre and
		  Soulier, Philippe",
    title		 = "Markov Chains",
    year		 = "2018",
    publisher	 = "Springer International Publishing"
}

@article{Gelman1992diagnostics,
author = {Andrew Gelman and Donald B. Rubin},
title = {{Inference from Iterative Simulation Using Multiple Sequences}},
volume = {7},
journal = {Statistical Science},
number = {4},
publisher = {Institute of Mathematical Statistics},
pages = {457 -- 472},
year = {1992},
doi = {10.1214/ss/1177011136},
URL = {https://doi.org/10.1214/ss/1177011136}
}

@Article{geman1984stochastic,
    title		 = {Stochastic relaxation, {Gibbs} distributions, and the
		  {Bayesian} restoration of images},
    author	 = {Geman, Stuart and Geman, Donald},
    journal	 = {IEEE Transactions on pattern analysis and machine
		 intelligence},
    number	 = {6},
    volume	 = {PAMI-6},
    pages		 = {721--741},
    year		 = {1984},
    publisher	 = {IEEE}
}

@Article{giles2015mlmc,
    title		 = {Multilevel {Monte Carlo} methods},
    volume	 = {24},
    journal	 = {Acta Numerica},
    publisher	 = {Cambridge University Press},
    author	 = {Giles, Michael B.},
    year		 = {2015},
    pages		 = {259–328}
}

@Article{glynn2014exact,
    title		 = {Exact estimation for {M}arkov chain equilibrium
		 expectations},
    author	 = {Glynn, Peter W and Rhee, Chang-han},
    journal	 = {Journal of Applied Probability},
    volume	 = {51},
    number	 = {A},
    pages		 = {377--389},
    year		 = {2014},
    publisher	 = {Cambridge University Press}
}

@Article{goodman2009coupling,
    title		 = {Coupling control variates for {Markov chain Monte Carlo}},
    journal	 = {Journal of Computational Physics},
    volume	 = {228},
    number	 = {19},
    pages		 = {7127-7136},
    year		 = {2009},
    author	 = {Jonathan B. Goodman and Kevin K. Lin}
}

@Article{hastings1970mcmc,
    author	 = {W. K. Hastings},
    title		 = {Monte {Carlo} Sampling Methods Using {Markov} Chains and
		 Their Applications},
    journal	 = {Biometrika},
    year		 = {1970},
    volume	 = {57(1)},
    pages		 = {97--109}
}

@Article{heng2019couplings,
    author	 = {Heng, J and Jacob, P E},
    title		 = "{Unbiased Hamiltonian Monte Carlo with couplings}",
    journal	 = {Biometrika},
    volume	 = {106},
    number	 = {2},
    pages		 = {287-302},
    year		 = {2019},
    month		 = {02}
}

@Book{huber2016perfect,
    title		 = {Perfect simulation},
    author	 = {Huber, Mark L.},
    volume	 = {148},
    year		 = {2016},
    publisher	 = {CRC Press}
}

@Article{jacob2020unbiasedmcmc,
    author	 = {Jacob, Pierre E. and O’Leary, John and Atchad\'{e}, Yves
		 F.},
    title		 = {Unbiased {M}arkov chain {M}onte {C}arlo methods with
		 couplings},
    journal	 = {Journal of the Royal Statistical Society: Series B
		 (Statistical Methodology)},
    volume	 = {82},
    number	 = {3},
    pages		 = {543-600},
    year		 = {2020}
}

@Article{jacob2020unbiasedsmoothing,
    author	 = {Pierre E. Jacob and Fredrik Lindsten and Thomas B.
		 Sch\"on},
    title		 = {Smoothing With Couplings of Conditional Particle Filters},
    journal	 = {Journal of the American Statistical Association},
    volume	 = {115},
    number	 = {530},
    pages		 = {721-729},
    year		 = {2020},
    publisher	 = {Taylor \& Francis}
}

@article{Johnson1996convergence,
author = {Valen E. Johnson},
title = {Studying Convergence of {Markov} Chain {Monte Carlo} Algorithms Using Coupled Sample Paths},
journal = {Journal of the American Statistical Association},
volume = {91},
number = {433},
pages = {154--166},
year = {1996},
publisher = {ASA Website},
doi = {10.1080/01621459.1996.10476672},


URL = { 
    
    
        https://www.tandfonline.com/doi/abs/10.1080/01621459.1996.10476672
    

},
eprint = { 
    
    
        https://www.tandfonline.com/doi/pdf/10.1080/01621459.1996.10476672
    

}

}

@article{karjalainen2023mixing,
    author = {Karjalainen, Joona and Lee, Anthony and Singh, Sumeetpal S and Vihola, Matti},
    title = {Mixing time of the conditional backward sampling particle filter},
    journal = {Journal of the Royal Statistical Society Series B: Statistical Methodology},
    pages = {qkaf078},
    year = {2025},
    month = {12},
    issn = {1369-7412},
    doi = {10.1093/jrsssb/qkaf078},
}

@article{kim1998stochvol,
    author = {Kim, Sangjoon and Shephard, Neil and Chib, Siddhartha},
    title = {Stochastic Volatility: Likelihood Inference and Comparison with ARCH Models},
    journal = {The Review of Economic Studies},
    volume = {65},
    number = {3},
    pages = {361-393},
    year = {1998},
    month = {07},
    issn = {0034-6527},
    doi = {10.1111/1467-937X.00050}
}

@Article{kong1994sequential,
    author	 = {Kong, Augustine and Liu, Jun S. and Wong, Wing Hung},
    journal	 = {Journal of the American Statistical Association},
    number	 = {425},
    pages		 = {278--288},
    title		 = {Sequential imputations and {B}ayesian missing data
		 problems},
    url		 = {http://www.jstor.org/stable/2291224},
    volume	 = {89},
    doi		 = {https://doi.org/10.1080/01621459.1994.10476469},
    year		 = {1994}
}

@Article{lee2020coupled,
    title		 = {Coupled conditional backward sampling particle filter},
    author	 = {Lee, Anthony and Singh, Sumeetpal S and Vihola, Matti},
    journal	 = {Annals of Statistics},
    volume	 = {48},
    number	 = {5},
    pages		 = {3066--3089},
    year		 = {2020},
    doi		 = {10.1214/19-AOS1922},
    publisher	 = {Institute of Mathematical Statistics}
}

@Article{lindvall1986lectures,
    author	 = {Torgny Lindvall and L. C. G. Rogers},
    title		 = {{Coupling of Multidimensional Diffusions by Reflection}},
    volume	 = {14},
    journal	 = {The Annals of Probability},
    number	 = {3},
    publisher	 = {Institute of Mathematical Statistics},
    pages		 = {860 -- 872},
    year		 = {1986}
}

@Article{metropolis1953equation,
    author	 = {Nicholas Metropolis and Arianna W. Rosenbluth and Marshall
		 N. Rosenbluth and Augusta H. Teller and Edward Teller},
    title		 = {Equation of state calculations by fast computing
		 machines},
    journal	 = {Journal of Chemical Physics},
    volume	 = {21},
    number	 = {6},
    pages		 = {1087--1092},
    doi		 = {10.1063/1.1699114},
    year		 = {1953}
}

@book{meyn2012markov,
    title = {Markov chains and stochastic stability},
    author = {Meyn, Sean P and Tweedie, Richard L},
    year = {2009},
    publisher = {Cambridge University Press}
}

@InProceedings{middleton2019unbiased,
    title		 = {Unbiased smoothing using particle independent
		  {M}etropolis-{H}astings},
    author	 = {Middleton, Lawrence and Deligiannidis, George and Doucet,
		 Arnaud and Jacob, Pierre E.},
    booktitle	 = {The 22nd International Conference on Artificial
		 Intelligence and Statistics},
    pages		 = {2378--2387},
    year		 = {2019},
    organization	 = {PMLR}
}

@Article{middleton2020unbiased,
    title		 = {Unbiased {M}arkov chain {M}onte {C}arlo for intractable
		 target distributions},
    author	 = {Middleton, Lawrence and Deligiannidis, George and Doucet,
		 Arnaud and Jacob, Pierre E.},
    journal	 = {Electronic Journal of Statistics},
    volume	 = {14},
    number	 = {2},
    pages		 = {2842--2891},
    year		 = {2020},
    publisher	 = {Institute of Mathematical Statistics and Bernoulli
		 Society}
}

@InCollection{neal:2011,
    author	 = {Radford M. Neal},
    year		 = {2011},
    title		 = {{MCMC} using {H}amiltonian dynamics},
    booktitle	 = {Handbook of {M}arkov Chain {M}onte {C}arlo},
    publisher	 = {Chapman \& Hall/CRC},
    chapter	 = {5},
    pages = {113--163}
}

@InProceedings{nguyen2022many,
    title		 = {Many processors, little time: {MCMC} for partitions via
		 optimal transport couplings},
    author	 = {Nguyen, Tin D and Trippe, Brian L and Broderick, Tamara},
    booktitle	 = {International Conference on Artificial Intelligence and
		 Statistics},
    pages		 = {3483--3514},
    year		 = {2022},
    organization	 = {PMLR}
}

@article{Papp2024couplings,
    author = {Papp, Tamás P and Sherlock, Chris},
    title = {Scalable couplings for the random walk Metropolis algorithm},
    journal = {Journal of the Royal Statistical Society Series B: Statistical Methodology},
    pages = {qkae113},
    year = {2024},
    month = {12},
    issn = {1369-7412},
    doi = {10.1093/jrsssb/qkae113},
    url = {https://doi.org/10.1093/jrsssb/qkae113},
    eprint = {https://academic.oup.com/jrsssb/advance-article-pdf/doi/10.1093/jrsssb/qkae113/61200648/qkae113.pdf},
}

@Article{propp1996exact,
    title		 = {Exact sampling with coupled {M}arkov chains and
		 applications to statistical mechanics},
    author	 = {Propp, James Gary and Wilson, David Bruce},
    journal	 = {Random Structures \& Algorithms},
    volume	 = {9},
    number	 = {1-2},
    pages		 = {223--252},
    year		 = {1996},
    publisher	 = {Wiley Online Library}
}

@inproceedings{renyi1961measures,
    title = {On measures of entropy and information},
    author = {R{\'e}nyi, Alfr{\'e}d},
    booktitle = {Proceedings of the Fourth Berkeley Symposium on Mathematical Statistics and Probability, Volume 1: Contributions to the Theory of Statistics},
    volume = {4},
    pages = {547--562},
    year = {1961},
    organization = {University of California Press}
}

@Article{rhee2015unbiased,
    title		 = {Unbiased estimation with square root convergence for {SDE}
    models},
    author	 = {Rhee, Chang-Han and Glynn, Peter W.},
    journal	 = {Operations Research},
    volume	 = {63},
    number	 = {5},
    pages		 = {1026--1043},
    year		 = {2015},
    publisher	 = {INFORMS}
}

@Article{roberts1998optimal,
    title		 = {Optimal scaling of discrete approximations to {L}angevin
		 diffusions},
    author	 = {Roberts, Gareth O and Rosenthal, Jeffrey S},
    journal	 = {Journal of the Royal Statistical Society: Series B
		 (Statistical Methodology)},
    volume	 = {60},
    number	 = {1},
    pages		 = {255--268},
    year		 = {1998},
    publisher	 = {Wiley Online Library}
}

@Article{vehtari2021rank,
    author	 = {Aki Vehtari and Andrew Gelman and Daniel Simpson and Bob
		 Carpenter and Paul-Christian B{\"u}rkner},
    title		 = {{Rank-Normalization, Folding, and Localization: An
		 Improved $\widehat{R}$ for Assessing Convergence of MCMC
		 (with Discussion)}},
    volume	 = {16},
    journal	 = {Bayesian Analysis},
    number	 = {2},
    publisher	 = {International Society for Bayesian Analysis},
    pages		 = {667 -- 718},
    year		 = {2021},
    doi		 = {10.1214/20-BA1221},
    url		 = {https://doi.org/10.1214/20-BA1221}
}

@Article{vihola2018mlmc,
    author	 = {Vihola, Matti},
    title		 = {Unbiased Estimators and Multilevel {Monte Carlo}},
    journal	 = {Operations Research},
    volume	 = {66},
    number	 = {2},
    pages		 = {448-462},
    year		 = {2018}
}

@Book{villani2009optimal,
    title		 = {Optimal transport: old and new},
    author	 = {Villani, C{\'e}dric},
    volume	 = {338},
    year		 = {2009},
    publisher	 = {Springer}
}

@InProceedings{wang2021coupling,
    title		 = { Maximal Couplings of the Metropolis-Hastings Algorithm },
    author	 = {Wang, Guanyang and O'Leary, John and Jacob, Pierre},
    booktitle	 = {Proceedings of The 24th International Conference on
		 Artificial Intelligence and Statistics},
    pages		 = {1225--1233},
    year		 = {2021},
    editor	 = {Banerjee, Arindam and Fukumizu, Kenji},
    volume	 = {130},
    series	 = {Proceedings of Machine Learning Research},
    month		 = {13--15 Apr},
    publisher	 = {PMLR}
}

@article{vats2021,
author = {Dootika Vats and Christina Knudson},
title = {{Revisiting the Gelman--Rubin Diagnostic}},
volume = {36},
journal = {Statistical Science},
number = {4},
publisher = {Institute of Mathematical Statistics},
pages = {518 -- 529},
year = {2021},
doi = {10.1214/20-STS812},
URL = {https://doi.org/10.1214/20-STS812}
}

@misc{ceriani2024linearcostunbiasedposteriorestimates,
      title={Linear-cost unbiased posterior estimates for crossed effects and matrix factorization models via couplings}, 
      author={Paolo Maria Ceriani and Giacomo Zanella},
      year={2024},
      eprint={2410.08939},
      archivePrefix={arXiv},
      primaryClass={stat.CO},
      url={https://arxiv.org/abs/2410.08939}, 
}

@article{corenflos2025pdmp,
author = {Corenflos, Adrien and Sutton, Matthew and Chopin, Nicolas},
title = {Debiasing piecewise deterministic Markov process samplers using couplings},
journal = {Scandinavian Journal of Statistics},
volume = {n/a},
number = {n/a},
pages = {1--43},
year=2025,
doi = {https://doi.org/10.1111/sjos.70015},
url = {https://onlinelibrary.wiley.com/doi/abs/10.1111/sjos.70015},
eprint = {https://onlinelibrary.wiley.com/doi/pdf/10.1111/sjos.70015}
}

@article{craiu2022double,
  title={{Double Happiness: Enhancing the Coupled Gains of L-lag Coupling via Control Variates}},
  author={Craiu, Radu V and Meng, Xiao-Li},
  journal={Statistica Sinica},
  volume={32},
  number={4},
  pages={1745--1766},
  year={2022},
  publisher={Academia Sinica}
}

@book{elementsinformationtheory,
author = {Cover, Thomas M. and Thomas, Joy A.},
title = {Elements of Information Theory},
year = {2005},
isbn = {9780471241959},
publisher = {John Wiley \& Sons, Ltd},
address = {USA}
}

@article{Kahn1951SplittingParticleTransmission,
  author = {Kahn, Herman and Harris, Theodore E},
  journal = {{National Bureau of Standards applied mathematics series}},
  pages = {27--30},
  title = {Estimation of particle transmission by random sampling},
  volume = 12,
  year = 1951
}

@article{Polson01122013,
author = {Nicholas G. Polson and James G. Scott and Jesse Windle},
title = {Bayesian Inference for Logistic Models Using Pólya–Gamma Latent Variables},
journal = {Journal of the American Statistical Association},
volume = {108},
number = {504},
pages = {1339--1349},
year = {2013},
publisher = {ASA Website},
doi = {10.1080/01621459.2013.829001},


URL = { 
    
        https://doi.org/10.1080/01621459.2013.829001
    
    

},
eprint = { 
    
        https://doi.org/10.1080/01621459.2013.829001
    
    

}

}

@incollection{atchade_unbiased,
  title     = {Unbiased {MCMC}},
  author    = {Atchad{\'e}, Yves F. and Jacob, Pierre E.},
  booktitle = {Handbook of {Markov Chain Monte Carlo}},
  editor    = {Craiu, Radu V. and Vats, Dootika and Jones, Galin and Brooks, Steve and Gelman, Andrew and Meng, Xiao-Li},
  edition   = {2nd},
  publisher = {Chapman \& Hall/CRC},
  year      = {2026},
  chapter   = {17}
}

@inproceedings{
srinivasan2025the,
title={The Polynomial {S}tein Discrepancy for Assessing Moment Convergence},
author={Narayan Srinivasan and Matthew Sutton and Christopher Drovandi and Leah F South},
booktitle={Forty-second International Conference on Machine Learning},
year={2025},
url={https://openreview.net/forum?id=2QaqxseJYT}
}

@inproceedings{gorham2017measuring,
  title={Measuring sample quality with kernels},
  author={Gorham, Jackson and Mackey, Lester},
  booktitle={International Conference on Machine Learning},
  pages={1292--1301},
  year={2017},
  organization={PMLR}
}

@article{gorham2015measuring,
  title={Measuring sample quality with {S}tein's method},
  author={Gorham, Jackson and Mackey, Lester},
  journal={Advances in neural information processing systems},
  volume={28},
  year={2015}
}

@Article{phan2026MMC,
    author	 = {Buu Phan and Gergely Flamich and Ashish Khisti and Shahab Asoodeh},
    title		 = {Multi-Marginal Couplings for {M}etropolis-{H}astings},
    journal	 = {arXiv preprint arXiv:2605.12807},
    year		 = 2026,
}

\clearpage
\begin{center}
{\Large\bfseries Supplementary Material for\\[0.6em]
General coupling-based $f$-divergence diagnostics for Markov chain Monte Carlo\par}
\vspace{1em}
Adrien Corenflos and Hai-Dang Dau
\end{center}

\vspace{2em}

\setcounter{section}{\value{lastappendix}}
\renewcommand{\sectionname}{Supplement}

\textbf{Notations.} We write $\mathcal{F}_t$ for the filtration generated by all particles, weights, and pairings up to time $t$ (including the shuffling step before Markov kernels at iteration $t+1$ are applied).
When a particle $X^n_t$ is paired with another particle $X^{A^n_t + N}$ in Algorithm~\ref{alg:weight-harmonization}, we write 
\[m^n_t = A^n_t + N\]
for short, dropping the $t$ index when context is clear.
Other notations are defined in Section~\ref{subsec:notations} of the main text.

\section{Proof of Theorem~\ref{thm:weighted_as_diag} (Divergence bounds from weights)}\label{app:proof-weighted-diag}
\begin{lemma}
\label{lem:convexity}
Suppose that the function $f: ]0, \infty[ \to ]-\infty, \infty[$ is convex with $f(1) = 0$.
Write $f'(x)$ for the one-sided limit
\[
f'(x):=\lim_{h \to 0^+} \frac{f(x+h) - f(x)}{h}.
\]
This quantity is well-defined and finite for all $x$.
Moreover, for all $x$ and $y$,
\begin{equation}
\label{eq:convex_tangent}
f(y) \geq f(x) + f'(x)(y-x).
\end{equation}
\end{lemma}
\begin{proof}
For all $y_1 < x < y_2 < z$, convexity of $f$ implies
\begin{equation}
\label{eq:convexbridge}
\frac{f(x) - f(y_1)}{x - y_1} \leq
\frac{f(y_2) - f(x)}{y_2 - x} \leq \frac{f(z) - f(x)}{z-x}.
\end{equation}
Thus, the function $L_x: h \mapsto h^{-1}[f(x+h) - f(x)]$ is lower bounded,
\begin{equation}
\label{eq:convex_L}
\frac{f(x) - f(y_1)}{x - y_1} \leq L_x(h),
\end{equation}
for $h > 0$.
Moreover, by~\eqref{eq:convexbridge} again, $L_x(h)$ is finite and nondecreasing on $]0, \infty[$, entailing the existence and finiteness of its right limit at $0$, noted by $f'(x)$.
Additionally, $L_x(h) \geq f'(x), \forall h > 0$, confirming~\eqref{eq:convex_tangent} for $y > x$.
Finally, taking $h \to 0^+$ in~\eqref{eq:convex_L} establishes~\eqref{eq:convex_tangent} for $y<x$.
\end{proof}
\begin{proof}[Proof of Theorem~\ref{thm:weighted_as_diag}]
The proof relies on the convexity of the function $f$ that defines the divergence. 
Let $Q_M = \sum_{m=1}^M f(MW^m) / M$ be our estimator.
Recall that we write $\psi$ for the Radon-Nikodym derivative $\dd \pi/\dd \mu$.
We have
\begin{equation}
\label{eq:decomp}
\begin{split}
    Q_M - D_f(\pi \| \mu) 
        &= \frac 1M \sum_{m=1}^M f(MW^m) - \int f(\psi(x)) \mu(\dd x) \\
        &= \left(\frac 1M \sum_{m=1}^M \left[f(MW^m) - f(\psi(X^m))\right]\right) + \\
        &+\left[\frac{1}{M} \sum_{m=1}^M f(\psi(X^m)) - \int f(\psi(x)) \mu(\dd x)\right] 
\end{split}
\end{equation}
Owing to the convexity of $f$ and Lemma~\ref{lem:convexity}, the first term in~\eqref{eq:decomp} can be bounded from below:
\begin{align*}
    \frac 1M \sum_{m=1}^M \left[f(MW^m) - f(\psi(X^m))\right]
        &\geq \frac{1}{M} \sum_{m=1}^M f'(\psi(X^m))\{M W^m - \psi(X^m)\}\\
        &= \left[\sum_{m=1}^M W^m f'(\psi(X^m))\right] - \left[\frac 1M \sum_{m=1}^M \psi(X^m) f'(\psi(X^m))\right]
\end{align*}
Let us denote the two terms on the right-hand side as $A_M$ and $B_M$ respectively.
By Assumption~\ref{ass:weighted-convergence}, $A_M \to \int f'(\psi(x)) \pi(\dd x)$ in probability.
By Assumption~\ref{ass:unweighted-LLN}, $B_M \to \int \psi(x) f'\{\psi(x)\} \mu(\dd x)$ in probability.

Since $\pi(\dd x) = \psi(x)\mu(\dd x)$, the limits are identical:
\[
\int f'(\psi(x)) \pi(\dd x) = \int f'(\psi(x)) \psi(x) \mu(\dd x).
\]
Therefore, the difference $A_M - B_M \to 0$ in probability.

Now consider the second term in~\eqref{eq:decomp} and call it $C_M$:
\[
C_M = \frac{1}{M} \sum_{m=1}^M f(\psi(X^m)) - \int f(\psi(x)) \mu(\dd x).
\]
By Assumption~\ref{ass:unweighted-LLN}, $C_M \to 0$ in probability.

Combining these results, we have a lower bound for the total expression:
\[
Q_M - D_f(\pi \| \mu) \geq (A_M - B_M) + C_M.
\]
Let $L_M = (A_M - B_M) + C_M$. 
Since $A_M - B_M \to 0$ in probability and $C_M \to 0$ in probability, their sum $L_M$ also converges to $0$ in probability.

This means that for any $\varepsilon > 0$, $\Prob(\abs{L_M} \geq \varepsilon/2) \to 0$.
The event $\{Q_M - D_f(\pi \| \mu) \leq -\varepsilon\}$ is a subset of the event $\{L_M \leq -\varepsilon\}$, because $Q_M - D_f(\pi \| \mu) \geq L_M$.
Thus,
\[
\Prob\paren{Q_M - D_f(\pi\|\mu) \leq -\varepsilon} \leq \Prob\paren{L_M \leq -\varepsilon}.
\]
The event $\{L_M \leq -\varepsilon\}$ is itself a subset of $\{\abs{L_M} \geq \varepsilon\}$.
Therefore,
\[
\Prob\paren{Q_M - D_f(\pi\|\mu) \leq -\varepsilon} \leq \Prob\paren{\abs{L_M} \geq \varepsilon}.
\]
Since $L_M \to 0$ in probability, the right-hand side tends to $0$ as $M \to \infty$.
\end{proof}

\section{Achievable improvement of the divergence bound}
\label{app:thm1bis}
Throughout this section, for a set of strictly positive weights $W^{1:2N}$ summing to one, we write
\begin{equation}\label{eq:diagnostic-shorthand}
    \mathcal{D}_f(W^{1:2N}) = \frac{1}{2N}\sum_{m=1}^{2N} f(2N W^m)
\end{equation}
for the diagnostic quantity of Lemma~\ref{lem:empirical-div}, and, for $u, v > 0$, we let
\begin{equation}\label{eq:jensen-gap}
	\Delta_f(u, v) = \frac{f(u) + f(v)}{2} - f\paren*{\frac{u + v}{2}}
\end{equation}
be the Jensen gap of $f$ over $\{u, v\}$, which is symmetric in its two arguments, non-negative by convexity of $f$, and vanishes when $u = v$.
\begin{theorem}\label{thm:weight-harmonization}
	Let $W^{1:2N}$ be a set of strictly positive weights summing to one, and, for some $j \in \curly*{1, \ldots, N}$, let $\tilde{W}^{1:2N}$ be defined as $\tilde{W}^n = W^n$, $n \in \curly*{1, \ldots, 2N} \setminus \{j, j+N\}$, and $\tilde{W}^j = \tilde{W}^{j+N} = (W^j + W^{j+N}) / 2$.
	Then, for any convex $f$, we have
	\begin{equation}\label{eq:harmonization-identity}
	\mathcal{D}_f(\tilde{W}^{1:2N}) = \mathcal{D}_f(W^{1:2N}) - \frac{1}{N} \Delta_f\paren*{2N W^j, 2N W^{j+N}}.
	\end{equation}
	In particular, 
	\begin{alignat}{2}
	0 &\leq \mathcal{D}_f(\tilde{W}^{1:2N}) &&\leq \mathcal{D}_f({W}^{1:2N})
	\end{alignat}
	and the right inequality is strict if $f$ is strictly convex and $W^j \neq W^{j+N}$.
	Additionally, writing $W_* = \min (W^{1:2N})$ and $W^* = \max(W^{1:2N})$, we have the tighter lower bound
	\begin{align}\label{eq:upper-bound-improvement}
	\mathcal{D}_f(\tilde{W}^{1:2N}) \geq \mathcal{D}_f(W^{1:2N}) - \frac{1}{N} \Delta_f\paren*{2N W_*, 2N W^*}.
	\end{align}
\end{theorem}
\begin{remark}\label{rem:alpha-opt}
    If the general updates
    $\tilde W^j = \alpha W^j + (1-\alpha) W^{j+N}$ and $\tilde W^{j+N} = (1-\alpha)W^j + \alpha W^{j+N}$ are used,
	we have
	\begin{equation}
	\mathcal{D}_f(\tilde{W}^{1:2N}) = \mathcal{D}_f(W^{1:2N}) - \frac{1}{N} \Delta^\alpha_f\paren*{2N W^j, 2N W^{j+N}}
	\end{equation}
	for the asymmetric gap $\Delta^\alpha_f(u, v) = \curly*{f(u) + f(v)}/2 - \curly*{f(\alpha u + (1-\alpha) v) + f((1-\alpha) u + \alpha v)}/2$, which recovers~\eqref{eq:jensen-gap} when $\alpha = 1/2$.
    The map $\alpha \mapsto f(\alpha u + (1-\alpha) v) + f((1-\alpha) u + \alpha v)$ is convex, as a sum of convex functions of affine maps, and symmetric about $\alpha = 1/2$; it is therefore minimized at $\alpha = 1/2$.
    Therefore, choosing $\alpha = 1/2$ maximizes the improvement, and does so simultaneously for every convex $f$.
\end{remark}
\begin{remark}
It is easy to see that $\min(\tilde W^{1:2N}) \geq \min(W^{1:2N})$ and $\max(\tilde W^{1:2N}) \leq \max(W^{1:2N})$, that is, the range of the weights can only shrink, which means that the upper bound on the improvement in~\eqref{eq:upper-bound-improvement} decreases as more harmonization operations are performed.
\end{remark}
Because Algorithm~\ref{alg:weight-harmonization} harmonizes at most $N$ disjoint pairs at a given iteration, and each of them may only contribute an improvement of order $1/N$, the gain achieved over a full iteration is itself bounded, uniformly in time, by the spread of the initial weights.
\begin{corollary}\label{cor:sweep-improvement}
	Let $W_t^{1:2N}$ be the normalized weights at step $t$ of Algorithm~\ref{alg:weight-harmonization}, and let
	\begin{equation}\label{eq:lambda-f}
	\bar\lambda_f = \Delta_f\paren*{2N \min(W_0^{1:2N}), 2N \max(W_0^{1:2N})}.
	\end{equation}
	Then, for any convex $f$, and almost surely for all $t \geq 0$,
	\begin{equation}\label{eq:sweep-improvement}
	\mathcal{D}_f(W_{t+1}^{1:2N}) \geq \mathcal{D}_f(W_{t}^{1:2N}) - \bar\lambda_f.
	\end{equation}
\end{corollary}
The proofs of both results rely on the following elementary property of the Jensen gap~\eqref{eq:jensen-gap}: it can only increase when its arguments are pulled apart.
\begin{lemma}\label{lem:jensen-gap-range}
	Let $f$ be convex on $]0, \infty[$ and let $0 < a \leq u \leq v \leq b$. Then $\Delta_f(u, v) \leq \Delta_f(a, b)$.
\end{lemma}
\begin{proof}
	We first show that, for any fixed $v > 0$, the map $u \mapsto \Delta_f(u, v)$ is non-increasing on $]0, v]$.
	Let $0 < u_1 < u_2 \leq v$ and write $m_i = (u_i + v)/2$, so that $m_2 - m_1 = (u_2 - u_1)/2$, $u_1 \leq m_1$, and $u_2 \leq m_2$.
	Then
	\begin{align*}
	\Delta_f(u_1, v) - \Delta_f(u_2, v) &= \frac{f(u_1) - f(u_2)}{2} + f(m_2) - f(m_1)\\
	&= \frac{u_2 - u_1}{2}\curly*{\frac{f(m_2) - f(m_1)}{m_2 - m_1} - \frac{f(u_2) - f(u_1)}{u_2 - u_1}} \geq 0,
	\end{align*}
	where the inequality holds because, $f$ being convex, the slope of its chord over $[m_1, m_2]$ is at least that of its chord over $[u_1, u_2]$ whenever $u_1 \leq m_1$ and $u_2 \leq m_2$.
	By symmetry of $\Delta_f$, the map $v \mapsto \Delta_f(u, v)$ is likewise non-decreasing on $[u, \infty[$ for any fixed $u > 0$.
	Applying the first statement to $u$ and $a$, and then the second one to $v$ and $b$, gives $\Delta_f(u, v) \leq \Delta_f(a, v) \leq \Delta_f(a, b)$.
\end{proof}
\begin{proof}[Proof of Theorem~\ref{thm:weight-harmonization}]
    The identity~\eqref{eq:harmonization-identity} follows from direct computation: the terms indexed by $m \notin \{j, j+N\}$ are left untouched by the operation, so that
  \begin{align*}
    \mathcal{D}_f(\tilde{W}^{1:2N}) - \mathcal{D}_f(W^{1:2N}) 
      &= \frac{1}{2N}\curly*{2 f\paren*{\frac{2N W^j + 2N W^{j+N}}{2}} - f(2N W^j) - f(2N W^{j+N})}\\
      &= - \frac{1}{N}\Delta_f\paren*{2N W^j, 2N W^{j+N}}.
  \end{align*}
  Because $f$ is convex, $\Delta_f \geq 0$, and it is clear that $\mathcal{D}_f(\tilde{W}^{1:2N}) \leq \mathcal{D}_f(W^{1:2N})$, with a strict inequality when $f$ is strictly convex and $W^j \neq W^{j+N}$.
  The left inequality is Jensen's inequality: the harmonized weights still sum to one, so that
  \begin{align*}
    \mathcal{D}_f(\tilde{W}^{1:2N}) = \frac{1}{2N}\sum_{m=1}^{2N} f(2N\tilde{W}^m) \geq f\paren*{\frac{1}{2N}\sum_{m=1}^{2N} 2N \tilde{W}^m} = f(1) = 0.
  \end{align*}
  Finally, both $W^j$ and $W^{j+N}$ belong to $[W_*, W^*]$, so that Lemma~\ref{lem:jensen-gap-range} gives $\Delta_f(2N W^j, 2N W^{j+N}) \leq \Delta_f(2N W_*, 2N W^*)$, and~\eqref{eq:upper-bound-improvement} follows from~\eqref{eq:harmonization-identity}.
\end{proof}
\begin{proof}[Proof of Corollary~\ref{cor:sweep-improvement}]
    At step $t+1$, Algorithm~\ref{alg:weight-harmonization} harmonizes the weights of the pairs $(n, m^n_t)$, $n \in C_{t+1}$, where $C_{t+1} \subseteq \{1, \ldots, N\}$ is the random set of pairs that coupled, and leaves all the other weights unchanged.
    These pairs are disjoint, so that, up to a relabelling of the particles, applying~\eqref{eq:harmonization-identity} to each of them in turn gives
    \begin{align*}
    \mathcal{D}_f(W_{t+1}^{1:2N}) = \mathcal{D}_f(W_{t}^{1:2N}) - \frac{1}{N}\sum_{n \in C_{t+1}} \Delta_f\paren*{2N W_t^n, 2N W_t^{m^n_t}}.
    \end{align*}
    Writing $W_{t,*} = \min(W_t^{1:2N})$ and $W_t^* = \max(W_t^{1:2N})$, each of the $\abs{C_{t+1}} \leq N$ terms of the sum is bounded by $\Delta_f(2N W_{t,*}, 2N W_t^*)$ by Lemma~\ref{lem:jensen-gap-range}.
    Moreover, harmonization replaces two weights by their average, so that $W_{t,*} \geq \min(W_0^{1:2N})$ and $W_t^* \leq \max(W_0^{1:2N})$, and a second application of Lemma~\ref{lem:jensen-gap-range} gives $\Delta_f(2N W_{t,*}, 2N W_t^*) \leq \bar\lambda_f$.
\end{proof}

\section{Expectation invariance and Radon-Nikodym interpretation}
\label{app:proof-invariance}
\begin{proof}[Proof of Proposition~\ref{prop:expectation-invariance}]
The proof follows from induction. 
The base case $t=0$ is immediate by definition.
We first compute the conditional expectation of the sum of un-normalized estimators, $\sum_{i=1}^{2N} w_{t+1}^i \varphi(X_{t+1}^i)$, given the history $\mathcal{F}_t$.
For all $m \in 1, \ldots, 2N$, the law of $X_{t+1}^m$ given $\mathcal F_t$ is $K(X_t^m, \dd x_{t+1})$, regardless of the specific coupling.
Therefore
\begin{equation}
    \label{eq:proof-invariance-ft}
    \Exp\brac*{w_t^n \varphi(X_{t+1}^n) + w_t^{m_t^n} \varphi(X_{t+1}^{m_t^n}) \mid \mathcal F_t} = w_t^n (K\varphi)(X_t^n) + w_t^{m_t^n} (K\varphi)({X_t^{m_t^n}}).
\end{equation}
Now, the key observation is that
\begin{equation}
    \label{eq:proof-invariance-conservation}
    w_{t+1}^n \varphi(X_{t+1}^n) + w_{t+1}^{m_t^n} \varphi(X_{t+1}^{m_t^n}) =
    w_t^n \varphi(X_{t+1}^n) + w_t^{m_t^n} \varphi(X_{t+1}^{m_t^n}).
\end{equation}
This is thanks to the harmonization mechanism.
When $X_{t+1}^n \neq X_{t+1}^{m_t^n}$ (no coupling), we simply keep the weights: $w_{t+1}^n = w_{t}^n$ and $w_{t+1}^{m_t^n} = w_{t}^{m_t^n}$.
When $X_{t+1}^n = X_{t+1}^{m_t^n} = x^*$ for some $x^*$ (coupling successful), we set 
\[w_{t+1}^n = w_{t+1}^{m_t^n} = (w_t^n + w_t^{m_t^n})/2\]
and both sides of~\eqref{eq:proof-invariance-conservation} are then equal to $(w_t^n + w_t^{m_t^n}) \varphi(x^*)$.
Combining~\eqref{eq:proof-invariance-ft} and~\eqref{eq:proof-invariance-conservation}, we have
\[
    \Exp\brac*{w_{t+1}^n \varphi(X_{t+1}^n) + w_{t+1}^{m_t^n} \varphi(X_{t+1}^{m_t^n}) \mid \mathcal F_t} = w_t^n (K\varphi)(X_t^n) + w_t^{m_t^n} (K\varphi)({X_t^{m_t^n}})
\]
and summing over all $n = 1, \ldots, N$ gives
\[
    \Exp\brac{\hat I_{t+1,N}(\varphi)\mid \mathcal F_t}
    =\Exp\brac*{\sum_{m=1}^{2N} w_{t+1}^m \varphi(X_{t+1}^m) \mid \mathcal F_t} = \sum_{m=1}^{2N} w_t^m (K\varphi)(X_t^m) = \hat I_{t,N}(K\varphi),
\]
which, by the law of total expectation, yields
\begin{align*}
\Exp[\hat{I}_{t+1,N}(\varphi)] = \Exp[\hat{I}_{t,N}(K\varphi)].
\end{align*}
For the inductive step, assume $\Exp\{\hat{I}_{t,N}(\psi)\} = 2N\gamma(\psi)$ for any bounded function $\psi$, where $\gamma = Z\pi$.
Applying the induction hypothesis with $\psi=K\varphi$ gives $\Exp\{\hat{I}_{t,N}(K\varphi)\} = 2N\gamma(K\varphi)$.
Since $\gamma$ is an invariant measure for $K$,
\begin{align*}
\gamma(K\varphi) = \int (K\varphi)(x) \gamma(\dd x) = \int \varphi(x) (\gamma K)(\dd x) = \int \varphi(x) \gamma(\dd x) = \gamma(\varphi)
\end{align*}
and the proof is finished.
\end{proof}

\begin{proof}[Proof of Corollary~\ref{coroll:unbiased}]
From Proposition~\ref{prop:expectation-invariance} and Assumption~\ref{asp:kernel_symmetric}, we have by symmmetry
\[
\Exp\brac{w_t^m \varphi(X_t^m)} = \int \varphi(x) \gamma(\dd x)
\]
for any index $m$ and any bounded continuous function $\varphi$.
Letting $r_t$ be the function defined by $r_t(X_t^m) = \Exp[w_t^m | X_t^m]$, we have
\begin{equation*}
    \begin{split}
\int r_t(x) \varphi(x) \mu_t(\dd x) &= \Exp[r_t(X_t^m)\varphi(X_t^m)] \\
&= \Exp[w_t^m \varphi(X_t^m)] \textrm{ (by law of total expectation)}\\
&= \int \varphi(x) \gamma(\dd x) = \int \frac{\dd \gamma}{\dd\mu_t} (x)\varphi(x) \mu_t(\dd x).
    \end{split}
\end{equation*}
This identity holds for any bounded function $\varphi$, concluding the proof.
\end{proof}

\section{Proof of Theorem~\ref{thm:consistency} (Consistency)}\label{app:proof-consistency}

\begin{lemma}\label{lem:jensen-kernel}
	For any $\alpha \geq 1$ and any $\pi$-invariant kernel $K$, if $\Exp_\pi\curly{|\varphi(X)|^{\alpha}} < \infty$ then
	\[\Exp_\pi[\abs{(K\varphi)(X)}^{\alpha}] < \infty.\]
\end{lemma}
\begin{proof}
	By Jensen's inequality
	\begin{align*}
		\Exp_\pi[\abs{(K\varphi)(X)}^{\alpha}] &=\int
		\abs*{\int \varphi(y)K(x, \dd y)}^{\alpha} \pi(\dd x)
		\leq \iint |\varphi(y)|^{\alpha} K(x, \dd y) \pi(\dd x).
	\end{align*}
	The $\pi$-invariance of the kernel then concludes the proof.
\end{proof}

\begin{lemma}\label{lem:jensen-coupling}
	For any non-negative function $\varphi$, we have
	\[ \Exp\brac*{\sum_{m=1}^{2N} (w_t^m)^2 \varphi(X_t^m)}\leq \Exp\brac*{\sum_{m=1}^{2N} (w_{t-1}^m)^2 (K\varphi)(X_{t-1}^m)}. \]
\end{lemma}
\begin{proof}
	Write
\begin{align*}
	\sum_{m=1}^{2N} (w_t^m)^2 \varphi(X_t^m)
	&= \sum_{n=1}^{N} \brac*{(w_t^n)^2 \varphi(X_t^n) +  (w_t^{m_{t-1}^n})^2 \varphi(X_t^{m_{t-1}^n})} \\
	&\leq \sum_{n=1}^{N} \brac*{(w_{t-1}^n)^2 \varphi(X_t^n) +  (w_{t-1}^{m_{t-1}^n})^2 \varphi(X_t^{m_{t-1}^n})}
\end{align*}
where the inequality `$\leq$' holds for each individual term of the summation thanks to the coupling mechanism. The proof is concluded by taking the expectation on both sides.
\end{proof}

\begin{proof}[Proof of Theorem~\ref{thm:consistency}]
	We prove the consistency for the un-normalized statement by induction on $t$. (The normalized estimate would then be automatically consistent, being simply the ratio of two un-normalized estimates.) The case $t=0$ is straighforward.
    We first remark that
\begin{align*}
	\hat \theta_t = \frac{1}{2N} \sum_{n=1}^{2N} w_t^n \varphi(X_t^n)
	&= \frac{1}{2N} \sum_{n=1}^{N} \brac*{w_t^n \varphi(X_{t}^n) + w_t^{m_{t-1}^n} \varphi(X_t^{m_{t-1}^n})} \\
	&= \frac{1}{2N} \sum_{n=1}^{N} \brac*{w_{t-1}^n \varphi(X_{t}^n) + w_{t-1}^{m_{t-1}^n} \varphi(X_t^{m_{t-1}^n})}
\end{align*}
where the equality holds for each individual term in the summation thanks to the weight harmonization mechanism. Remark that
\begin{align*}
\Exp\brac*{\frac{1}{2N}\sum_{m=1}^{2N} w_t^m \varphi(X_t^m)   \mid \mathcal F_{t-1} }
&= \frac{1}{2N} \sum_{n=1}^{N} \brac*{w_{t-1}^n K\varphi(X_{t-1}^n) + w_{t-1}^{m_{t-1}^n} K\varphi(X_{t-1}^{m_{t-1}^n})} \\
&:= \theta_{t-1}
\end{align*}
where $\mathcal F_{t-1}$ contains all variables up to time $t-1$ included, \textit{including} the shuffling. 
Now, $\theta_{t-1} \overset{\Prob}{\rightarrow} \int (K\varphi) \mathrm d\gamma = \int \varphi \mathrm \dd\gamma$ by induction hypothesis. (Note the subtle use of Lemma~\ref{lem:jensen-kernel} to make sure that the fourth moment condition also holds for $K\varphi$ instead of $\varphi$.) 
Thus we only need to show that
\[ \hat\theta_t-\theta_{t-1}  {\rightarrow} 0 \]
in probability.
Using Chebyshev's inequality
\begin{align*}
	&\quad\Prob\left( |\hat\theta_t-\theta_{t-1}| \geq \varepsilon \mid \mathcal F_{t-1} \right)\\
	&\leq \varepsilon^{-2}\Exp\brac*{
	\left(
	\hat\theta_t-\theta_{t-1}
	 \right) ^ 2
	 \mid \mathcal F_{t-1}} \\
	 &= \varepsilon^{-2} \Var \left(
	 \frac{1}{2N} \sum_{n=1}^{N} \brac*{w_{t-1}^n \varphi(X_{t}^n) + w_{t-1}^{m_{t-1}^n} \varphi(X_t^{m_{t-1}^n})}
	 \mid \mathcal F_{t-1}\right) \\
	 &= (2\varepsilon N)^{-2} \sum_{n=1}^N 
	 \Var(w_{t-1}^n \varphi(X_{t}^n) + w_{t-1}^{m_{t-1}^n} \varphi(X_t^{m_{t-1}^n})|\mathcal F_{t-1}) \\
	 &\leq \frac 12 (\varepsilon N)^{-2} \sum_{n=1}^N \brac*{(w_{t-1}^n)^2 (K[\varphi^2])(X_{t-1}^n) + 
	 (w_{t-1}^{m_{t-1}^n})^2 (K[\varphi^2])(X_{t-1}^{m_{t-1}^n})} \\
	 &= \frac{1}{2\varepsilon^2N^2} \sum_{m=1}^{2N} (w_{t-1}^m)^2 (K[\varphi^2])(X_{t-1}^m).
\end{align*}
Now taking expectation of both sides and applying Lemma~\ref{lem:jensen-coupling} repeatedly, we have
\begin{align}
	\Prob\left( |\hat\theta_t-\theta_{t-1}|\geq \varepsilon \right)
	\label{eq:conv_proba}
	 \leq \frac{1}{\varepsilon^2N} \Exp[(w_0^1)^2 K^t[\varphi^2](X_0^1)].
\end{align}
Now the last expectation does not depend on $N$, so to prove that~\eqref{eq:conv_proba} converges to $0$ as $N \to \infty$ it remains to check that that expectation is finite. We have
\[ 
\Exp \brac{(w_0^1)^2 K^t(\varphi^2)(X_0^1)} = Z\Exp_{\pi} \brac{w_0(X) K^t(\varphi^2)(X)} \leq Z(
	\Exp_\pi[w_0^2] \Exp_\pi\brac{\abs{K^t(\varphi^2)}^2})
^{1/2}
\]
using Cauchy--Schwarz inequality.
The first expectation is finite by assumption. 
The second is finite too given that $\Exp_\pi[\varphi^4] < \infty$ and Lemma~\ref{lem:jensen-kernel} with $\alpha = 2$ and kernel $K^t$.
\end{proof}

\section{From consistency to diagnostics}
\label{app:consistency_to_diag}
\subsection{Regular Markov kernels}
\label{app:regular-kernel}
\begin{definition}
\label{def:inf-norm}
Let $K(x, \dd y)$ be a Markov kernel from $\bbR^d$ to $\bbR^d$ such that, for any distribution $\mu$ on $\mathbb R^d$ admitting a bounded density with respect to the Lebesgue measure, $\mu K$ also admits a bounded density.
We define
\[ \norm{K}_\infty = \sup_{\mu}\frac{\norm{\mu K}_\infty}{\norm{\mu}_\infty} \]
where the supremum is taken over all densities $\mu$ in $\bbR^d$ for which we define $\norm{\mu}_\infty = \sup_{x \in \bbR^d} \mu(x)$, and the density $\mu K$ is given by $(\mu K)(x) = \int \mu(\dd y) K(y, dx)/\operatorname{Leb}(\dd x)$.
\end{definition}
We are now ready to define the notion of regular Markov kernels.
\begin{definition}
\label{def:kernel-regular}
    A Markov kernel from $\bbR^d$ to $\bbR^d$ is called \emph{regular} if $\norm{K}_{\infty} < \infty$ where $\norm{K}_{\infty}$ is defined in~\ref{def:inf-norm}.
\end{definition}

\subsection{Regularity conditions for Corollary~\ref{cor:diag_kl}}
\label{app:corollary-kl-gauss}
In this section we verify the regularity conditions of Corollary~\ref{cor:diag_kl} for the random walk Metropolis--Hastings algorithm on a Gaussian target distribution $\pi$ and Gaussian starting distributions $\mu_0$.
The arguments can be easily generalized to larger classes of distributions.
The only non-trivial point is the regularity of the Metropolis--Hastings kernel in the sense of Definition~\ref{def:kernel-regular}.
Given the current state $y$, the random walk Metropolis--Hastings algorithm proposes a new state using a proposal $q(\dd y^*|y)$ of the form
\[ q(y^*|y) = \mathcal N(y^*; y, \delta \Sigma) \]
for some covariance matrix $\Sigma$ and a scale parameter $\delta$.
The next state is accepted with probability
\[ a(y, y^*) = 1 \wedge \frac{\pi(y^*)}{\pi(y)}. \]
The average rejection rate is given by
\[
\bar r(y) = \int q(y^*|y) (1 - a(y, y^*)) \dd y^*
\]
so the full kernel reads
\[
K(y, \dd y^*) = q(\dd y^*|y) a(y,y^*) + \bar r(y) \delta_y(\dd y^*).
\]
Starting from a measure $\mu$, after one iteration the chain has distribution
\[
(\mu K)(x)
= \int \mu(y) q(x|y) a(y,x) \dd y + \bar r(x) \mu(x).
\]
Thus, using the fact that $r(x)$ and $a(y,x)$ are bounded by $1$ and the symmetric form of the proposal $q$,
\begin{equation*}
\begin{split}
(\mu K)(x) &\leq \int \mu(y) q(x|y) \dd y + \mu(x) \\
&\leq \norm{\mu}_{\infty} \int q(x|y) \dd y + \norm{\mu}_\infty = \norm{\mu}_\infty \int q(y|x) \dd y + \norm{\mu}_\infty \\
&= 2\norm{\mu}_\infty.
\end{split}
\end{equation*}
Hence, the kernel $K$ is regular with $\norm{K}_\infty \leq 2$.
\subsection{Proofs of Corollaries~\ref{cor:diag-smooth} and~\ref{cor:diag_kl}}
\begin{lemma}
\label{lem:propagation_finiteness}
Suppose that the initial distribution $\mu_0$ and the target distribution $\pi$ are such that $\dd \pi/\dd \mu_0 \leq M$ for some finite $M$. 
Then the distribution at the iteration $t$ of the MCMC algorithm satisfies $\dd \pi/\dd \mu_t \leq M$.
\end{lemma}
\begin{proof}
We have
\begin{align}
    \frac{\dd \pi}{\dd \mu_t}(x_t) 
        &= \Exp\brac*{\frac{\dd (\pi \times K^t)}{\dd (\mu_0 \times K^t)}(X_0, X_t) \mid X_t = x_t}\\
        &\leq \Exp\brac*{M \frac{\dd K^t(x_0, \cdot)}{\dd K^t(x_0, \cdot)}(X_t) \mid X_t = x_t} \leq M.
\end{align}
\end{proof}

To prove Corollaries~\ref{cor:diag-smooth} and~\ref{cor:diag_kl}, we need to verify Assumptions~\ref{ass:weighted-convergence} and~\ref{ass:unweighted-LLN} on the system of particles $X_t^{1:N}$ and weights $W_t^{1:N}$ so that Theorem~\ref{thm:weighted_as_diag} can be applied.
Using Theorem~\ref{thm:consistency}, Assumption~\ref{ass:weighted-convergence} is fulfilled if
\begin{equation}
\label{eq:ass-w-in-act}
\Exp_{\pi} \brac*{f'\paren*{
\frac{\dd \pi}{\dd \mu_t}(X)
}^4} < \infty.
\end{equation}
By the strong law of large numbers and the independence of $X_t^i$ and $X_t^j$ for $1 \leq i < j \leq N $ and $N+1 \leq i < j \leq 2N$, Assumption~\ref{ass:unweighted-LLN} is fulfilled if
\newcommand{\dpidmut}{\frac{\dd \pi}{\dd \mu_t}}
\begin{align}
\label{eq:ass-uw-in-act-1}
\Exp_{\mu_t}\brac*{\dpidmut(X) \abs*{f'\paren*{\dpidmut(X)}}} &< \infty;\\
\label{eq:ass-uw-in-act-2}
\Exp_{\mu_t}\brac*{\abs*{f\paren*{\dpidmut(X)}}} &< \infty.
\end{align}

\begin{proof}[Proof of Corollary~\ref{cor:diag-smooth}]
Recall that $\pi(x) \leq L \mu_0(x)$ for some constant $L$.
Using Lemma~\ref{lem:propagation_finiteness} we have $\pi(x) \leq L\mu_t(x)$ as well.
It is then easy to verify~\eqref{eq:ass-w-in-act}, ~\eqref{eq:ass-uw-in-act-1} and~\eqref{eq:ass-uw-in-act-2} by noting that all the three functions $f'(t)$, $tf'(t)$, and $f(t)$ are bounded on $[0, L]$.
\end{proof}

\begin{proof}[Proof of Corollary~\ref{cor:diag_kl}]
Recall that $\pi(x) \leq L_1 \mu_0(x)$ for some constant $L_1$.
Using Lemma~\ref{lem:propagation_finiteness} we have $\pi(x) \leq L_1\mu_t(x)$ as well.
It is then easy to verify~\eqref{eq:ass-uw-in-act-1} and~\eqref{eq:ass-uw-in-act-2} by noting that $f(t) = t \log t$ and $tf'(t) = t(\log t + 1)$ are both bounded on the interval $[0, L_1]$.
To verify~\eqref{eq:ass-w-in-act}, we rewrite it for our specific $\kl$ divergence:
\begin{equation}
\label{eq:fourth-moment-kl}
\Exp_\pi \brac*{\paren*{1 + \log \frac{\pi(X)}{\mu_t(X)}}^4} < \infty.
\end{equation}
Using the regularity of kernel $K$, the hypothesis $\norm{\mu_0}_\infty \leq L_2$, and the fact that $\pi(x) \leq L_1 \mu_t(x)$ we have
\[ -\log L_1 + \log \pi(x) \leq \log \mu_t(x) \leq t \log \norm{K}_\infty + \log L_2 \]
which means
\begin{equation}
    |\log \mu_t(x)| \leq |\log L_1| + |\log \pi(x)| + |t\log \norm{K}_\infty| + |\log L_2|.
\end{equation}
Combining this inequality with the assumption $\Exp_\pi\brac{(1+\log \pi(X))^4} < \infty$, we obtain~\eqref{eq:fourth-moment-kl} and conclude the proof.
\end{proof}

\section{Proof of non-asymptotic properties of the expected bound}
\label{app:proof-exp-upper-bound}
\begin{proof}[Proof of Proposition~\ref{prop:exp_upper_bound}]
Using Jensen's inequality and $\Exp\brac{w_t^m | X_t^m} = \dd\gamma/\dd \mu_t(X_t^m)$ (Corollary~\ref{coroll:unbiased}), we have
\[
\Exp\brac*{f\paren*{\frac{w_t^m}{Z}} \middle\vert X_t^m} \geq f\paren*{\frac{\dd\pi}{\dd\mu_t}(X_t^m)}.
\]
Summing over $m$ and taking expectation of both sides yield the desired result.
\end{proof}
\begin{proof}[Proof of Proposition~\ref{prop:exp_upper_bound_imprecise}]
Recall that
\[
M^{-1} \Exp\brac*{\sum w_t^m \varphi(X_t^m) \middle\vert \mathcal F_{t-1}} = M^{-1} \sum w_{t-1}^m(K\varphi)(X_{t-1}^m).
\]
Thus
\[
M^{-1} \Exp\brac*{\sum w_t^m \varphi(X_t^m) \middle\vert \mathcal F_{0}} = M^{-1} \sum w_0^m (K^t\varphi)(X_0^m).
\]
Note that $\sum w_t^m = \sum w_0^m$ which is $\mathcal F_0$-measurable, so
\[
M^{-1} \Exp\brac*{\sum MW_t^m \varphi(X_t^m) \middle\vert \mathcal F_{0}} = \sum W_0^m (K^t\varphi)(X_0^m) = \hat \pi_0^M K^t \varphi.
\]
Let $I|\mathcal F_t \sim \operatorname{Unif}\{ 1 , \ldots, M\}$, we have
\begin{equation*}
\begin{split}
\hat \pi_0^M K^t \varphi &=
\Exp\brac{MW_t^I \varphi(X_t^I)\mid \mathcal F_0} =
\Exp\brac*{
\Exp\brac{MW_t^I|X_t^I, \mathcal F_0} \varphi(X_t^I) \middle\vert \mathcal F_0
} \\
&= \int (\hat \mu_0^M K^t) (\dd x) \Exp\brac{MW_t^I|X_t^I=x, \mathcal F_0} \varphi(x) \\
&= \int \frac{\dd (\hat \mu_0^M K^t)}{\dd (\hat \pi_0^M K^t)}(x) \Exp\brac{MW_t^I|X_t^I=x, \mathcal F_0} \varphi(x) (\hat \pi_0^MK^t)(\dd x).
\end{split}
\end{equation*}
This holds for all bounded function $\varphi$, so
\[
\Exp\brac{MW_t^I|X_t^I, \mathcal F_0} = \frac{\dd (\hat \pi_0^M K^t)}{\dd (\hat \mu_0^M K^t)}(X_t^I)
\]
where
\[
\hat \mu_0^M = \frac 1M \sum_m \delta_{X_0^m}.
\]
By Jensen,
\[
\Exp\brac{f(MW_t^I)|X_t^I, \mathcal F_0} \geq f\paren*{\frac
{\dd (\hat \pi_0^M K^t)}
{\dd (\hat \mu_0^M K^t)}(X_t^I)
}.
\]
Taking expectation with respect to $\mathcal F_0$ of both sides gives
\[
\Exp\brac*{M^{-1}\sum f(MW_t^m)|\mathcal F_0} \geq D_f(\hat \pi_0^M K^t \| \hat \mu_0^M K^t).
\]
Taking expectation of both sides, using $\Exp\brac{\hat \pi_0^M}= \tilde \pi_0^M$ and $\Exp\brac{\hat \mu_0^M} = \mu_0$, and the joint convexity of $D_f$, we have
\[
\Exp\brac*{M^{-1}\sum f(MW_t^m)} \geq D_f(\tilde \pi_0^M K^t \| \mu_t).
\]
The proof is finished.
\end{proof}

To prove Theorem~\ref{thm:exp_upper_bound_precise}, we first need two lemmata.
\begin{lemma}
\label{lem:bounded_f_prime}
Under the setting of Theorem~\ref{thm:exp_upper_bound_precise}, we have
\[
\frac{L_1}{L_2} \leq \frac{\dd (\tilde \pi_0^M K^t)}{\dd \mu_t}(x) \leq {\frac{L_2}{L_1}}, \quad 
\frac{L_1}{L_2} \leq \frac{\dd \pi}{\dd \mu_t}(x) \leq \frac{L_2}{L_1}.
\]
\end{lemma}
\begin{proof}
Since $\gamma(x)/\mu_0(x) \in [L_1, L_2]$ and $\pi(x) = \gamma(x)/Z$, we have $L_1 \leq Z \leq L_2$ and $L_1/L_2\leq \dd\pi/\dd\mu_0 \leq L_2/L_1$.
Using $\norm{\dd (\nu_2 K^t)/\dd (\nu_1 K^t)}_\infty \leq \norm{\dd \nu_2/\dd \nu_1}_\infty$ for any two measures $\nu_1$ and $\nu_2$, we have $\dd \pi/\dd \mu_t \leq L_2/L_1$.
Moreover, for nonnegative functions $\varphi$,
\begin{align}
\tilde \pi_0^M \varphi &=
\Exp\brac*{\sum W_0^m \varphi(X_0^m)} =
\Exp\brac*{\sum_m \frac{(\dd \gamma/\dd \mu_0)(X_0^m)}
{\sum_{m'} (\dd \gamma/\dd \mu_0)(X_0^{m'})} \varphi(X_0^m)
} \label{eq:weighted_IS} \\
&\leq \Exp\brac*{\sum \frac{L_2}{ML_1} \varphi(X_0^m)} =
{\frac{L_2}{L_1}} \mu_0(\varphi).
\end{align}
This gives $\dd \tilde \pi_0^M/\dd \mu_0 \leq L_2/L_1$. Applying the inequality $\norm{\dd (\nu_2 K^t)/\dd (\nu_1 K^t)}_\infty \leq \norm{\dd \nu_2/\dd \nu_1}_\infty$ again concludes the upper bounds of the lemma.

To prove the lower bounds, note that $\dd \pi/\dd \mu_0 \geq L_1/L_2$. Applying $\norm{\dd (\nu_2 K^t)/\dd (\nu_1 K^t)}_\infty \leq \norm{\dd \nu_2/\dd \nu_1}_\infty$ for $\nu_1 = \pi$ and $\nu_2 = \mu_0$ leads to $\dd \pi/\dd \mu_t \geq L_1/L_2$.
Similarly,~\eqref{eq:weighted_IS} gives $\dd \tilde \pi_0^M/ \dd \mu_0 \geq L_1/L_2$. Applying $\norm{\dd (\nu_2 K^t)/\dd (\nu_1 K^t)}_\infty \leq \norm{\dd \nu_2/\dd \nu_1}_\infty$ for $\nu_1 = \tilde \pi_0^M$ and $\nu_2 = \mu_0$ gives $\dd (\tilde \pi_0^M K^t)/\dd \mu_t \geq L_1/L_2$.
\end{proof}

The second lemma controls the law $\tilde \pi_0^M$ pointwise.
\begin{lemma}
\label{lem:pointwise_bias_importance_sampling}
Under the setting of Theorem~\ref{thm:exp_upper_bound_precise}, we have that $\tilde \pi_0^M$ is absolutely continuous with respect to $\pi$ and
\[
\abs*{\frac{\dd \tilde \pi_0^M}{\dd \pi}(x) - 1} \leq \frac{\kappa^2 (\kappa - 1)}{M} \quad \text{for all } x \in \mathcal X.
\]
\end{lemma}
\begin{proof}
Put $g(x) := \dd \pi / \dd \mu_0 (x)$, so that $W_0^m = g(X_0^m) / \sum_{i=1}^M g(X_0^i)$ and, as shown in the proof of Lemma~\ref{lem:bounded_f_prime}, $\kappa^{-1} \leq g \leq \kappa$.
Let $\varphi$ be a bounded measurable function.
Recall that the definition of $\tilde \pi_0^M$ is:
\[
\tilde \pi_0^M \varphi = \Exp\brac*{\sum_{m=1}^M W_0^m \varphi(X_0^m)}
\]
where $X_0^{1:M} \overset{\textrm{iid}}{\sim} \mu_0$.
By using the exchangeability of $X_0^{1:M}$, conditioning on $X_0^1$, and noting that $g = \dd \pi/\dd \mu_0$, we have
\[
\tilde \pi_0^M \varphi = M \, \Exp\brac*{\frac{g(X_0^1)\varphi(X_0^1)}{g(X_0^1) + \sum_{i=2}^M g(X_0^i)}} = \Exp\brac*{g(X_0^1) \beta_M(X_0^1) \varphi(X_0^1)} 
= \Exp_{\pi}\brac{ \beta_M(X) \varphi(X)}
\]
where
\[
\beta_M(x) := \Exp\brac*{\frac{1}{\hat Z(x)}}, \quad \hat Z(x) = M^{-1}\paren*{g(x) + \sum_{m=2}^{M} g(X_0^m)}.
\]
In particular
\[
\frac{\dd \tilde \pi_0^M}{\dd \pi  } = \beta_M.
\]
Since $E_{\mu_0}\brac{g(X)} = 1$ and $g$ is bounded, $\hat Z(x)$ tends to $1$ in probability when $M \to \infty$. Remarking the formal Taylor expansion
\begin{equation}
\label{eq:is_bias_expansion}
\frac{1}{\hat Z(x)} = 1 + [1 - \hat Z(x)] + \sum_{k=2}^\infty (1-\hat Z(x))^k,
\end{equation}
we truncate higher order terms, yielding
\[
\frac{1}{\hat Z(x)} = 1 + (1 - \hat Z(x)) + \frac{(1 - \hat Z(x))^2}{\hat Z(x)},
\]
and
\[
\beta_M(x) - 1 = \Exp\brac*{1 - \hat Z(x)} + \Exp\brac*{\frac{(1 - \hat Z(x))^2}{\hat Z(x)}}.
\]
Combining the above with the following elementary facts
\begin{align*}
\hat Z(x) &\geq \kappa^{-1},\\
\Exp[1-\hat Z(x)] = \frac{1-g(x)}{M}, &\quad
\operatorname{Var}(1-\hat Z(x)) = \frac{M-1}{M^2}
\operatorname{Var}_{\mu_0}(g),\\
|1-g(x)| \leq \kappa-1, &\quad \operatorname{Var}_{\mu_0}(g) \leq (\kappa-1)^2,
\end{align*}
we write
\begin{equation*}
\begin{split}
\abs*{\frac{\dd \tilde \pi_0^M}{\dd \pi} - 1} = \abs*{\beta_M - 1} &\leq \frac{\kappa-1}{M} + {\kappa} \brac*{\paren*{\Exp[1-\hat Z(x)]}^2 + \Var(1-\hat Z(x))} \\
&\leq \frac{\kappa - 1}{M} + \kappa \brac*{\frac{(\kappa - 1)^2}{M^2} + \frac{M-1}{M^2}(\kappa-1)^2} \\
&\leq \frac{\kappa^2(\kappa - 1)}{M},
\end{split}
\end{equation*}
concluding the proof.
\end{proof}
\begin{remark}
Considering higher-order expansions in~\eqref{eq:is_bias_expansion} shows that $\beta_M(x) - 1 = \curly{\mu_0(g^2) - g(x)} / M + \mathcal{O}(M^{-2})$ uniformly in $x$, so that $\tilde \pi_0^M \varphi - \pi \varphi = \curly{\pi(\varphi) \mu_0(g^2) - \pi(g \varphi)} / M + \mathcal{O}(M^{-2})$.
This is the classical first-order bias of self-normalized importance sampling, here uniformly in $x$ thanks to the two-sided bound on the weights.
\end{remark}

\begin{proof}[Proof of Theorem~\ref{thm:exp_upper_bound_precise}]
By Lemma~\ref{lem:bounded_f_prime}, both $\dd (\tilde \pi_0^M K^t) / \dd \mu_t$ and $\dd \pi / \dd \mu_t$ take values in $[\kappa^{-1}, \kappa]$,
so
\begin{equation}
\label{eq:f-div-diff-l1}
\begin{split}
\abs*{D_f(\tilde \pi_0^M K^t \| \mu_t) - D_f(\pi \| \mu_t)} &= \abs*{\int \mu_t(\dd x) \brac*{f\paren*{\frac{\dd (\tilde \pi_0^M K^t)} {\dd \mu_t}(x)} - f\paren*{\frac{\dd \pi}{\dd \mu_t}(x)}}} \\
&\leq \sup_{[\kappa^{-1}, \kappa]} |f'| \int \mu_t(\dd x) \abs*{\frac{\dd (\tilde \pi_0^M K^t)} {\dd \mu_t}(x) - \frac{\dd \pi}{\dd \mu_t}(x)} \\
&= 2\sup_{[\kappa^{-1}, \kappa]} |f'| \norm{\tilde \pi_0^M K^t - \pi}_\tv.
\end{split}
\end{equation}
Recall that Assumption~\ref{asp:geometric_ergodicity} implies
\[
\norm{\delta_x K^t - \pi}_\tv \leq \frac{C}{2}\rho^t V(x).
\]
A simple application of Jensen would give $\norm{\tilde \pi_0^M K^t - \pi}_\tv \leq 0.5C\rho^t \tilde \pi_0^M(V)$ which does not tend to zero as $M \to \infty$.
Instead, writing $P_\pi$ for the probability kernel $P_\pi(x, B) = \pi(B)$ satisfying $\nu P_\pi = \pi, \forall \nu$, we have
\[
\norm{\delta_x (K^t - P_\pi)}_\tv \leq \frac{C}{2}\rho^t V(x),
\]
entailing via Jensen
\[
\begin{split}
\norm{\tilde \pi_0^M K^t - \pi}_\tv &= \norm{(\tilde \pi_0^M - \pi)(K^t - P_\pi)}_\tv \leq \frac C 2 \rho^t \int V(x) |\tilde \pi_0^M - \pi| (\dd x) \\
&= \frac C 2\rho^t \int V(x) \abs*{
\frac{\dd \tilde \pi_0^M }{\dd \pi}(x) - 1
}\pi(\dd x) \leq \frac C 2 \frac{\rho^t}{M} \kappa^2(\kappa-1) \pi(V).
\end{split}
\]
We thus have
\[
D_f(\pi \| \mu_t) \leq D_f(\tilde \pi_0^M K^t \| \mu_t) + C \pi(V) \frac{\rho^t}{M} \kappa^2 (\kappa - 1) \sup_{[\kappa^{-1}, \kappa]} |f'|,
\]
and 
the total variation case follows by using $|f(x) - f(y)| \leq |x - y| / 2$ in~\eqref{eq:f-div-diff-l1}.
\end{proof}

\section{Proof of Proposition~\ref{prop:decreasing-bound} (Non-increasing bound)}
Consider the asymptotic $f$-divergence upper bound at time $t$:
\begin{equation}
    \frac{1}{2N}\sum_{n=1}^{2N} f(2N W_t^n).
\end{equation}
The weight update only affects pairs of particles that couple. 
We first analyse the effect on a single pair $(n, m^n = A_t^n+N)$. 
The weights before the update are $W_t^n$ and $W_t^{m^n}$.
\begin{itemize}
    \item If the particles do not couple, $W_{t+1}^n = W_t^n$ and $W_{t+1}^{m^n} = W_t^{m^n}$. 
    The contribution to the upper bound, $f(2N W_t^n) + f(2N W_t^{m^n})$, is unchanged.
    \item If the particles do couple, the new weights become $W_{t+1}^n = W_{t+1}^{m^n} = (W_t^n + W_t^{m^n}) / 2$. 
    By convexity of $f$, using Jensen's inequality, the contribution to the upper bound is \emph{decreased}
    \begin{align*}
        f(2N W_{t+1}^n) + f(2N W_{t+1}^{m^n}) \leq f(2N W_t^n) + f(2N W_t^{m^n}).
    \end{align*}
\end{itemize}
In either case, for any realization of the algorithm and for any pair, the contribution to the $f$-divergence bound is non-increasing.
Summing over all pairs, we have for any realization:
\begin{align*}
    \frac{1}{2N}\sum_{n=1}^{2N} f(2N W_{t+1}^n) \leq \frac{1}{2N} \sum_{n=1}^{2N} f(2N W_{t}^n).
\end{align*}

\section{Proof of Theorem~\ref{thm:as-convergence} (Bound decreasing to zero)}\label{app:proof-as-convergence}
The proof is split into 3 main lemmata which we informally describe below.
Lemma~\ref{lem:proj} ascertains that the harmonization operation is a projection onto diagonal subspaces for the weights.
Lemma~\ref{lem:conv} proves that a sequence of orthogonal projections taken from a finite set, and applied an infinite amount of times, is the projection on the intersection of the subspaces.
Lemma~\ref{lem:ergodicity} is used to show that the harmonization operation happens infinitely often on the correct subspaces, linking Lemma~\ref{lem:proj} to Lemma~\ref{lem:conv}.
Finally, Lemma~\ref{lem:finite-bernoulli-sum} is a technical lemma only used to prove Lemma~\ref{lem:ergodicity}.

\begin{lemma}
	\label{lem:proj}
	For $i \neq j$ the matrix $H^{ij}$ of size $2N \times 2N$ defined by
	\begin{itemize}
		\item $H^{ij}_{m,m}=1, \forall m \notin \{i, j\}$,
		\item $H^{ij}_{ij} = H^{ij}_{ji} = H^{ij}_{ii} = H^{ij}_{jj} = 1/2$,
		\item all other elements are zero,
	\end{itemize}
    that is, 
    \begin{equation}
        H^{ij} = 
        \begin{pmatrix}
            1 & \cdots & 0 & \cdots & 0 & \cdots & 0 \\
            \vdots & \ddots & \vdots & & \vdots & & \vdots \\
            0 & \cdots & 1/2 & \cdots & 1/2 & \cdots & 0 \\
            \vdots & & \vdots & \ddots & \vdots & & \vdots \\
            0 & \cdots & 1/2 & \cdots & 1/2 & \cdots & 0 \\
            \vdots & & \vdots & & \vdots & \ddots & \vdots \\
            0 & \cdots & 0 & \cdots & 0 & \cdots & 1
        \end{pmatrix}
    \end{equation}
    is an orthogonal projection onto the hyperplane $L_{ij} = \{x \in \bbR^{2N} | x_i = x_j \}$.
\end{lemma}

\begin{lemma}\label{lem:conv}
	Let $x \in \bbR^d$ and let $\mathcal L$ be a finite family of subspaces of $\bbR^d$. 
    Let $L_1, L_2, \ldots$ be a sequence of elements in $\mathcal L$ such that, for all $n$
	\[ 
    \cap_{i=n}^\infty L_i = L^*.
	\]
    and denote $\proj_{A}$ the orthogonal projection on a subspace $A$.
	Then
	\[
	\proj_{L_n} \circ \cdots \circ \proj_{L_1}(x) \to \proj_{L^*}(x)
	\]
	as $n\to\infty$.
\end{lemma}

\begin{lemma}
	\label{lem:ergodicity}
    If two particles $X_t^m$ and $X_t^{m'}$ are paired after shuffling, i.e. $m' = A_t^m + N$, define $B_t(m') := m$. 
    For $1 \leq i \leq N$ and $N+1 \leq j \leq 2N$, let
    \begin{align*}
    	\tau(t, i, j) &:= \inf\{s \geq 0 \mid j = N + A_{t+s}^i\}, \\
        N_{ij} &:= \sum_t \mathds 1\curly*{\tau(t, i, j) \leq 1}.
    \end{align*}
    Then, under Assumptions~\ref{asp:ergodicity},~\ref{asp:compact_set}, and~\ref{asp:finite_coupling_time}, $N_{ij} = \infty$ almost surely, leading to an infinite number of $t$ at which the two weights $W_t^i$ and $W_t^j$ are harmonized.
\end{lemma}
\begin{remark}
    In plain English, each term of $N_{ij}$ checks whether the two pairs respectively containing $X_t^i$ and $X_t^j$ are either paired right now or will be paired at the very next iteration.%
\end{remark}

\begin{lemma}\label{lem:finite-bernoulli-sum}
Let $(\mathcal{G}_t)$ be a filtration and $(Z_t)$ be a sequence of Bernoulli variables such that $Z_t$ is $\mathcal{G}_t$-measurable. 
If there exists $\rho > 0$ such that
\[
\liminf_{s \to \infty} \Prob(Z_s = 1 \mid \mathcal{G}_t) \geq \rho
\]
almost surely, then $\sum_{t=0}^{\infty} Z_t$ is infinite almost surely.
\end{lemma}

We defer the proof of the lemmata to the end of this section, and now provide the proof of Theorem~\ref{thm:as-convergence}.

\begin{proof}[Proof of Theorem~\ref{thm:as-convergence}]
	Define the vector $W^{1:2N}_t = [W_t^1 \ldots W_t^{2N}]^\top$ and the matrix $H^{ij}$ as in Lemma~\ref{lem:proj}.
	Let $C_t = \{ n | X_t^n = X_t^{A_{t-1}^n+N}\}$ the set of indices for coupled chains.
	Then
	\[
	W^{1:2N}_t = \prod_{n \in C_t} H_{n, A_{t-1}^n+N} \vec W_{t-1}
\]
where the notation $\prod$ is legitimate as the considered matrices commute.
Additionally since the coupling times are finite a.s., for each $A_t^n$ and $n$ we will necessarily multiply our weight vector by $H_{n, A_t^n+N}$ at some point in the future.
By Lemma~\ref{lem:ergodicity}, all the projections $H_{i,j}$ for $i \leq N$ and $N+1 \leq j \leq 2N$ will be involved an infinite number of times.
Finally we have
\[
	\bigcap_{\substack{1\leq i \leq N\\ N + 1 \leq j \leq 2N}} L_{ij} = \{x | x_1 = \ldots = x_{2N}\}
\]
where $L_{ij}$ are defined in Lemma~\ref{lem:proj}.
Applying Lemma~\ref{lem:conv} we have \[W^{1:2N}_t \to \curly*{1/(2N) \ldots 1/(2N)}^{\top}\] almost surely.
The weights being bounded, this almost-sure convergence entails the convergence in expectation.
By the continuity of $f$, we have $\sum_{m=1}^{2N} f(2NW_t^m)/2N$ converges to $0$ almost surely.
Using the dominated convergence theorem, we see that it also converges in expectation if $f$ is continuous at $0$ (and thus is bounded on $[0, 2N]$).
\end{proof}

\begin{proof}[Proof of Lemma~\ref{lem:proj}]
	Clearly $H^{ij}$ is symmetric, and direct calculation shows that $H^{ij} H^{ij} = H^{ij}$.
    Now, solving $H^{ij} x = 0$ gives the solution $x_m = 0$ for all $m \notin \{i, j\}$, and $x_i=-x_j$, the orthogonal of which is $L_{ij}$.
\end{proof}

\begin{proof}[Proof of Lemma~\ref{lem:conv}]

Define $a^* = \proj_{L^*}(x)$ and let $s_n = \proj_{L_n} \circ \cdots \circ \proj_{L_1}(x)$.

\textbf{Step 1.} Show that, for all $v \in L_n$, we have
\begin{equation}
\label{eq:decrease_global}
\norm{s_n - v}^2 = \norm{s_{n-1}-v}^2 - d(s_{n-1}, L_n)^2,
\end{equation}
where $d(\cdot, A)$ denotes the distance to a closed set $A$ for the Euclidean norm.
This comes from the identity
\[
\norm{w}^2 = \norm{\proj_{L_n}(w)}^2 + d(w, L_n)^2
\]
and we replace $w = s_{n-1}-v$ and remark that $\proj_{L_n}(s_{n-1}-v) = s_n - v$ as well as
\[
d(s_{n-1}-v, L_n) = d(s_{n-1}, L_n + v) = d(s_{n-1}, L_n).
\]

\textbf{Step 2.} Let $a$ be an accumulation point of $(s_n)$. We prove by contradiction that $a \in L^*$.
Suppose $a \notin L^*$. By definition of accumulation point there exists an increasing sequence $(k_n)$ such that $s_{k_n} \to a$.
Additionally $a \notin L^*$ hence there are an infinity of $n$ such that $a \notin L_n$.
We can therefore define
\begin{equation}
\label{eq:def_tn}
t_n = \min(m > k_n | a \notin L_m ) < \infty
\end{equation}
and note that there is an infinite number of such $t_n$'s because $k_n \to \infty$.
By construction $a \in \cap_{j=k_{n} + 1}^{t_n - 1} L_j$, so, using the identity~\eqref{eq:decrease_global} recursively (with the edge case $t_{n} = k_n + 1$ being an equality), we have
\begin{equation}
	\label{eq:decrease_local}
	\norm{s_{t_n - 1} - a} \leq \norm{s_{t_n - 2} - a} \leq \cdots \leq \norm{s_{k_n} - a}.
\end{equation}
Moreover $\mathcal L$ is finite so let 
\begin{equation}
\label{eq:def_epsilon}
\varepsilon = \min_{L \in \mathcal L,  a \notin L} d(a, L) > 0.
\end{equation}

\textbf{Step 2 (contradiction).} We now apply the identity~\eqref{eq:decrease_global} again
\[
\norm{s_{t_n} - a^*}^2 = \norm{s_{t_n - 1} - a^*}^2 - d(s_{t_n-1}, L_{t_n})^2.
\]
For $n$ large enough such that $\norm{s_{k_n} - a} \leq \varepsilon/2$ we have
\begin{align*}
	d(s_{t_n-1}, L_{t_n}) &\geq d(a, L_{t_n}) - \norm{s_{t_{n}-1} - a} \quad \text{by the triangle inequality}\\
	&\geq \varepsilon - \norm{s_{k_n} - a} \quad \text{by~\eqref{eq:def_tn}, ~\eqref{eq:decrease_local}, and~\eqref{eq:def_epsilon}} \\
	&\geq \varepsilon - \varepsilon/2 = \varepsilon/2.
\end{align*}
This means that $\norm{s_{n} - a^*} \geq 0$ makes an infinite number of decreases of magnitude at least $\varepsilon^2/4$.
This is impossible, so we must have $a \in L^*$.

\textbf{Step 3.} 
The remainder of the proof identifies $a^*$ by means of a fixed point argument.
We have
\[
\proj_{L^*}(s_{n+1}) = \proj_{L^*} \circ \proj_{L_{n+1}} (s_n) =
\proj_{L^*}(s_n).
\]
So $\proj_{L^*}(s_n)$ is constant with respect to $n$ and is equal to $a^*$.
Since $a$ is an accumulation point of $s_n$ we must have $\proj_{L^*}(a) = a^*$ as well.
Yet $a \in L^*$ so we have $a=a^*$.
\end{proof}

\begin{proof}[Proof of Lemma~\ref{lem:ergodicity}]
    Let
    \[
    Z_t:= \mathds 1\curly*{\tau(t, i,j) \leq 1}.
    \]
    Note that $Z_t$ is $\mathcal F_{t+1}$-measurable.

    We aim to apply Lemma~\ref{lem:finite-bernoulli-sum}. Fix an arbitrary deterministic time $T$ and consider $t > T + 1$.
    Consider the indicators
\begin{align*}
Z_t^0 &:= \mathds 1\curly{j = N+A_t^i},\\
Z_t^1 &:= \mathds 1\curly{j \neq N + A_t^i; \bar K(X_t^i, X_t^{N+A_t^i}, \cdot, \cdot) \textrm{ and } \bar K(X_t^{B_t(j)}, X_t^j, \cdot, \cdot) \textrm{ both couple}},\\
Z_t^2 &:= \mathds 1\curly{Z_t^1} \mathds 1\curly{j = N + A_{t+1}^i}.
\end{align*}
Here, $Z_t^0=1$ indicates that the indices $i$ and $j$ are already paired at time $t$.
If that is not the case, $Z_t^1=1$ indicates that the relevant pairs couple at time $t+1$; and, if that happens, $Z_t^2=1$ is the event that the shuffling at $t+1$ pairs $i$ and $j$.
It is straightforward to check that $Z_t^0Z_t^1 = 0$ and $\curly{Z_t^2 = 1} \in \curly{Z_t^1 = 1}$.
Moreover
\begin{equation}
\label{eq:zt-decompose}
Z_t \geq Z_t^0 + Z_t^2.
\end{equation}
    First, consider the shuffling term $Z_t^2$. 
    The shuffling $A_{t+1}$ is uniform and independent of the particle positions and coupling outcomes, given the set of coupled indices.
    If $Z_t^1=1$, the pairs have coupled at $t+1$, so they are available to be paired.
    Thus, because there are at most $N$ particles that can be paired to $j$,
    \begin{equation*}
    \Prob\paren*{Z_t^2=1 \mid Z_t^1=1, \sigma(\mathcal F_t, \mathcal F_{t+1}^X) } \geq \frac{1}{N},
    \end{equation*}
    where $\mathcal F_{t+1}^X$ is the sigma-algebra generated by the particles up to $t+1$.
    The law of total expectation then gives
    \begin{equation}
    \label{eq:zt2zt1}
    \Prob(Z_t^2=1 \mid Z_t^1=1, \mathcal F_t) \geq \frac{1}{N}.
    \end{equation}
    
    Now consider $Z_t^1$. 
    We want a lower bound for the probability that the two pairs in its definition couple in one step.
    To do so, first choose a compact set $\setC$ large enough such that
    \begin{equation}
    \label{eq:setC_def}
    \pi(\setC) > 1 - \frac{1}{4N}
    \end{equation}
    where $\pi$ is the invariant measure.
    By Assumption~\ref{asp:compact_set}, we have
    \begin{equation}
    \label{eq:setC_prop}
    \nu = \inf_{(x, y) \in \setC\times\setC} \Prob(X' = Y' \mid x, y) > 0, 
    \end{equation}
    that is, the probability that a pair whose both elements lie in $\setC$ couples in one iteration is lower-bounded by some strictly positive $\nu$.
    Now let $\bar X_t = (X_t^i, X_t^{N+A_t^i}, X_t^{B_t(j)}, X_t^j)$.
    By~\eqref{eq:setC_prop}, if $\bar X_t \in \setC^4$, then the probability of coupling at the next step is bounded below by $\nu^2$ (since we need two independent couplings).
    So 
    \begin{equation}
    \label{eq:zt1}
    \Prob(Z_t^1=1 \mid \mathcal F_t) \geq \nu^2 \mathds{1}(\bar X_t \in \setC^4) \mathds 1(Z_t^0 = 0).
    \end{equation}
    Combining~\eqref{eq:zt2zt1} and~\eqref{eq:zt1} and noting that $Z_t^2 \leq Z_t^1$ give
    \[
    \Prob(Z_t^2=1|\mathcal F_t) \geq \frac{\nu^2}{N} \mathds{1}(\bar X_t \in \setC^4) \mathds 1(Z_t^0 = 0).
    \]
    Applying~\eqref{eq:zt-decompose}, noting that $\nu \leq 1$ and $t > T$, and using the law of total expectation yield
    \[
    \Prob(Z_t=1 \mid \mathcal F_T) \geq \frac{\nu^2}{N} \Prob(\bar X_t \in \setC^4 \mid \mathcal F_T).
    \]
    We will now show that $\liminf_{t \to \infty} \Prob(\bar X_t \in \setC^4 \mid \mathcal F_T) \geq \frac 12$, in fact we shall show the stronger statement
    $\liminf_{t \to \infty} \Prob( X_t^{1:2N} \in \setC^{2N} \mid \mathcal F_T) \geq \frac 12$.
    Write
\begin{equation}
\begin{split}
\Prob(X_t^{1:2N} \in \setC^{2N} \mid \mathcal F_T) &\geq 1 - \sum_{m=1}^{2N} \Prob(X_t^m \notin \setC \mid \mathcal F_T) \\
&\geq 1 - \sum_{m=1}^{2N} \brac*{\pi(\mathcal X \setminus \setC) + \norm{\delta_{X_T^m} K^{t-T} - \pi}_{\tv}} \\
&\geq \frac 12 - \sum_{m=1}^{2N} \norm{\delta_{X_T^m} K^{t-T} - \pi}_{\tv}
\end{split}
\end{equation}
    where the first inequality is a rewrite of the union bound; the second inequality follows from the definition of the total variation distance: $\Exp_{\mu}\{f(X)\}\leq \Exp_{\pi}\{f(X)\} + \norm{\mu - \pi}_{\tv}$ for $\abs{f} \leq 1$; and the last inequality follows from~\eqref{eq:setC_def}.
    By Assumption~\ref{asp:ergodicity} on the ergodicity
    \[\liminf_{t \to \infty} \Prob( X_t^{1:2N} \in \setC^{2N} \mid \mathcal F_T) \geq \frac 12.\]
    Consequently,
    \[
    \liminf_{t \to \infty} \Prob(Z_t=1 \mid \mathcal F_T) \geq \frac{\nu^2}{2N} > 0 \quad \text{a.s.}
    \]
    Lemma~\ref{lem:finite-bernoulli-sum} then implies that $\sum_{t=0}^\infty Z_t = \infty$ almost surely.
    For each $t$ such that $j = N+A_t^i$, Assumption~\ref{asp:finite_coupling_time} implies that there exists a finite $s$ where $W_{t+s}^i = W_{t+s}^j = (W_t^i + W_t^j)/2$.
    Since $\sum_{t=0}^\infty Z_t = \infty$, such harmonization events must take place an infinite number of times.
\end{proof}

\begin{proof}[Proof of Lemma~\ref{lem:finite-bernoulli-sum}] 
Let $E$ be the event that the sum is finite:
\[
E = \left\{ \sum_{t=0}^{\infty} Z_t < \infty \right\}.
\]
We aim to show that $\Prob(E) = 0$.

    On the set $E$, there are only finitely many successes, meaning $Z_s(\omega) \to 0$ as $s \to \infty$. 
    On the set $E^c$, the indicator function $\mathds{1}_E$ is zero.
    Therefore, the sequence of random variables $Z_s \mathds{1}_E$ converges to $0$ pointwise everywhere as $s \to \infty$.

    Fix an arbitrary time $t$. 
    Since $|Z_s \mathds{1}_E| \leq 1$, we have:
    \[
    \lim_{s \to \infty} \Exp[Z_s \mathds{1}_E \mid \mathcal{G}_t] = 0,
    \]
    almost surely.

    Consider the conditional probability provided in the hypothesis. We decompose the expectation over $E$ and $E^c$:
    \[
    \Prob(Z_s = 1 \mid \mathcal{G}_t) = \Exp\paren{Z_s \mid \mathcal{G}_t} = \Exp\paren{Z_s \mathds{1}_E \mid \mathcal{G}_t} + \Exp\paren{Z_s \mathds{1}_{E^c} \mid \mathcal{G}_t}.
    \]

    Taking the limit inferior as $s \to \infty$ on both sides:
    \[
    \liminf_{s \to \infty} \Prob(Z_s = 1 \mid \mathcal{G}_t) = 0 + \liminf_{s \to \infty} \Exp\brac{Z_s \mathds{1}(E^c) \mid \mathcal{G}_t}.
    \]
    Using the hypothesis $\liminf_{s \to \infty} \Prob(Z_s = 1 \mid \mathcal{G}_t) \geq \rho$ almost surely, and noting that $Z_s \mathds{1}_{E^c} \leq \mathds{1}_{E^c}$:
    \[
    \rho \leq \liminf_{s \to \infty} \Exp\brac*{Z_s \mathds{1}(E^c) \mid \mathcal{G}_t} \leq \Exp\brac*{\mathds{1}(E^c) \mid \mathcal{G}_t} = 1 - \Prob(E \mid \mathcal{G}_t).
    \]
    Rearranging this inequality yields:
    \[
    \Prob(E \mid \mathcal{G}_t) \leq 1 - \rho,
    \]
    almost surely.

    The event $E$ is measurable with respect to $\mathcal{G}_{\infty} = \sigma(\bigcup_t \mathcal{G}_t)$. 
    Note that $E \in \mathcal{G}_{\infty}$ because $E = \bigcup_{M = 1}^{\infty} \bigcap_{t=0}^{\infty} \curly*{Z_{0} + \ldots + Z_{t} \leq M}$.
    By Lévy's upward theorem, the conditional probability $\Prob(E \mid \mathcal{G}_t)$ converges to the indicator variable $\mathds{1}_E$ almost surely as $t \to \infty$.
    Taking the limit as $t \to \infty$ in our derived inequality:
    \[
    \mathds{1}_E \leq 1 - \rho.
    \]
    If $\omega \in E$, then $1 \leq 1 - \rho$, which implies $\rho \leq 0$. 
    This contradicts the assumption that $\rho > 0$.

    Therefore, $\Prob(E) = 0$, implying $\sum_{t=0}^{\infty} Z_t$ is infinite almost surely.
\end{proof}

\section{Proof of Theorem~\ref{thm:ergodic-weight-convergence} (Exponential convergence)}\label{app:proof-ergodic}
The core idea is to show that the expected squared Euclidean distance to the uniform weight vector $\bar{W}$ is a Lyapunov function for the process, contracting every two steps. 
Let $V_t = \norm{W_t - \bar{W}}_2^2 = \sum_{i=1}^{2N} \{W_t^i - 1/(2N)\}^2$. 
Note that $V_t=0$ if and only if $W_t = \bar{W}$.

We start by stating two easily verified identities:
for any collection of real numbers $(q_n)_{n=1}^N$,
\begin{equation}\label{eq:first}
    \sum_{n=1}^N (q_n - \bar{q})^2 = \frac{1}{2N} \sum_{m,n=1}^{N}(q_n - q_m)^2
\end{equation}
with $\bar{q} = \sum_{n=1}^N q_n / N$, and
for any $x, y, z \in \bbR$,
\begin{equation}\label{eq:second}
    (x-y)^2 + (x + y - 2z)^2 = 2 [(x-z)^2 + (y-z)^2].
\end{equation}

Now, at each step $t+1$, for each pair of indices $(n, m^n_t)$ where $m^n_t = A_t^n+N$, the particles are propagated through the coupled kernel $\bar{K}$. 
If the particles couple, which occurs with probability $p_{n,t} \ge p_c$, their corresponding weights are averaged: $W_{t+1}^n = W_{t+1}^{m^n_t} = (W_t^n + W_t^{m^n_t})/2$. If they do not couple, the weights remain unchanged.

The weight update rule implies that the sum of weights is conserved, $\sum_{i=1}^{2N} W_{t+1}^i = \sum_{i=1}^{2N} W_t^i = 1$.

We can express $V_{t+1}$ in terms of $V_t$. 
The change in the sum of squares only affects the weights of coupled pairs. 
Let $C_{t+1}$ be the random set of indices $n \in \{1, \dots, N\}$ for which the corresponding pair of particles coupled at step $t+1$. 
For a single coupled pair indexed by $n \in C_{t+1}$ (with partner $m^n_t = A_t^n+N$), the change in their contribution to $V_t$ is:
\begin{align*}
    \Delta V_n 
      &= 2\paren*{\frac{W_t^n+W_t^{m^n_t}}{2} - \frac{1}{2N}}^2 - \paren*{W_t^n - \frac{1}{2N}}^2 - \paren*{W_t^{m^n_t} - \frac{1}{2N}}^2 \\
      &= -\frac{1}{2}(W_t^n - W_t^{m^n_t})^2,
\end{align*}
as per the second identity \eqref{eq:second} stated at the beginning of this proof.
Summing over all coupled pairs, we get the total change in $V_t$:
\begin{equation*}
    V_{t+1} = V_t + \sum_{n \in C_{t+1}} \Delta V_n = V_t - \frac{1}{2}\sum_{n \in C_{t+1}} \left(W_t^n - W_t^{m^n_t}\right)^2.
\end{equation*}
Now, let $\mathcal{F}_t$ be the filtration generated by the history of particles and pairings up to time $t$. 
We take the conditional expectation with respect to the random coupling events.
\begin{align}
    \Exp[V_{t+1} \mid \mathcal{F}_t] &= \Exp\brac*{V_t - \frac{1}{2}\sum_{n=1}^N \mathds{1}(n \in C_{t+1}) (W_t^n - W_t^{m^n_t})^2 \mid \mathcal{F}_t} \\
    &= V_t - \frac{1}{2}\sum_{n=1}^N \Exp\brac*{\mathds{1}(n \in C_{t+1}) \mid \mathcal{F}_t} (W_t^n - W_t^{m^n_t})^2.\label{eq:trick}
\end{align}
By Assumption~\ref{ass:coupling}, the probability that pair $n$ couples is at least $p_c$, regardless of the current particle states $X_t^n, X_t^{m^n_t}$. 
Thus, $\Exp\brac*{\mathds{1}(n \in C_{t+1}) \mid \mathcal{F}_t} \ge p_c$. This gives us:
\begin{equation*}
    \Exp[V_{t+1} \mid \mathcal{F}_t] \le V_t - \frac{p_c}{2}\sum_{n=1}^N (W_t^n - W_t^{m^n_t})^2.
\end{equation*}
The term $\sum_{n=1}^N (W_t^n - W_t^{m^n_t})^2$ represents the potential for variance reduction for a given pairing $A_t$. 

Applying the same argument to $\Exp\brac*{V_{t} \mid \mathcal{F}_{t-1}} \leq V_{t-1} - p_c \sum_{n=1}^N (W_{t-1}^n - W_{t-1}^{m^n_{t-1}})^2 / 2$, we can then write, using the tower law of expectations, that, conditionally on all random variables generated up to time $t-1$,
\begin{align*}
    \Exp\brac{V_{t+1} \mid \mathcal{F}_{t-1}}
        \le V_{t-1} &- \frac{p_c}{2}\sum_{n=1}^N (W_{t-1}^n - W_{t-1}^{m^n_{t-1}})^2 - \frac{p_c}{2}\Exp\brac*{\sum_{n=1}^N (W_t^n - W_t^{m^n_t})^2 \mid \mathcal{F}_{t-1}}.
\end{align*}
Now, it is straightforward to show that, because $(A_{t}^n)_{n=1}^N$ is a uniform permutation (or derangement) of $(A_{t-1}^n)_{n=1}^N$ 
on $C_t$,
\begin{align*}
    \Exp\brac*{\sum_{n=1}^N (W_t^n - W_t^{m^n_t})^2 \mid \mathcal{F}_{t-1}} \geq \Exp\brac*{\sum_{i,j \in C_{t}} \frac{1}{\abs{C_t}}(W_t^i - W_t^{m^j_{t-1}})^2 \mid \mathcal{F}_{t-1}} 
\end{align*}
where, by convention $0 / 0 = 0$ and empty sums are null so that the inequality is trivial if no coupling happens.
As a consequence, using the definition of $W_{t}^n$ as an average of $W_{t-1}^n$ and $W_{t-1}^{m^{n}_{t-1}}$ as well as using \eqref{eq:first} stated at the beginning of this proof, the above inequality implies
\begin{align}
    \Exp\brac*{\sum_{n=1}^N (W_t^n - W_t^{m^n_t})^2 \mid \mathcal{F}_{t-1}}
        &\geq \frac{1}{4} \Exp\brac*{\sum_{i, j \in C_t} \frac{1}{\abs{C_t}} \paren*{W_{t-1}^{i} + W_{t-1}^{m^i_{t-1}} - W_{t-1}^{j} - W_{t-1}^{m^j_{t-1}}}^2 \mid \mathcal{F}_{t-1}}\\
        &\geq \frac{1}{4 N} p_c^2 \sum_{i, j=1}^N \paren*{W_{t-1}^{i} + W_{t-1}^{m^i_{t-1}} - W_{t-1}^{j} - W_{t-1}^{m^j_{t-1}}}^2
\end{align}
where we have used the fact that $\abs{C_t} \leq N$, and the same inequality of the sum on $C_t$ as in~\eqref{eq:trick}.
Finally, using \eqref{eq:first}
\begin{align}
    \Exp\brac*{\sum_{n=1}^N (W_t^n - W_t^{m^n_t})^2 \mid \mathcal{F}_{t-1}}&\geq \frac{p_c^2}{2} \sum_{n=1}^N \paren*{W_{t-1}^{n} + W_{t-1}^{m^n_{t-1}} - \frac{1}{N}}^2
\end{align}
and combining everything, noting that $p_c / 2 > p_c^3 / 4$,
\begin{align*}
    \Exp\brac*{V_{t+1} \mid \mathcal{F}_{t-1}} 
        &\leq V_{t-1} 
        - \frac{p_c}{2}\sum_{n=1}^N (W_{t-1}^n - W_{t-1}^{m^n_{t-1}})^2 - \frac{p_c}{2}\Exp\brac*{\sum_{n=1}^N (W_t^n - W_t^{m^n_t})^2 \mid \mathcal{F}_{t-1}} \\
        &\leq V_{t-1} 
        - \frac{p_c}{2}\sum_{n=1}^N (W_{t-1}^n - W_{t-1}^{m^n_{t-1}})^2 - \frac{p_c^3}{4}\sum_{n=1}^N \paren*{W_{t-1}^{n} + W_{t-1}^{m^n_{t-1}} - \frac{1}{N}}^2\\      
        &\leq V_{t-1} 
        - \frac{p_c^3}{4}\sum_{n=1}^N \brac*{\paren*{W_{t-1}^n - W_{t-1}^{m^n_{t-1}}}^2 + \paren*{W_{t-1}^{n} + W_{t-1}^{m^n_{t-1}} - \frac{1}{N}}^2}\\   
        &= V_{t-1} 
        - \frac{p_c^3}{2}\sum_{n=1}^N \brac*{\paren*{W_{t-1}^n - \frac{1}{2N}}^2 + \paren*{W_{t-1}^{m^n_{t-1}} - \frac{1}{2N}}^2}\\ 
        &= \paren*{1 - \frac{p_c^3}{2}} V_{t-1}.
\end{align*}
where the penultimate equality comes from~\eqref{eq:second} and the last one from the definition of $V_{t-1}$.

The result then follows from the tower law, starting the recursion either at $\Exp(V_1)$ for odd $t$'s or $\Exp(V_0)$ for even ones.

\section{Couplings}\label{app:couplings}
\subsection{Reflection maximal coupling}\label{subsec:reflection-maximal}
When $p(x) \sim \mathcal{N}(\mu_1, \Sigma)$ and $q(x) \sim \mathcal{N}(\mu_2, \Sigma)$ are two multivariate Gaussian densities with the same covariance matrix, it is possible to form a maximal coupling of $X \sim p$, $Y \sim q$, that is, a coupling such that $\Prob(X = Y)$ is maximized as per Algorithm~\ref{alg:reflection-maximal}.
The reflection maximal version (see e.g. \citealt{lindvall1986lectures}) furthermore incorporates a reflection, which is known to minimize the expected meeting time of Ornstein--Uhlenbeck trajectories, as those in Section~\ref{subsec:ornstein-uhlenbeck}.

\FloatBarrier
\begin{fakealgoenv}
\label{alg:reflection-maximal}
Reflection-maximal coupling for $\mathcal{N}(\mu_1, L L^\top)$ and $\mathcal{N}(\mu_2, L L^\top)$
\begin{tabbing}
\qquad \enspace \textbf{Input:} \= Means $\mu_1, \mu_2$ and matrix $L$. \\
\qquad \enspace \textbf{Output:} A coupled pair $(x, y)$. \\
\\
\qquad \enspace Set $z \leftarrow L^{-1} (\mu_1 - \mu_2)$ and $e \leftarrow z / \norm{z}$. \\
\qquad \enspace Sample $V \sim \mathcal{N}(0, I)$ and $U \sim \Uniform([0, 1])$. \hspace{2em} \textit{\% Independently} \\
\qquad \enspace \textbf{If} $\mathcal{N}(V; 0, I) \, U < \mathcal{N}(V + z; 0, I)$ \textbf{then} \\
\qquad \qquad \enspace Set $W \leftarrow V + z$. \\
\qquad \enspace \textbf{Else} \\
\qquad \qquad \enspace Set $W \leftarrow V - 2 \left\langle e, V \right\rangle e$. \\
\qquad \enspace \textbf{End If} \\
\qquad \enspace Set $x \leftarrow \mu_1 + L V$ and $y \leftarrow \mu_2 + L W$.
\end{tabbing}
\end{fakealgoenv}

\subsection{Reflection maximal coupling for random walk Metropolis--Hastings}
Given the current state, the proposal distribution for the random-walk Metropolis--Hastings algorithm is Gaussian: $q(y \mid x) = \mathcal{N}(y; x, \delta \Sigma)$, and 
the proposed state is accepted with probability $\alpha(y, x) = 1 \wedge \{\gamma(y) / \gamma(x)\}$.
Different coupling systems have been proposed for this proposal~\citep{wang2021coupling,Papp2024couplings}.
Here we take a slightly suboptimal but practical approach which consists in using the reflection maximal coupling of the proposals and then using a common uniform variate to accept or reject the proposals.
This is summarized in Algorithm~\ref{alg:reflection-maximal-rw}.
\FloatBarrier
\begin{fakealgoenv}\label{alg:reflection-maximal-rw}
Coupling for the random walk kernel with pre-conditioning covariance $L L^{\top} = \Sigma$ and target $\pi \propto \gamma$.
\begin{tabbing}
\qquad \enspace \textbf{Input:} \= Positions $x$, $y$, step-size $\delta$\\
\qquad \enspace \textbf{Output:} A coupled pair $(x, y)$. \\
\\
\qquad \enspace Compute $\mu_x \gets x$ \\
\qquad \enspace Compute $\mu_y \gets y$\\
\qquad \enspace Sample $x', y'$ using Algorithm~\ref{alg:reflection-maximal} with $(\mu_x, \mu_y, \sqrt \delta L)$\\
\qquad \enspace Sample a uniform $U \sim \Uniform([0, 1])$\\
\qquad \enspace \textbf{If} $U < \alpha(x', x)$\\
\qquad \qquad \enspace Set $x \gets x'$\\
\qquad \enspace \textbf{If} $U < \alpha(y', y)$\\
\qquad \qquad \enspace Set $y \gets y'$
\end{tabbing}
\end{fakealgoenv}

The resulting coupling is successful if the proposals are coupled and both are accepted.

\subsection{Reflection maximal coupling of Metropolis-adjusted Langevin algorithm}\label{subsec:mala-couplings}
Given the current state, the proposal distribution for MALA is Gaussian: $q(y \mid x) = \mathcal{N}(y; x + \delta/2\, \Sigma \nabla \log \gamma(x), \delta \Sigma)$, and 
the proposed state is accepted with probability $\alpha(y, x) = 1 \wedge \{\gamma(y) q(x \mid y)\} / \{\gamma(x) q(y \mid x)\}$.
Here we again take the suboptimal approach of using the reflection maximal coupling of the proposals and a common uniform variate to accept proposals.
This is summarized in Algorithm~\ref{alg:reflection-maximal-mala}.
\FloatBarrier
\begin{fakealgoenv}\label{alg:reflection-maximal-mala}
Coupling for the MALA kernel with pre-conditioning covariance $L L^{\top} = \Sigma$ and target $\pi \propto \gamma$.
\begin{tabbing}
\qquad \enspace \textbf{Input:} \= Positions $x$, $y$, step-size $\delta$\\
\qquad \enspace \textbf{Output:} A coupled pair $(x, y)$. \\
\\
\qquad \enspace Compute $\mu_x \gets x + \delta/2\, \Sigma \nabla \log \gamma(x)$ \\
\qquad \enspace Compute $\mu_y \gets y + \delta/2\, \Sigma \nabla \log \gamma(y)$\\
\qquad \enspace Sample $x', y'$ using Algorithm~\ref{alg:reflection-maximal} with $(\mu_x, \mu_y, \sqrt \delta L)$\\
\qquad \enspace Sample a uniform $U \sim \Uniform([0, 1])$\\
\qquad \enspace \textbf{If} $U < \alpha(x', x)$\\
\qquad \qquad \enspace Set $x \gets x'$\\
\qquad \enspace \textbf{If} $U < \alpha(y', y)$\\
\qquad \qquad \enspace Set $y \gets y'$
\end{tabbing}
\end{fakealgoenv}
The resulting coupling is again successful if the proposals are coupled and both are accepted.

\subsection{P\'olya--Gamma sampler}\label{app:pgg}
The P\'olya--Gamma sampler was introduced in \citet{Polson01122013} as a data augmentation Gibbs scheme for the logistic regression problem as described in Section~\ref{subsec:polya-gamma}.
In its uncoupled version, given the current state $\beta \in \mathbb{R}^d$ of the parameter chain, it alternates between
\begin{enumerate}
\itemsep-1em 
    \item $W_{i} \sim \PG(1, \abs{x_i \cdot \beta})$, $i=1, \ldots, n$\\
    \item $\beta \sim \mathcal{N}\curly*{\mu(W), \Sigma(W)}$
\end{enumerate}
where $X \in \mathbb{R}^{n \times d}$ is the design matrix of $n$ observations and $d$ covariates, $\tilde{y}_i = y_i/2 \in \{-1/2, 1/2\}$ is the scaled outcome, $W=(W_1,\ldots,W_n)^\top$ is the vector of Pólya--Gamma latent variables, $\mu(W) = \Sigma(W)(X^\top \tilde{y} + B^{-1}b)$, $\Sigma(W) = [X^{\top} \mathrm{diag}(W) X + B^{-1}]^{-1}$, and $\mathcal{N}(b, B)$ is the Gaussian prior for $\beta$, with $\beta \in \mathbb R^d$ and $B \in \mathbb R^{d \times d}$.

Couplings for such chains are readily available owing to the P\'olya--Gamma distribution having tractable likelihood ratios $\PG(x; 1, c)/\PG(x; 1, c')$.
In Algorithm~\ref{alg:coupling-pgg}, we describe such a coupling, which corresponds to~\citep[Algorithm 7]{biswas2019estimating}. Let the log of the un-normalized P\'olya--Gamma density kernel be $h(z, \omega) = \log\cosh(z/2) - \frac{1}{2}z^2\omega$.

\FloatBarrier
\begin{fakealgoenv}\label{alg:coupling-pgg}
Maximal coupling of P\'olya--Gamma variables.
\begin{tabbing}
    \qquad \enspace \textbf{Input:} Distributions $\PG(1, c_1)$, $\PG(1, c_2)$. \\
    \qquad \enspace \textbf{Output:} Coupled samples $(\omega_1, \omega_2)$ from the Pólya--Gamma distributions.\\    
    \\
    \qquad \enspace Draw $\omega_1 \sim \text{PG}(1, c_1)$ \\
    \qquad \enspace Draw $U_1 \sim \text{Uniform}(0,1)$ \\
    \qquad \enspace \textbf{If} $\log(U_1) \le h(c_2, \omega_1) - h(c_1, \omega_1)$: \\
    \qquad \enspace \qquad Set $\omega_2 = \omega_1$ \\
    \qquad \enspace \textbf{Else}: \\
    \qquad \qquad \enspace \textbf{Repeat}: \\
    \qquad \qquad \qquad \enspace Draw $\omega_2' \sim \text{PG}(1, c_2)$ \\
    \qquad \qquad \qquad \enspace Draw $U_2 \sim \text{Uniform}(0,1)$ \\
    \qquad \qquad \enspace \textbf{Until} $\log(U_2) > h(c_1, \omega_2') - h(c_2, \omega_2')$ \\
    \qquad \qquad \enspace Set $\omega_2 = \omega_2'$
\end{tabbing}
\end{fakealgoenv}
The full coupled sampler is then given by Algorithm~\ref{alg:coupling-pgg-full}.

\FloatBarrier
\begin{fakealgoenv}\label{alg:coupling-pgg-full}
P\'olya--Gamma Gibbs Coupling Step.
\begin{tabbing}
    \qquad \enspace \textbf{Input:} Current states $\beta, \hat{\beta} \in \mathbb{R}^d$. \\
    \qquad \enspace \textbf{Output:} Next states $\beta', \hat{\beta}' \in \mathbb{R}^d$. \\
    \\
    \qquad \enspace For $i=1, \dots, n$: \\
    \qquad \qquad Sample $(W_{i}, \hat{W}_{i})$ from Algorithm~\ref{alg:coupling-pgg} for $\text{PG}(1, \abs{x_i^\top \beta})$ and $\text{PG}(1, \abs{x_i^\top \hat{\beta}})$. \\
    \qquad \enspace Sample $(\beta', \hat{\beta}')$ from a maximal coupling of $\mathcal{N}(\mu(W), \Sigma(W))$ and $\mathcal{N}(\mu(\hat{W}), \Sigma(\hat{W}))$.
\end{tabbing}
\end{fakealgoenv}

We refer to \citet[Section S4 and Algorithm 7, respectively]{jacob2020unbiasedmcmc,biswas2019estimating} for details.

\subsection{Other samplers and couplings}
A large collection of couplings have been proposed for different MCMC kernels over the past decade or so.
While we do not expect this list to be exhaustive, we have collected here several of them which seem to cover most of the techniques employed.

\begin{enumerate}
    \item Hamiltonian Monte Carlo~\citep[HMC][]{neal:2011} MCMC is perhaps one of the most used MCMC algorithm. 
    It is worth noting that no exact coupling of the method has been proposed to date. 
    In \citet{heng2019couplings}, the authors instead implement a two-scale coupling of a mixture of HMC and of a random-walk. 
    The HMC component is then made to contract in space using techniques akin to the reflection of Section~\ref{subsec:reflection-maximal} while the random walk is used to eventually try and make the chains stick once they are close enough.
    This is a case of a two-scale coupling, where an algorithm is made to behave differently when the states are far apart versus when they are close.
    \item Piecewise deterministic Markov processes~\citep{bierkens2019zig,bierkens2020boomerang,bouchard2018bps} are classes of continuous-time Markov chains on a joint position-velocity state-space with tractable dynamics between velocity jumps, which may be coupled too~\citep{corenflos2025pdmp}.
    The underlying technique relies mostly on \emph{clock synchronization}: managing the time and type of the jump events to control the dynamics of the process.
    These couplings are however less efficient than for discrete-time dynamics owing to their largely more rigid structure.
    \item Optimal transport~\citep[see, e.g.,][]{villani2009optimal,cuturi2019compOT} contraction is commonly used in order to enforce geometric (in the squared Euclidean distance sense) proximity, either for a subpart of the sampler~\citep[in the context of a Gibbs scheme for instance]{nguyen2022many}, or as part of a two-scale method~\citep{biswas2022coupling,ceriani2024linearcostunbiasedposteriorestimates}, where the two chains are brought closer together, and then the optimal transport coupling is swapped for one that may make the chains meet exactly.
    \item Gibbs(-like) couplings in general can be designed or studied component-wise or in a collapsed manner~\citep{ceriani2024linearcostunbiasedposteriorestimates}, sometimes reaching different conclusions.
    A class of such samplers which has attracted a lot of research over the past few years is conditional sequential Monte Carlo methods~\citep{jacob2020unbiasedsmoothing,lee2020coupled,karjalainen2023mixing}, which, in some sense, extend Metropolis-within-Gibbs methods to systems with tractable Markovian dependencies.
    \item Pseudo-marginal MCMC~\citep{andrieu2009pseudomarginal} algorithms have also been the subject of research and their couplings have been studied in~\citep{middleton2019unbiased,middleton2020unbiased}. 
    The general idea is often to design two different couplings, one for the latent state used in acceptance, and the other for the proposal distribution used for the parameter of interest.
\end{enumerate}
Other uses of couplings we did not mention in this article include multilevel Monte Carlo~\citep[see, e.g.,][]{giles2015mlmc,rhee2015unbiased,vihola2018mlmc} or coupling from past~\citep[see, e.g.,][]{propp1996exact,huber2016perfect}.

\section{Other empirical outputs}\label{app:other-empirical-res}

\subsection{Gaussian additional statistics}
We now report other computed $f$-divergences based on our weight-harmonization procedure.
In Figure~\ref{fig:tv-gaussian}, the total variation is computed by itself owing to no closed form solutions being available in general, in Figure~\ref{fig:kl-gaussian} we report the Kullback--Leibler divergence profile, while in Figure~\ref{fig:hellinger-gaussian}, we report the squared Hellinger distance profile.

\begin{figure}
    \centering
    \includegraphics[width=\linewidth]{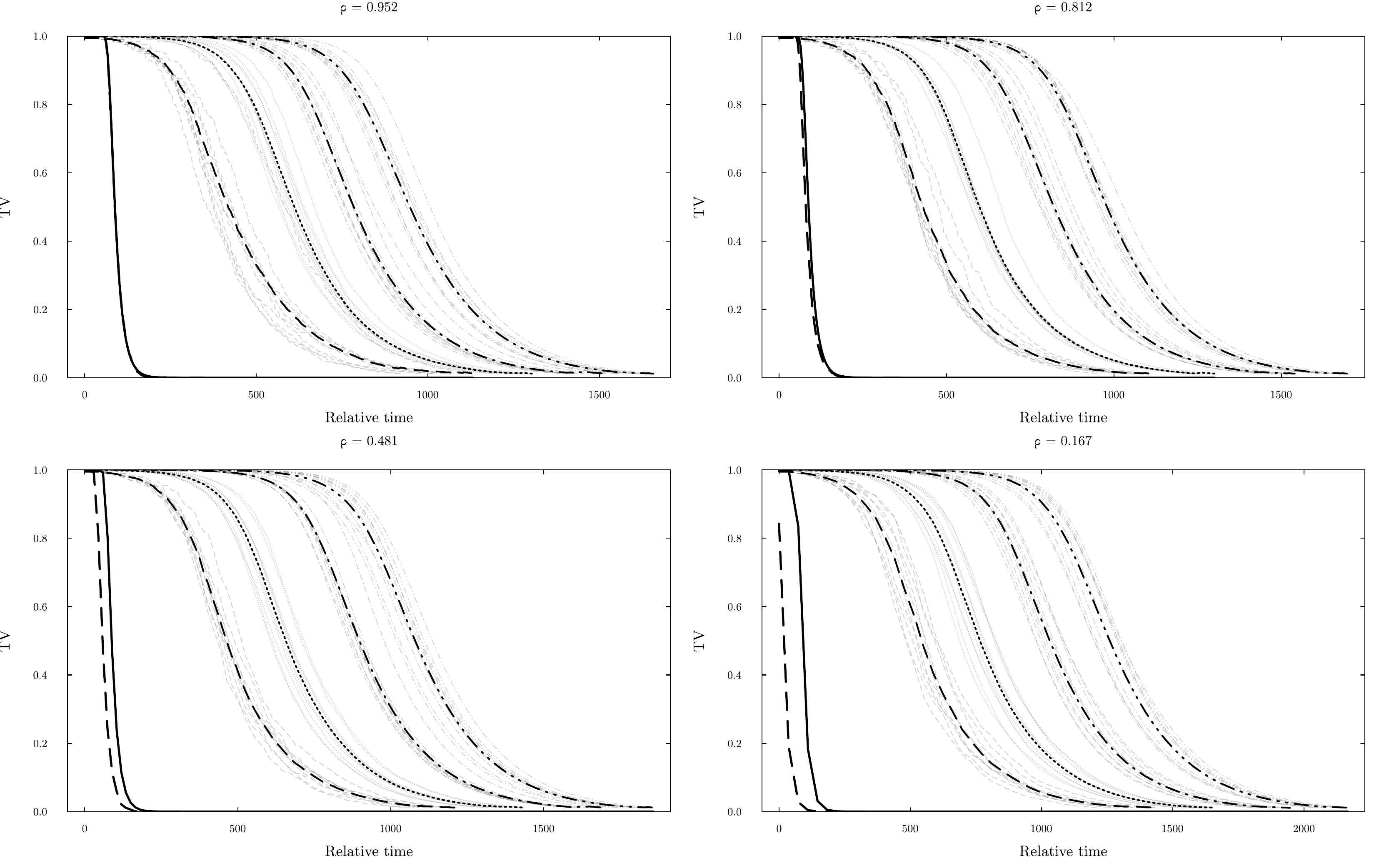}
    \caption[]{Total variation for the Ornstein--Uhlenbeck Gaussian example. {\protect\tikz[baseline=-0.5ex]\draw[black, line width=2.5pt] (0,0)--(0.5cm,0);}~CLT closed-form; {\protect\tikz[baseline=-0.5ex]\draw[black, line width=2.5pt, dashed] (0,0)--(0.5cm,0);} \citet{biswas2019estimating} L-lag estimate; harmonization (dashed gray, by $N$): {\protect\tikz[baseline=-0.5ex]\draw[gray, line width=1.0pt, dashed] (0,0)--(0.5cm,0);} $N=100$; {\protect\tikz[baseline=-0.5ex]\draw[gray, line width=1.0pt, dotted] (0,0)--(0.5cm,0);} $N=1{,}000$; {\protect\tikz[baseline=-0.5ex]\draw[gray, line width=1.0pt, dashdotted] (0,0)--(0.5cm,0);} $N=10{,}000$; {\protect\tikz[baseline=-0.5ex]\draw[gray, line width=1.0pt, dash pattern=on 5pt off 2pt on 1pt off 2pt on 1pt off 2pt] (0,0)--(0.5cm,0);} $N=100{,}000$. Thin transparent lines show individual seed realizations; thick lines show the mean.}
    \label{fig:tv-gaussian}
\end{figure}

\begin{figure}
    \centering
    \includegraphics[width=\linewidth]{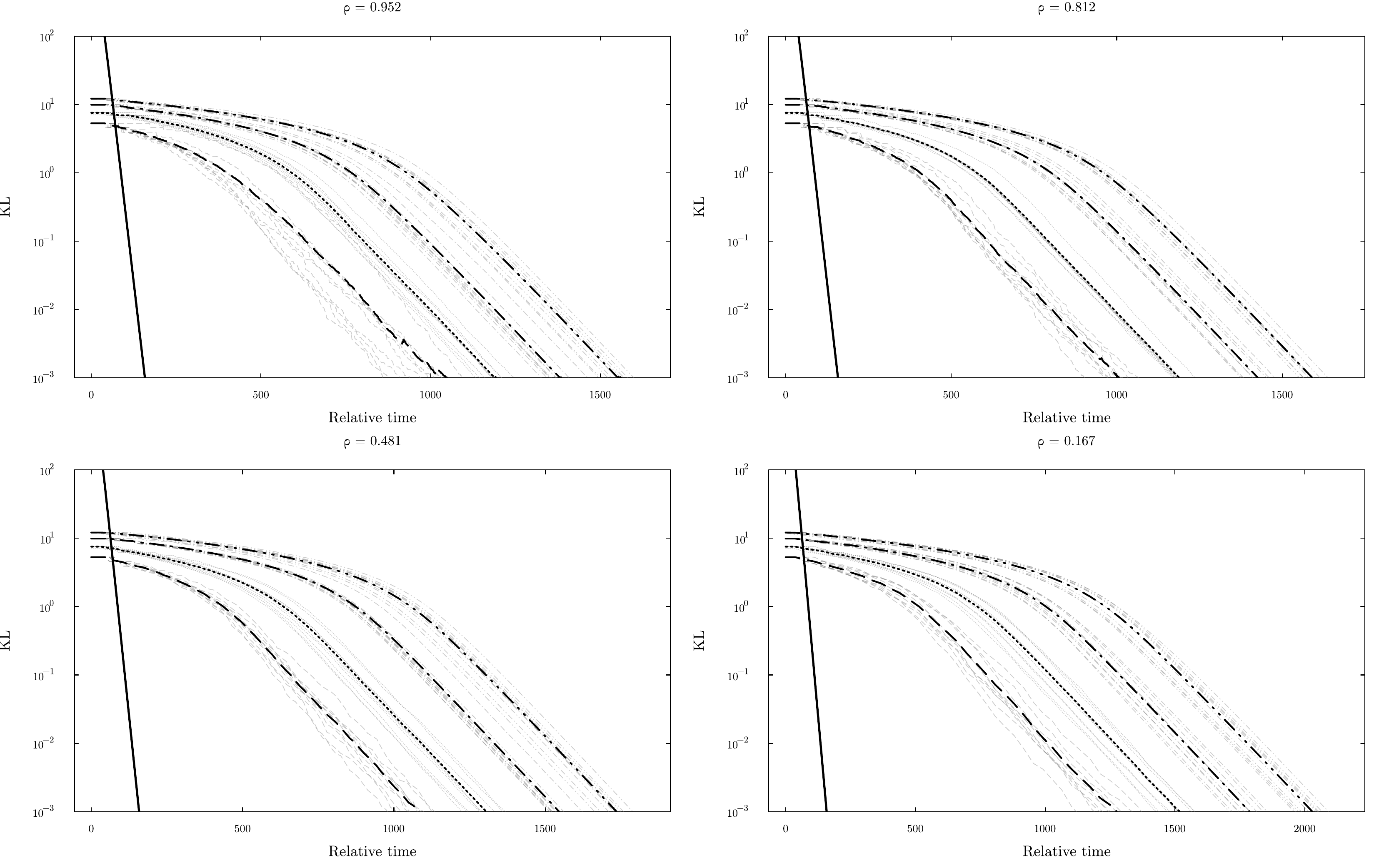}
    \caption[]{KL divergence $\mathrm{KL}(\pi\|\mu_t)$ (log scale) for the Ornstein--Uhlenbeck Gaussian example. {\protect\tikz[baseline=-0.5ex]\draw[black, line width=2.5pt] (0,0)--(0.5cm,0);}~closed-form theoretical value; empirical: {\protect\tikz[baseline=-0.5ex]\draw[gray, line width=1.0pt, dashed] (0,0)--(0.5cm,0);} $N=100$; {\protect\tikz[baseline=-0.5ex]\draw[gray, line width=1.0pt, dotted] (0,0)--(0.5cm,0);} $N=1{,}000$; {\protect\tikz[baseline=-0.5ex]\draw[gray, line width=1.0pt, dashdotted] (0,0)--(0.5cm,0);} $N=10{,}000$; {\protect\tikz[baseline=-0.5ex]\draw[gray, line width=1.0pt, dash pattern=on 5pt off 2pt on 1pt off 2pt on 1pt off 2pt] (0,0)--(0.5cm,0);} $N=100{,}000$. Thin transparent lines show individual seed realizations; thick lines show the mean.}
    \label{fig:kl-gaussian}
\end{figure}

\begin{figure}
    \centering
    \includegraphics[width=\linewidth]{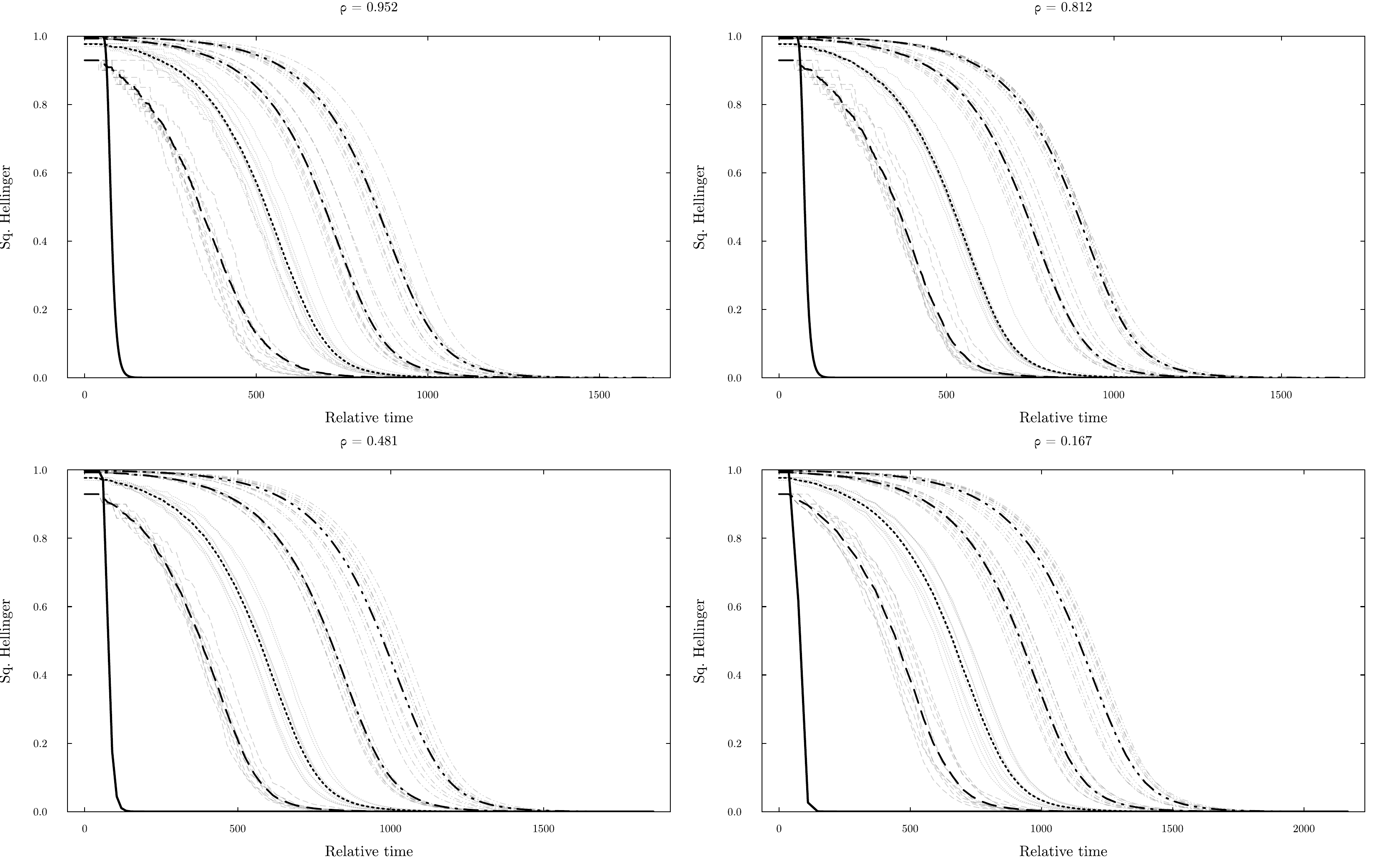}
    \caption[]{Squared Hellinger distance $H^2(\pi,\mu_t)$ for the Ornstein--Uhlenbeck Gaussian example. {\protect\tikz[baseline=-0.5ex]\draw[black, line width=2.5pt] (0,0)--(0.5cm,0);}~closed-form theoretical value; empirical: {\protect\tikz[baseline=-0.5ex]\draw[gray, line width=1.0pt, dashed] (0,0)--(0.5cm,0);} $N=100$; {\protect\tikz[baseline=-0.5ex]\draw[gray, line width=1.0pt, dotted] (0,0)--(0.5cm,0);} $N=1{,}000$; {\protect\tikz[baseline=-0.5ex]\draw[gray, line width=1.0pt, dashdotted] (0,0)--(0.5cm,0);} $N=10{,}000$; {\protect\tikz[baseline=-0.5ex]\draw[gray, line width=1.0pt, dash pattern=on 5pt off 2pt on 1pt off 2pt on 1pt off 2pt] (0,0)--(0.5cm,0);} $N=100{,}000$. Thin transparent lines show individual seed realizations; thick lines show the mean.}
    \label{fig:hellinger-gaussian}
\end{figure}

\subsection{P\'olya--Gamma sampler additional statistics}
\label{app:polya-gamma-additional}

In Figure~\ref{fig:all-stats-credit}, we report additional statistics for several $f$-divergences: the total variation distance, Kullback--Leibler divergence, the reversed Kullback--Leibler divergence, and the $\chi^2$ and squared Hellinger distances.
\begin{figure}
    \centering
    \includegraphics[width=0.5\linewidth]{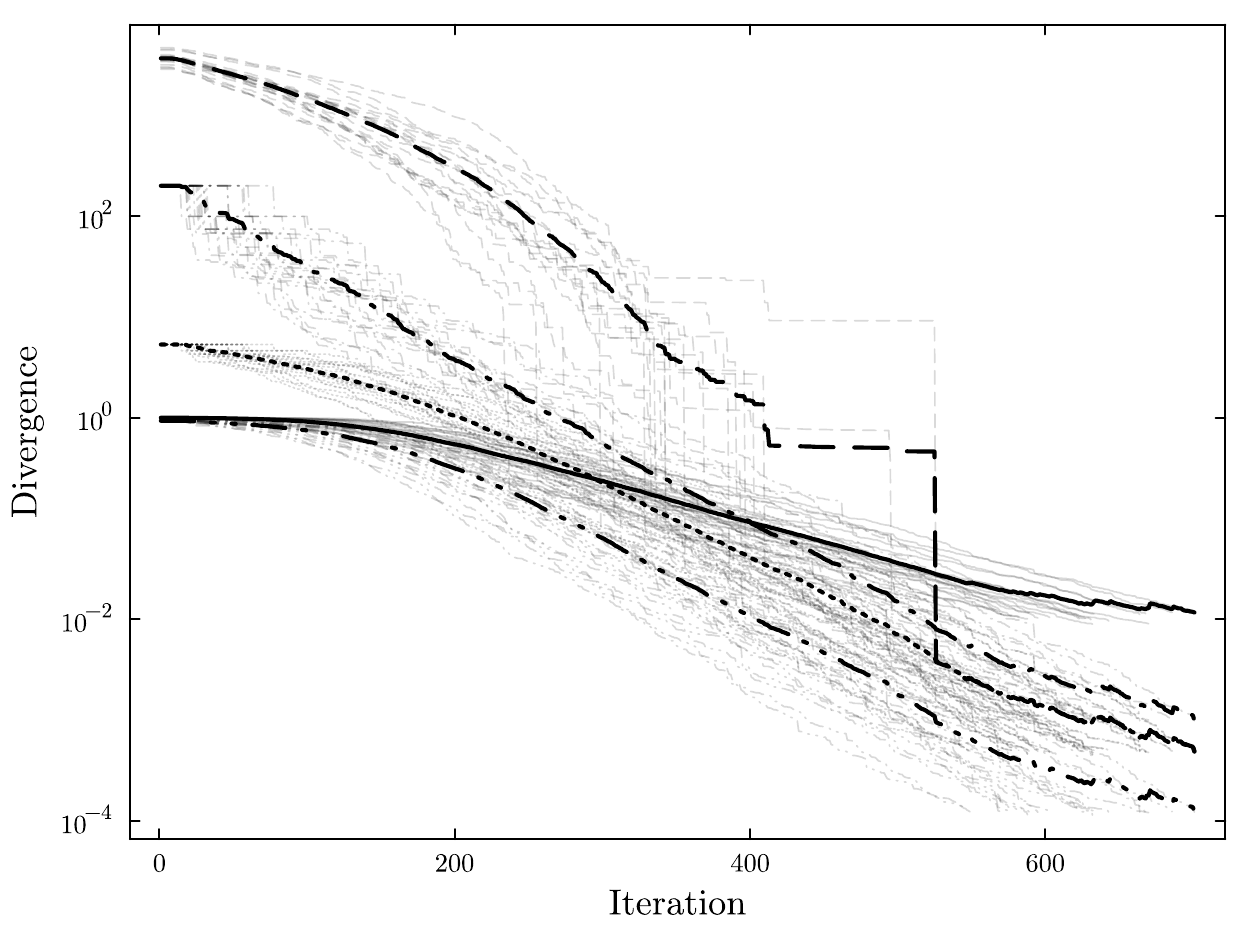}
    \caption[]{All five $f$-divergence asymptotic upper bounds from weight harmonization ($N=100$, 20 seeds) for the P\'olya--Gamma Gibbs sampler. Mean curves (thick) with individual seed realizations (transparent, same style): {\protect\tikz[baseline=-0.5ex]\draw[black, line width=2.0pt] (0,0)--(0.5cm,0);}~total variation; {\protect\tikz[baseline=-0.5ex]\draw[black, line width=2.0pt, dashed] (0,0)--(0.5cm,0);}~reverse KL; {\protect\tikz[baseline=-0.5ex]\draw[black, line width=2.0pt, dotted] (0,0)--(0.5cm,0);}~KL; {\protect\tikz[baseline=-0.5ex]\draw[black, line width=2.0pt, dashdotted] (0,0)--(0.5cm,0);}~$\chi^2$; {\protect\tikz[baseline=-0.5ex]\draw[black, line width=2.0pt, dash pattern=on 5pt off 2pt on 1pt off 2pt on 1pt off 2pt] (0,0)--(0.5cm,0);}~squared Hellinger.}
    \label{fig:all-stats-credit}
\end{figure}

\end{document}